\documentclass[12pt,twoside,a4paper]{memoir}

\usepackage[T1]{fontenc}
\usepackage{eucal}
\usepackage{microtype}
\usepackage{csquotes}
\usepackage{amsthm}
\usepackage{mathtools}
\usepackage{amssymb,amsfonts,stmaryrd}

\usepackage{tikz}
\usetikzlibrary{cd}
\usetikzlibrary{calc,patterns,angles,quotes,positioning,fit,backgrounds,automata}

\usepackage[citestyle=numeric-comp,bibstyle=ieee,sorting=nyt,sortcites,style=alphabetic,maxbibnames=10,firstinits=true]{biblatex} 
\bibstyle{apsrev}
\usepackage{enumitem}
\setlist[enumerate,1]{label={\roman*)}}
\usepackage[unicode=true, pdfusetitle, colorlinks=true, allcolors=blue, bookmarks=true]{hyperref}
\usepackage[left=3cm,right=3cm,top=3cm,bottom=3cm]{geometry}
\usepackage{pdfpages}
\usepackage{afterpage}

\newcommand\blankpage{%
    \null
    \thispagestyle{empty}%
    \newpage
    }

\usepackage{stmaryrd}
\usepackage{libertine}
\usepackage[libertine,cmintegrals,cmbraces,vvarbb]{newtxmath}
\usepackage[scaled=0.95]{zi4}
\AtBeginDocument{%
  \let\mathbb\relax
  \DeclareMathAlphabet{\mathbb}{U}{msb}{m}{n}%
}

\usepackage{silence}
\theoremstyle{plain}
\newtheorem{theorem}{Theorem}
\newtheorem*{theorem*}{Theorem}

\newtheorem*{lemma*}{Lemma}
\newtheorem{proposition}[theorem]{Proposition}
\newtheorem*{proposition*}{Proposition}
  
\newtheorem*{corollary*}{Corollary}
\theoremstyle{definition}

\newtheorem*{definition*}{Definition}
\newtheorem{examplex}{Example}
\newtheorem*{example*}{Example}
\theoremstyle{remark}
\newtheorem{remark}{Remark}
\newtheorem*{remark*}{Remark}

\newtheorem*{conjecture*}{Conjecture}

\newtheorem*{problem*}{Problem}

\newenvironment{example}
  {\pushQED{\qed}\examplex}
  {\popQED\endexamplex}

\newcommand*{\RR}{\mathbb{R}}

\newcommand*{\dd}{\mathrm{d}}
\DeclareMathOperator{\Id}{Id}
\newcommand*{\contr}[1]{\iota_{#1}}
\newcommand*{\liedv}[1]{\mathcal{L}_{#1}}

\newcommand*{\parder}[2]{\frac{\partial#1}{\partial #2}}
\newcommand*{\der}[2]{\frac{\dd #1}{\dd #2}}
\DeclareMathOperator{\Leg}{Leg}
\DeclareMathOperator{\diag}{diag}
\newcommand{\incl}{i}
\newcommand{\abs}[1]{\left| #1 \right|}

\let\norm\undefined 
\DeclarePairedDelimiter\norm{\lVert}{\rVert}

\let\emphoriginal\emph
\let\oldemph\emph
\let\emph\textbf
\let\vec\mathbf
\mathtoolsset{showonlyrefs}

\bibliography{biblio}

\AtEveryBibitem{
  \ifentrytype{online}{}{
  \ifentrytype{misc}{}{
   \ifentrytype{unpublished}{}{
    \clearfield{url}
  }}}
}

\DeclareSourcemap{
  \maps[datatype=bibtex]{
    \map[overwrite=true]{
      \step[fieldset=urldate, null]
      \step[fieldset=language, null]
      \step[fieldset=address, null]
     \step[fieldset=addendum, null]
     \step[fieldset=pagetotal, null]
     \step[fieldset=location, null]
    }
  }
}

\title{The geometry of Rayleigh dissipation\\
\Large Master's Thesis\\
\large Master in Theoretical Physics\\
Autonomous University of Madrid and\\Institute of Theoretical Physics (CSIC-UAM)}

\author{Author: Asier López-Gordón
\\Supervisor: Manuel de León (ICMAT-CSIC)
\\Tutor: Germán Sierra}

\date{\today}

\begin{document}
\pagenumbering{Alph}
\includepdf{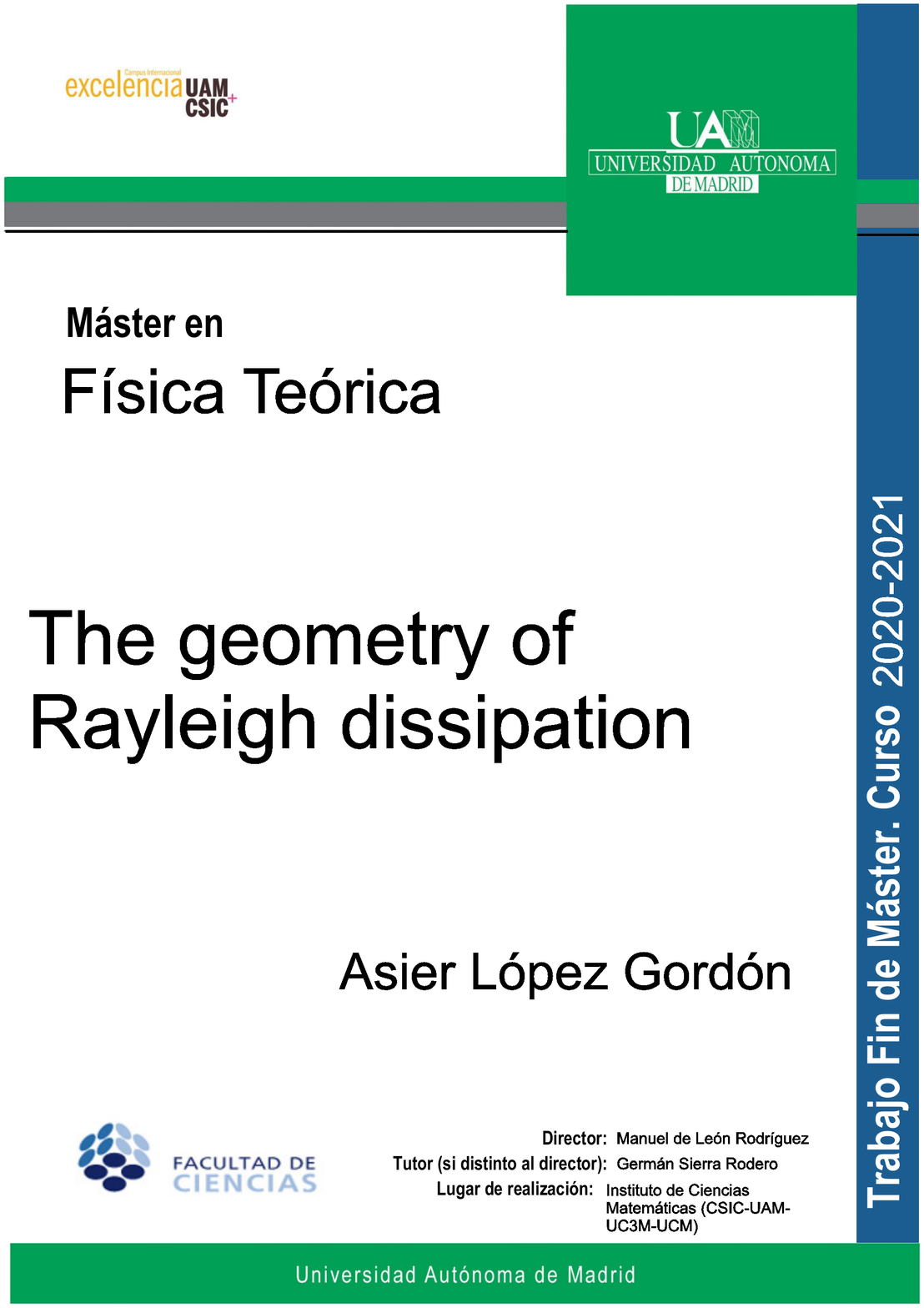}

\frontmatter

\maketitle
\thispagestyle{empty}

\afterpage{\blankpage}

\newpage
\tableofcontents
\afterpage{\blankpage}

\chapter*{Abstract}
\addcontentsline{toc}{chapter}{Abstract}

Geometric mechanics is a branch of mathematical physics that studies classical mechanics of particles and fields from the point of view of geometry. In a geometric language, symmetries can be expressed in a natural manner as vector fields that generate the corresponding symmetry group. Moreover, with the geometric machinery, the phase space of a mechanical system with symmetries can be reduced to a space with less dimensions.

As it is well-known, non-conservative forces cannot be written as the gradient of a potential, so they cannot be absorbed in the Lagrangian or Hamiltonian. The description of a mechanical system subject to a non-conservative force requires an external force together with the Lagrangian or Hamiltonian function. In addition, external forces emerge for the description of certain systems with non-holonomic constraints.

In this master's thesis, autonomous Hamiltonian and Lagrangian systems are studied in the framework of symplectic geometry. After introducing the geometric tools that will be employed, several results regarding symmetries and constants of the motion are reviewed. The method of symplectic reduction for systems with symmetry is also presented. In a second part, non-conservative systems are presented in a geometric language. A Noether's theorem for Lagrangian systems subject to external forces is obtained. Other results regarding symmetries and constants of the motion are derived as well. Furthermore, a theory for the reduction of forced Lagrangian systems invariant under the action of a Lie group is presented. These results are particularized for the so-called Rayleigh dissipation, that is, external forces that can be written as the derivative of a ``potential'' with respect to the velocities.

\chapter*{Resumen}
\addcontentsline{toc}{chapter}{Resumen}
La mecánica geométrica es la rama de la física matemática que estudia la mecánica clásica, tanto de partículas como de campos, desde el punto de vista de la geometría. En el lenguaje de la geometría, las simetrías pueden verse de manera natural como campos de vectores, los cuales son generadores del grupo de simetría correspondiente. Asimismo, con la maquinaria de la geometría, el espacio de fases de un sistema mecánico con simetrías puede ser reducido a un espacio de dimensión inferior.

Como es bien sabido, las fuerzas no conservativas no pueden escribirse como el gradiente de un potencial, de modo que no se pueden absorber en un lagrangiano o hamiltoniano. La descripción de un sistema mecánico sometido a una fuerza no conservativa requiere considerar una fuerza externa junto al lagrangiano o hamiltoniano. Además, las fuerzas externas aparecen al estudiar ciertos sistemas con ligaduras no holónomas.

En este Trabajo Fin de Máster se estudian sistemas autónomos, tanto hamiltonianos como lagrangianos, en el marco de la geometría simpléctica. Tras introducir las herramientas geométricas que se emplearán, se revisan los resultados en la bibliografía concernientes a las simetrías y constantes del movimiento. También se presenta el método de la reducción simpléctica. En una segunda parte, se presentan los sistemas no conservativos en un lenguaje geométrico. Se obtiene un teorema de Noether para sistemas lagrangianos sometidos a fuerzas externas. Se obtienen además otros resultados acerca de las simetrías y constantes del movimiento de estos sistemas. Asimismo, se presenta una teoría para la reducción de sistemas lagrangianos forzados invariantes bajo la acción de un grupo de Lie. Estos resultados se particularizan para la llamada disipación de Rayleigh, esto es, fuerzas externas que se pueden escribir como la derivada de un <<potencial>> con respecto a las velocidades.
\chapter*{Acknowledgements}
\addcontentsline{toc}{chapter}{Acknowledgements}

I would like to thank my supervisor, Manuel de León, for introducing me to the research world and to the ICMAT's Geometric Mechanics and Control Group. His careful corrections and comments on the original manuscript are deeply appreciated. 

I am grateful with Manuel Lainz for several enlightening discussions and explanations. I am also thankful with ICMAT's administration staff, specially with Esther Ruiz, for helping me deal with several bureaucratic processes. Additionally, I want to thank Germán Sierra for having been my tutor.

I am wholeheartedly thankful to my parents for having supported all my academic and personal projects. 
I want to acknowledge them for having made me a curious and hard-working person.

I am grateful for the love and advice from all my friends, who have supported me in my best and worst moments. In particular, I must thank my friends and colleagues Laura Pérez and Irene Abril for encouraging me to apply for a JAE Intro grant, which eventually has lead to the existence of this text.

This work has been supported by the CSIC's JAE Intro Grant JAEINT\_20\_01494. A renewal of the grant by ICMAT is also acknowledged.
\mainmatter
\chapter{Introduction}
\chapterprecishere{``The miracle of the appropriateness of the language of mathematics for the \mbox{formulation} of the laws of physics is a wonderful gift which we neither \mbox{understand} nor deserve.''\par\raggedleft--- \textup{E.P. Wigner, Comm. Pure Appl. Math., \textbf{13}: 1-14 (1960)}}

Since the dawn of physics, mechanics has been one of its mostly studied branches. Moreover, mechanics is the basis of astronomy, engineering and other disciplines.
Analogously, geometry is one of the most ancient and fundamental branches of mathematics. 
As it is well-known, the first mathematically precise formulation of mechanics is due to Newton. 
 Unlike in the Newtonian framework, in the Lagrangian and Hamiltonian formalisms mechanical systems can be parametrized in arbitrary systems of coordinates without the need of incorporating \emphoriginal{ad hoc} forces. 
  The next logical step, accomplished by geometric mechanics, is to formulate laws of motion in an intrinsic way, that is, without the need of using coordinates. See reference \cite{leon_historia_2017} for a historical review of geometric mechanics.
  In the language of geometry, the phase space of a system has the structure of a manifold endowed with a symplectic form, and its dynamics is given by the integral curves of a certain vector field.

Symmetries and constants of the motion are key concepts in modern mechanics. The knowledge of these properties gives a valuable information of the mechanical system to be studied, and it simplifies the resolution of its equations of motion. Their relations have been deeply explored since the seminal work by Emmy Noether \cite{noether_invariant_1971,kosmann-schwarzbach_noether_2011,neeman_impact_1999}. The geometric formulation of mechanics, together with the theory of Lie groups and Lie algebras, allows to reduce the phase space of a system with symmetry, that is, to describe the dynamics in a space with less dimensions.

The description of multiple physical systems requires considering an external force, together with the Lagrangian or the Hamiltonian, for instance, a dissipative force. 
Moreover, non-holonomic systems have external forces caused by the constraints \cite{de_leon_diego_96,iglesias-ponte_towards_2008}. These forces restrain the dynamics to the submanifold of constraints, which is the actual phase space of the system.

 Geometrically, an external force can be regarded as a semibasic 1-form. A particular type of external forces are the ones derived from the so-called Rayleigh dissipation function, that is, forces which are linear in the velocities. The concept of Rayleigh dissipation can be generalised for non-linear forces that can be written as the derivative of a ``potential'' with respect to the velocities. Symmetries and constants of the motion can also be considered for non-conservative systems (that is, systems with external forces), where they will depend on both the Lagrangian or Hamiltonian and on the external force. The method of reduction can also be generalised for non-conservative systems, requiring that the (sub)group of symmetries leaves the external force invariant as well as the Lagrangian or Hamiltonian.

This master's thesis is organized as follows:
\begin{itemize}
\item Part \ref{part_1} is a bibliographical revision of geometric mechanics for autonomous Hamiltonian and Lagrangian systems depending on the positions and the velocities.
\begin{itemize}
\item Section \ref{section_basics} is an introduction to the basic mathematical tools that will be employed throughout the text: fibre bundles, Cartan's calculus of differential forms and the actions of Lie groups and Lie algebras over manifolds.
\item Section \ref{section_symplectic} presents symplectic geometry, the geometric structure of phase spaces.
\item Section \ref{section_Hamiltonian} is devoted to the Hamiltonian formalism presented in an intrinsic form. Symmetries and constants of the motion are presented in the language of differential geometry.
\item Section \ref{section_tangent} introduces the geometric structures in tangent bundles required for formulating the Lagrangian formalism.
\item Section \ref{section_Lagrangian} presents the Lagrangian formalism in its intrinsic form. The concepts introduced in the previous section are employed for studying symmetries and constants of the motion.
\item Section \ref{section_reduction} is dedicated to symplectic reduction. If a Lie group of symmetries acts over the configuration space of a physical system, its dynamics can be described in a reduced space by means of this procedure.
\end{itemize}
\item Part \ref{part_paper} is devoted to mechanical systems with external forces. It covers the results from a paper of the author, coauthored by M. de León and M. Lainz, \cite{de_leon_symmetries_2021}.
\begin{itemize}
\item Section \ref{section_external_forces} introduces mechanical systems subject to external forces in the language of geometry, both in the Lagrangian and Hamiltonian formalisms.
\item Section \ref{section_Lagrangian_forced} studies symmetries and constants of the motion of forced Lagrangian systems, generalizing the results from Section \ref{section_Lagrangian}.
\item Section \ref{section_reduction_forced} extends symplectic reduction, explained in Section \ref{section_reduction}, for forced Lagrangian systems in which both the Lagrangian and the external force are invariant under a Lie (sub)group of symmetries.
\item Section \ref{section_Rayleigh} covers mechanical systems subject to Rayleigh external forces, particularizing the results from sections \ref{section_Lagrangian_forced} and \ref{section_reduction_forced}.
\end{itemize}
\item Part \ref{conclusions} summarizes the results obtained throughout this master's thesis, and presents several lines for future research.
\end{itemize}

\chapter{Hamiltonian and Lagrangian mechanics}\label{part_1}
\section{Geometric foundations}\label{section_basics}

This section is devoted to establish and recall the basic mathematical framework that will be employed throughout the text. The main references are \cite{abraham_foundations_2008,leon_methods_1989,godbillon_geometrie_1969}. See also references \cite{dubrovin_modern_1985,chern_lectures_1999,Lee2012} for an introduction to differential geometry of manifolds.

Recall that a manifold is roughly a space that locally resembles $\RR^n$. More specifically, a differentiable manifold\footnote{Unless otherwise stated, throughout this thesis every manifold will be assumed to be a differentiable manifold. One can also define the more general concept of topological manifold.} can be defined as follows.

A \emph{differentiable $n$-dimensional manifold} is a set $M$ with the following properties:
\begin{enumerate}
\item The set $M$ is a finite or countably infinite union of subsets $U_q$.
\item Each subset $U_q$ has defined on it local coordinates $x_q^\alpha:U_q\to \RR$, with $\alpha=1,\ldots,n$.
\item Each non-empty intersection of subsets, $U_p\cap U_q$, has defined on it at least two coordinate systems. Each of these two coordinate systems are required to be expressible in terms of the other as one-to-one differentiable functions:
\begin{equation}
\begin{aligned}
	&x_{p_{0}}^{\alpha}=x_{p}^{\alpha}\left(x_{q}^{1}, \ldots, x_{q}^{n}\right),\\
	&x_{q}^{\alpha}=x_{q}^{\alpha}\left(x_{p}^{1}, \ldots, x_{p}^{n}\right). \label{transition_functions}
\end{aligned}
\end{equation}
Then the Jacobian $\det(\partial x_p^\alpha/x_q^\alpha)$ is non-zero in the region of intersection.
\end{enumerate}

 Each pair $(U_p,\varphi_{U_p})$, with $\varphi_{U_p}=(x_p^1,\ldots,x_p^n)$, is called a \emph{coordinate chart} of $M$. A set of charts $\left\{(U_i,\varphi_{U_i})  \right\}$ which covers $M$ (that is, $\cup_i U_i=M$) is called an \emph{atlas}. The functions \eqref{transition_functions} are called \emph{transition functions}. If, for each pair of charts in an atlas, the transition function between them is of class $C^r$, then the atlas is said to be of class $C^r$. The $C^r$ differentiable structure of $M$ is given by the \emph{maximal atlas} of class $C^r$, that is, the union of every atlas of class $C^r$. It endows the set $M$ with the structure of a topological space.


Let $M$ and $N$ be two manifolds with $\dim M =m$ and $\dim N=n$. Consider a continuous mapping $f:M\to N$. Suppose that there exists coordinate charts $(U,\varphi_U)$ and $(V,\psi_V)$ at $p\in M$, and $f(p)\in N$, respectively, such that the mapping
\begin{equation}
	\psi_V \circ f \circ \varphi_U^{-1}:\varphi_U(U) \to \psi_V(V)
\end{equation}
is of class $C^\infty$ at the point $\varphi_U(p)$. Then the mapping $f$ is called $C^\infty$ at $p$. A mapping $f:M\to N$ that is $C^\infty$ at every $p\in M$ is called a \emph{differentiable mapping} (or a \emph{smooth mapping}). Moreover, if $m\geq n$, a differentiable mapping $f:M\to N$ is said to be a \emph{submersion} if $\operatorname{rank} f=n$ at every point of $M$.

A \emph{fibre bundle} is a triple $(E,\pi,M)$, where
 $E$ is a manifold called the \emph{total space}, $M$ is a manifold called the \emph{base space} and $\pi:E\to M$ is a surjective submersion called the \emph{projection} of the bundle. For each $x$ in $M$, the submanifold $\pi^{-1}(x)=E_x$ is a surjective submersion called the \emph{fibre} over $x$. Moreover, $E$ is called a \emph{fibred manifold} over $M$. A mapping $s:M\to E$ such that $\pi\circ s =\mathrm{id}$ is called a \emph{section} of $E$.

There are two types of fibre bundles which are particularly relevant in physics: vector bundles and principal bundles. In the former each fibre has the structure of a vector field, whereas in the latter each fibre has the structure of a Lie group. Throughout this thesis vector bundles, and particularly the tangent and cotangent bundle, will play the major role. However, it is also worth mentioning that principal bundles are closely related with gauge theories (see reference \cite{collinucci_topology_2006} and references therein).


Recall that, given a vector space $V$, $\mathrm{GL}(V)$ denotes the general linear group of $V$, that is, the set of all invertible linear transformations $V\to V$, together with the composition of these transformations. Consider two manifolds $E$ and $M$, and let $\left\{(U_i,\varphi_{U_i})  \right\}$ be an atlas of $M$.
A \emph{(real) vector bundle} $E$ of rank $q$ over $M$, with \emph{typical fibre} $V=\RR^q$, is a fibre bundle $(E,\pi,M)$ such that:
\begin{enumerate}
\item Every map $\varphi_{U_i}$ is a diffeomorphism from $U_i\times \RR^q$ to $\pi^{-1}(U_i)$, and
\begin{equation}
	\pi\circ \varphi_{U_i}(p,y)=p
\end{equation}
for every $p\in U$ and every $y\in \RR^q$.
\item For any fixed $p\in U$, let
\begin{equation}
	\varphi_{U_i,p}(y) = \varphi_{U_i}(p,y), \quad y\in \RR^q.
\end{equation}
Then $\varphi_{U_i,p}:\RR^q\to \pi^{-1}(p)$ is an homeomorphism (that is, it is continuous and has a continuous inverse mapping). 
\item When $U_i\cap U_j\neq \varnothing$ for $i\neq j$, then, for any $p\in U_i\cap U_j$,
\begin{equation}
	g_{U_iU_j}(p) = \varphi_{U_j,p}^{-1} \circ \varphi_{U_i,p}: \RR^q \to \RR^q
\end{equation}
is a linear automorphism of $V=\RR^q$, that is, $g_{U_iU_j}(p)\in \mathrm{GL}(V)$. Moreover, the map $g_{U_i U_j}:U_i \cap U_j\to \mathrm{GL}(V)$ is differentiable.
\end{enumerate}

Let $M$ be an $m$-dimensional differentiable manifold. A \emph{curve} at $x\in M$ is a differentiable mapping $\sigma:I\to M$, with $I\ni 0$ an open interval of $\RR$ and $\sigma(0)=x$. Two curves $\sigma$ and $\tau$ are said to be \emph{tangent} at $x$ if there exists a neighbourhood $U$ of $x$ with local coordinates $(x^i)$ such that
\begin{equation}
	\left.\frac{\mathrm{d} (x^i\circ \sigma)} {\mathrm{d}t }\right|_{t=0}
	=\left.\frac{\mathrm{d} (x^i\circ \tau)} {\mathrm{d}t }\right|_{t=0},\quad 1\leq i\leq m.
\end{equation}
This definition is independent of the local coordinate system and defines an equivalence class of curves. The equivalence class $[\sigma]$ of $\sigma$ is called the \emph{tangent vector} of $\sigma$ at $x$, also denoted by $\dot{\sigma}(0)$. The set of equivalence classes of curves at $x$ is called the \emph{tangent space of $M$ at $x$}, denoted by $T_x M$. The coordinates $(x^i)$ induce a basis $\left\{\partial/\partial x^i  \right\}$ of $T_xM$.
 Let
\begin{equation}
	TM = \bigcup_{x\in M} T_x M.
\end{equation}
Let $\tau_M:TM\to M$ be the canonical projection, given by $\tau_m(v)=x$ for each $v\in T_x M$. The triple $(TM,\tau_M,M)$ is a vector bundle called the \emph{tangent bundle of $M$}. The $2m$-dimensional manifold $TM$ is also sometimes called the tangent bundle of $M$. 

If $f:M\to N$ is a mapping of class $C^1$ between the manifolds $M$ and $N$, the \emph{tangent of $f$} is given by
\begin{equation}
\begin{aligned}
	Tf:TM&\to TN\\
	Tf([\sigma]_x)&=[f\circ \sigma]_{f(x)}
\end{aligned}
\end{equation}
for each $x\in M$.

Consider a vector bundle $(E,\pi,M)$. A differentiable mapping $s:M\to E$
is called a \emph{section} of the vector bundle if
\begin{equation}
	\pi\circ s =\operatorname{id}:M\to M.
\end{equation}
A \emph{vector field} on $M$ is a section $X:M\to TM$ assigning to each point $x\in M$ a tangent vector $X(x)$ of $M$ at $x$. If $M$ has local coordinates $(x^i)$, then $\left\{ (\partial/\partial x^i)_x  \right\}_{i=1}^{\dim M}$ is a basis for $T_xM$ for each $x\in M$. A vector field $X$ can thus be locally written as\footnote{Throughout the text, the summation convention over repeated indices is employed.}
\begin{equation}
	X=X^i(x) \frac{\partial  } {\partial x^i}.
\end{equation}
A curve $\sigma:I\to M$ is called an \emph{integral curve} of $X$ if, for every $t\in I$, $X(\sigma(t))$ is the tangent vector to the curve at $\sigma(t)$.
A vector field acts as a differential operator on a differentiable function $f$, locally 
\begin{equation}
	X(f)=X^i \frac{\partial f} {\partial x^i}.
\end{equation}

Let $M$ be a $m$-dimensional differentiable manifold. The \emph{cotangent space of $M$ at $x$}, denoted by $T_x^*M$, is the dual vector space of $T_xM$ (see \cite{ibort_alberto_notas_2001} for dual vector space and other concepts from linear algebra). 
An element $\alpha$ of $T_x^*M$ is called a \emph{tangent covector} of $M$ at $x$. The coordinates $(x^i)$ induce a basis $\left\{\dd x^i  \right\}$ of $T_x^*M$.
Set 
\begin{equation}
	T^*M=\bigcup_{x\in M} T_x^* M,
\end{equation}
and let $\pi_M:T^*M\to M$ be the canonical projection given by $\pi_M(\alpha)=x$ for each $\alpha\in T_x^*M$. The triple $(T^*M,\pi_M, M)$ is a vector bundle called the \emph{cotangent bundle of $M$}. The $2m$-dimensional manifold $T^*M$ is also sometimes referred to as the cotangent bundle of $M$.

A \emph{$1$-form} on $M$ is a section $\alpha:M\to T^*M$ assigning to each point $x\in M$  a tangent covector $\alpha(x)\in T_x^*M$. Consider a differentiable function $f:M\to \RR$. The \emph{differential} of $f$ at $x\in M$ is a linear mapping
\begin{equation}
	\dd f(x):T_x M \to T_{f(x)} \RR \cong \RR.
\end{equation}
Moreover, $\dd f(x)$ can be regarded as a tangent covector of $M$ at $x$. Let $v\in T_x M$ and consider a curve $\sigma$ on $M$ such that $\sigma(0)=x$ and $\dot{\sigma}(0)=v$. Then
\begin{equation}
	\dd f(x)(v) = \dd f(x) [v] = [f\circ \sigma],
\end{equation}
but $ [f\circ \sigma]$ is the tangent vector of $\RR$ at $f(x)$ given by the curve $f\circ \sigma$, so if $t$ denotes the coordinate of $\RR$, 
\begin{equation}
	[f\circ \sigma] = \left( \frac{\mathrm{d}(f\circ \sigma)} {\mathrm{d}t}  \right)_{t=0} \left( \frac{\partial  } {\partial t}  \right)_{f(x)}.
\end{equation}
On the other hand,
\begin{equation}
	v(f) = \left( \frac{\mathrm{d}(f\circ \sigma)} {\mathrm{d}t}  \right)_{t=0},
\end{equation}
and hence
\begin{equation}
	\dd f(x) (v) = v(f).
\end{equation}
In particular, if $M$ has local coordinates $(x^i)$, then the set of 1-forms $\left\{ \dd x^i  \right\}_{i=1}^{\dim M}$ satisfies
\begin{equation}
	(\dd x^i) (x) \left( \frac{\partial  } {\partial x^j}  \right)_x = \left( \frac{\partial  } {\partial x^j}  \right)_x (x^i) = \delta^{i}_{j},
\end{equation}
where $ \delta^{i}_{j}$ is the well-known Kronecker delta. 
Therefore $\left\{(\dd x^i)(x)  \right\}_{i=1}^{\dim M}$ is a basis for $T_x^*M$. In fact, it is the dual basis of $\left\{(\partial/\partial x^i)_x  \right\}_{i=1}^{\dim M}$.
 Then a 1-form $\alpha$ can be locally written as
\begin{equation}
	\alpha=\alpha_i(x) \dd x^i.
\end{equation}
The differential of $f$ at $x$ given by
\begin{equation}
	\dd f(x) = \left( \frac{\partial f} {\partial x^i}  \right)_x (\dd x^i) (x),
\end{equation}
which is the usual expression from elementary calculus.

A $1$-form $\alpha$ is naturally paired to a vector field $X$, yielding a function $\left\langle \alpha,X  \right\rangle$ on $M$, locally given by
\begin{equation}
 	\left\langle \alpha,X  \right\rangle\equiv \alpha(X) = \alpha_i(x) X^i(x).
 \end{equation} 

From the point of view of geometric mechanics, if a manifold $Q$ is the configuration space (that is, the space of positions), its tangent bundle $TQ$ can be regarded as the space of velocities, and its cotangent bundle $T^*Q$ as the phase space. In fact, if $Q$ has coordinates $(q^i)$, it induces coordinates $(q^i,\dot{q}^i)$ on $TQ$, and coordinates $(q^i,p_i)$ on $T^*Q$. Physically $(q^i)$ are interpreted as the positions (generalized coordinates), $(\dot{q}^i)$ as the (generalized) velocities, and $(p_i)$ as the conjugate momenta.

Let $V$ and $W$ be vector spaces, and denote by $V^*$ and $W^*$ their respective dual vector spaces. The \emph{tensor product} $V^*\otimes W^*$ is defined so that for each $v^*\in V^*$ and $w^*\in W^*$, there exists a $v^*\otimes w^*\in V^*\otimes W^*$ such that
\begin{equation}
	v^*\otimes w^* (x,y) = v^*(x)\ w^*(y)
\end{equation}
for each $x\in V$, $y\in W$. Similarly, one can define $V\otimes W$. In addition, the condition
\begin{equation}
	v^*\otimes w^* \left( x\otimes y  \right) = v^*(x)\ w^*(y)   	
\end{equation}  
is required, so that  $V^*\otimes W^*$ is the dual of $V\otimes W$.
The \emph{tensor space of type $(r,s)$} is given by
\begin{equation}
	T^r_s (V) = V\underbrace{\otimes \cdots \otimes}_{r\text{-times}} V \otimes V^*\underbrace{\otimes \cdots \otimes}_{s\text{-times}}V^*.
\end{equation}
A \emph{tensor field $K$ of type $(r,s)$} on a manifold $M$ is a mapping assigning an element $K(x)\in T^r_s(T_x M)$ to each point $x\in M$, in other words, it is a section of the tensor bundle $T^r_s (TM)$. Locally, if M has coordinates $(x^i)$, $K$ can be written as
\begin{equation}
	K = K^{i_1\cdots i_r}_{j_1\cdots j_s} 
	\frac{\partial  } {\partial x^{i_1}}
	\otimes \cdots \otimes
	\frac{\partial  } {\partial x^{i_r}}
	\otimes \dd x^{j_1} 
	\otimes \cdots \otimes
	\dd x^{j_s}.
\end{equation}
A skew-symmetric tensor field of type $(0,p)$ is called a \emph{$p$-form}. A function $f:M\to \RR$ is sometimes called a \emph{$0$-form}. The set of $p$-forms on $M$ forms a vector bundle denoted by $\Omega^p(M)$. The full space of differential forms of a manifold $M$ is
\begin{equation}
	\Omega(M)\coloneqq \bigoplus_{p=0}^m\Omega^p(M).
\end{equation}
The \emph{exterior product} (or \emph{wedge product}) of $p$ basis elements of $T_x^*M$ is given by
\begin{equation}
	\dd x^{i_{1}} \wedge \dd x^{i_{2}} \wedge \cdots \wedge \dd x^{i_{p}}=\frac{1}{p !} \sum_{\sigma}\mathrm{sgn}(\sigma)\ \dd x^{\sigma\left(i_{1}\right)} \otimes \dd x^{\sigma\left(i_{2}\right)} \otimes \cdots \otimes \dd x^{\sigma\left(i_{p}\right)},
\end{equation}
where the sum is taken over all permutations $\sigma$ of $(1,2,\ldots,p)$. These products form a local basis of $\Omega^p(M)$. For instance, $\left\{\dd x\wedge \dd y,\dd x\wedge \dd z, \dd y\wedge \dd z  \right\}$ is a basis of $\Omega^2(\RR^3)$. Then a $p$-form $\alpha$ can be locally written as
\begin{equation}
	\alpha=\alpha_{i_1\cdots i_p} \dd x^{i_1} \wedge \cdots \wedge \dd x^{i_p}.
\end{equation}
The \emph{exterior product} of a $p$-form $\alpha$
and a $q$-form $\beta$ is given by
\begin{equation}
	\alpha \wedge \beta =\frac{1}{p ! q !} \alpha_{i_{1} \ldots i_{p}} \beta_{j_{1} \ldots j_{q}} \dd x^{i_{1}} \wedge \cdots \wedge \dd x^{i_{p}} \wedge \dd x^{j_{1}} \wedge \cdots \wedge \dd x^{j_{q}}.
\end{equation}
An $m$-dimensional manifold $M$ is said to be \emph{orientable} if there exists an $m$-form $\omega$ on $M$ such that $\omega(x)\neq 0$ for all $x\in M$; $\omega$ is called a \emph{volume form}. If a volume form $\Omega$ has compact support in a region contained in a neighbourhood $U\subset M$, it can be integrated, yielding the volume of the corresponding orientable manifold. If $(x^i)$ are coordinates on $U$, $\Omega$ can be written as
\begin{equation}
	\Omega = f(x^1,\ldots,x^{2m})\ \dd x^{1}\wedge \cdots \wedge \dd x^{2m},
\end{equation}
and the integral of $\Omega$ on $M$ is given by
\begin{equation}
	\int_M \Omega = \int_U f(x^1,\ldots,x^{2m})\ \dd x^{1}\cdots \dd x^{2m},
\end{equation}
where the integral on the right-hand side is the usual Riemann integral.

The \emph{exterior derivative} on $M$ is the unique linear operator $\dd:\Omega (M)\to \Omega(M)$ such that
\begin{enumerate}
\item $\dd:\Omega^p(M)\to \Omega^{p+1}(M)$,
\item $\dd^2=0$,
\item for each $f\in C^\infty(M)$, $\dd f$ is the differential of $f$, locally given by
\begin{equation}
	\dd f = \frac{\partial f} {\partial x^i} \dd x^i;
\end{equation}
\item if $\alpha\in \Omega^p(M)$ and $\beta\in \Omega(M)$, then
\begin{equation}
	\dd (\alpha\wedge \beta)=\dd \alpha \wedge \beta +(-1)^p \wedge \dd \beta.
\end{equation}
\end{enumerate}

A $p$-form $\omega$ is called \emph{closed} if $\dd \omega=0$, and \emph{exact} if there exists a $(p-1)$-form $\alpha$ such that $\omega=\dd \alpha$. Clearly, every exact form is closed. Furthermore, if $\omega$ is closed it is locally exact, that is, for each point in $M$ there exists a neighbourhood $U$ such that $\omega|_U$ is exact. The latter result is known as the \emph{Poincaré lemma}.

The \emph{interior product} of a vector field $X$ by a $p$-form $\alpha$ is defined as follows:
\begin{enumerate}
\item $\contr{X}\alpha=0$ if $p=0$,
\item $\contr{X}\alpha=\alpha(X)$ if $p=1$,
\item $\contr{X}\alpha$ is a $(p-1)$-form such that
\begin{equation}
	(\contr{X}\alpha) (Y_1,\ldots,Y_{p-1})=\alpha(X,Y_1,\ldots,Y_{p-1})
\end{equation}
for each $Y_1,\ldots,Y_{p-1}\in TM$.
\end{enumerate}

Let $U$ be a open neighbourhood of the manifold $M$. A \emph{local one-parameter group (of diffeomorphisms)} is a smooth mapping $\Phi:(-\epsilon,\epsilon)\times U \to M$, denoted by $\Phi_t(p)$ for each $p\in U$ and each $\abs{t}<\epsilon$, such that for any $p\in U$ and any $\abs{s},\abs{t}<\epsilon$
\begin{enumerate}
\item $\Phi_0(p)=p$;
\item if $\abs{t+s}<\epsilon$ and $\Phi_t(p)\in U$, then $\Phi_{t+s}(p)=\Phi_s\circ \Phi_t(p)$.
\end{enumerate}
A local one-parameter group induces a vector field on $U$. Indeed, if $(x^i)$ are coordinates at $U$ one can write
\begin{equation}
	X_p = \left.X^i_p \frac{\partial  } {\partial x^i}\right|_p,
\end{equation}
where
\begin{equation}
	X_{p}^i = \left.\frac{\mathrm{d}x^i(\Phi_t(p))} {\mathrm{d}t }\right|_{t=0}.
\end{equation}
Conversely, a vector field $X$ on $M$ defines a local one-parameter group for each point $p\in M$. Consider the system of ordinary differential equations
\begin{equation}
	\frac{\mathrm{d}x_p^i(t)} {\mathrm{d} t}=X^i(x_p(t)),\quad 1\leq i\leq m,
\end{equation}
with the initial conditions $x_p(0)=(x_p^1(0),\ldots,x_p^m(0))=p$. 
It can be easily shown that
\begin{equation}
	\Phi_t(p)=\Phi(t,p)=x_p(t)
\end{equation}
is a local one-parameter group. The local one-parameter group generated by $X$ is also known as the \emph{flow} of $X$.

As it will be detailed throughout the text, physically vector fields can represent infinitesimal symmetries, and the one-parameter groups they generate represent their associated finite symmetries. For instance, $X=\partial/\partial x$ generates translations along the $x$-direction.

Let $\varphi:M\to N$ be a smooth mapping between the manifolds $M$ and $N$. The \emph{differential} of $\varphi$ at $x$ is a linear mapping $\dd \varphi$
such that, for any function $f$ on $N$ and every vector field $X$ on $M$,
\begin{equation}
	\dd \varphi(X)(f) = X (f\circ \varphi).
\end{equation}
Let $\alpha$ be a $p$-form on $N$. The \emph{pullback} of $\alpha$ by $\varphi$, denoted by $\varphi^*\alpha$, is a $p$-form on $M$ given by
\begin{equation}
	\varphi^*\alpha(X_1,\ldots,X_p)(x)=\alpha(\dd\varphi(X_1),\ldots,\dd\varphi(X_p))(\varphi(x))
\end{equation}
for each $x\in M$ and $X_1,\ldots,X_p\in T_xM$. In particular, for each $f\in \Omega^0(N)$, $\varphi^*f=f\circ \varphi$. Let now $X$ be a vector field on $M$. 
The \emph{pushforward} of $X$ by $\varphi$ is a vector field $\varphi_*X$ on $N$ defined by
\begin{equation}
	(\varphi_* X)(g) = X(g\circ \varphi).
\end{equation}
for each function $g:N\to \RR$.

Let $X$ a vector field on $M$, and $\Phi_t$ the local one-parameter group generated by $X$.
Let $\alpha$ be a $p$-form on $M$.
The \emph{Lie derivative} of $\alpha$ with respect to $X$, denoted by $\liedv{X}\alpha$, is a $p$-form on $M$ given by
\begin{equation}
	\liedv{X}\alpha
	 = \lim_{t\to 0} \frac{\Phi_{t}^*(\alpha)-\alpha}{t}.
	 \label{Lie_derivative_definition}
\end{equation}
In practice, the Lie derivative can be computed by means of the so-called \emph{Cartan's formula}:
\begin{equation}
	\liedv{X} \alpha = \contr X (\dd \alpha)+ \dd (\contr{X}\alpha). \label{Cartan_formula}
\end{equation}
In particular, 
\begin{equation}
	\liedv{X}f=X(f)
\end{equation}
for every function $f$ on $M$. The definition of Lie derivative can be extended for a tensor field of type $(r,s)$. Consider the mapping
\begin{equation}
	\tilde{\Phi}_t: T^r_s(T_{\Phi_t(p)}M) \to T^r_s(T_{p}M) 
\end{equation}
given by
\begin{equation}
	\tilde{\Phi}_t \left( v_1\otimes \cdots \otimes v_r \otimes v^{*1} \otimes \cdots \otimes v^{*s}  \right)
	= \left( \Phi_t^{-1}  \right)_* (v_1) \otimes \cdots \otimes \left( \Phi_t^{-1}  \right)_* (v_r)
	\otimes \Phi_t^* (v^{*1})\otimes \cdots \otimes \Phi_t^* (v^{*s}),
\end{equation}
for any $v_1,\ldots,v_r \in T_{\Phi_t(p)}M$ and any $v^{*1},\ldots,v^{*s} \in T^*_{\Phi_t(p)}M$. If $T$ is a $(r,s)$-type tensor field, its \emph{Lie derivative} with respect to $X$ is defined by
\begin{equation}
	\liedv{X}T = \lim_{t\to 0} \frac{\tilde{\Phi}_t(T)-T}{t},
\end{equation}
and it is also an $(r,s)$-type tensor field.

Given two vector fields $X$ and $Y$ on $M$, consider a new vector field $[X,Y]$ on $M$ such that
\begin{equation}
	[X,Y](f) = X (Y(f)) - Y (X(f))
\end{equation}
for every function $f$ on $M$. This vector field is called the \emph{Lie bracket} (or the \emph{commutator}) of $X$ and $Y$. The commutator of $X$ and $Y$ coincides with the Lie derivative of $Y$ with respect to $X$, namely
\begin{equation}
	[X,Y] = \liedv{X} Y = -\liedv{Y} X.
\end{equation}
In addition, one can easily check that the following properties hold:
\begin{enumerate}
\item $[X,Y]=-[Y,X]$,
\item $[fX,gY]=fX(g)Y-gY(f)X+fg[X,Y]$,
\item $\left[[X,Y],Z\right]+\left[[Y,Z],X\right]+\left[[Z,X],Y\right]=0$ (Jacobi identity).

\end{enumerate}

Other relevant identities that will be employed are the following:
\begin{equation}
\begin{aligned}
	&\contr{X}(\alpha \wedge \beta)=\contr{X} \alpha \wedge \beta+(-1)^{p} \alpha \wedge \contr{X} \beta,\\
	&\contr{X}\contr{X}\alpha=0,\\
	&X(f)=\contr{X} \dd f,\\
	& \dd \alpha(X_0,\ldots, X_p) 
	= \sum_{i=0}^p (-1)^i X_i \left( \alpha(X_0,\ldots, X_{i-1},X_{i+1},\ldots, X_p)  \right)\\
	&\quad+\sum_{i<j} (-1)^{i+j} \alpha \left([X_i,X_j],X_0,\ldots, X_{i-1},X_{i+1},\ldots, X_{j-1},X_{j+1},\ldots,X_p  \right)\\
	&\varphi^*(\contr{X}\alpha)=\contr{\varphi_*X} (\varphi^*\alpha),\\
	&\varphi^*(\dd \alpha)=\dd (\varphi^*\alpha),\\
	&\liedv{fX}\alpha =f\liedv{X}\alpha +\dd f \wedge \contr{X}\alpha,\\
	&\contr{[X,Y]}\alpha =\liedv{X}\contr{Y}\alpha -\contr{Y}\liedv{X}\alpha,\\
	&\liedv{X} \dd \alpha =\dd \liedv{X}\alpha,\\
	&\liedv{X}\contr{X}\alpha =\contr{X}\liedv{X}\alpha,
\end{aligned}
\label{Cartan_properties}
\end{equation}
where $\alpha$ and $\beta$ are a $p$-form and a $q$-form respectively, and $X_0,\ldots, X_p$ are vector fields.

Finally, the concepts of Lie group and Lie algebras, as well as its actions on manifolds, will be important for the study of symmetries and reduction. See references \cite{ortega_momentum_2004,abraham_foundations_2008,leon_methods_1989,hall_lie_2015}. Recall that a \emph{Lie group} $G$ is a differentiable manifold which is also an algebraic group whose operations are differentiable, that is, 
\begin{enumerate}
\item it is equipped with a smooth and associative multiplication map $\mu:G\times G\to G$,
\item there exists an identity element $e$,
\item for each $g\in G$, there exists an inverse $g^{-1}$,
 \end{enumerate} 
so that, for any $g\in G$,
\begin{equation}
	\mu(g, e)=\mu(e, g)=g,
\end{equation}
and 
\begin{equation}
	\mu(g, g^{-1})=\mu(g^{-1}, g)=e.
\end{equation}
An \emph{action of a Lie group} $G$ on a manifold $M$ is a smooth mapping $\Phi:G\times M\to M$ such that
\begin{equation}
	\Phi(e,x)\equiv \Phi_e (x)=x
\end{equation}
for any $x\in M$. An action is said to be
\begin{enumerate}
\item \emph{free} if, for every $x\in M$, $\Phi_g(x)=x$ if and only if $g=e$,
\item \emph{proper} if and only if $\Phi$ is a proper mapping, that is, for any compact subset $K\subset M$, $\Phi^{-1}(K)$ is also compact. 
\end{enumerate}

\begin{remark}
The action of a compact Lie group is always proper.
\end{remark}

The \emph{isotropy group} of $x\in M$ is the Lie subgroup
\begin{equation}
G_{x}=\left\{g \in G \mid \phi_{g}(x)=x\right\} \subset G
\end{equation}
Let $F:M\to N$ be a mapping between the manifolds $M$ and $N$, and let $\Phi:G\times M\to M$ and $\Psi:G\times N\to N$ be group actions. The mapping $F$ is said to be \emph{equivariant} (with respect to these actions of $G$) if
\begin{equation}
	F(\Phi_g x) = \Psi_g F(x)
\end{equation}
for any $g\in G$.

 An \emph{algebra} is a vector space $V$ equipped with a bilinear multiplication. If this multiplication, denoted by $[\cdot,\cdot]:V\times V\to V$, additionally satisfies
\begin{enumerate}
\item anti-symmetry: $[X,Y]=-[Y,X]$,
\item Jacobi identity: $[X,[Y, Z]]+[Y,[Z, X]]+[Z,[X, Y]]=0$,
\end{enumerate}
for every $X,Y,Z\in V$, then $(V,[\cdot,\cdot])$ forms a \emph{Lie algebra}. For each Lie group $G$, there is a corresponding Lie algebra $\mathfrak{g}\approxeq T_eG$. The left translation $L_g:G\to G$ defined by $g\in G$ is given by
\begin{equation}
	L_g(h) = gh
\end{equation}
for each $h\in G$. A vector field $X$ on $G$ is called \emph{left invariant} if the pushforward of every left translation leaves it unchanged, namely
\begin{equation}
	L_{g*}X=X
\end{equation}
for every $g\in G$. The vector subspace $\mathfrak{X}_L(G)$ of left invariant vector fields on $G$  is isomorphic to $T_eG$, the tangent space of $G$ at the identity. More specifically, for each $X\in \mathfrak{X}_L(G)$, it corresponds its value at the identity $X(e)\in T_eG$. The inverse of this map is given by
\begin{equation}
\begin{aligned}
	T_e G & \to \mathfrak{X}_L(G)\\
	\xi &\mapsto \xi_L,
\end{aligned}
\end{equation}
with $\xi_L$ such that $\xi_L(e)=\xi$. Via this isomorphism, $T_e G$ becomes a Lie algebra, with Lie bracket
\begin{equation}
	[\xi_L,\eta_L] = [\xi,\eta]_L,
\end{equation}
where $[\cdot,\cdot]$ on the left-hand side denotes the Lie bracket of vector fields on $G$.
It is known as the Lie algebra of $G$ and usually denoted by $\mathfrak{g}$.

 Consider a left invariant vector field $X$ on $G$.
The \emph{exponential map} is given by
\begin{equation}
\begin{aligned}
\exp:\mathfrak{g}&\to G\\
\exp(\xi)&=\gamma_\xi(1),
\end{aligned}
\end{equation}
where $\xi\in \mathfrak{g}$, and $\gamma_\xi:\RR\to G$ is the integral curve of $\xi$ through $e$.

The \emph{adjoint action} is an action of $G$ on its Lie algebra $\mathfrak{g}$ given by
\begin{equation}
\begin{aligned}
\operatorname{Ad}:G\times \mathfrak{g}&\to \mathfrak{g}\\
(g,\xi) &\mapsto \operatorname{Ad}_g \xi = \left.\frac{\mathrm{d} } {\mathrm{d}t} g \exp(t\xi) g^{-1} \right|_{t=0}.
\end{aligned}
\end{equation}
The \emph{coadjoint action} is an action $\operatorname{Ad}^*:G\times \mathfrak{g}^*\to \mathfrak{g}^*$, with $\mathfrak{g}^*$ the dual $\mathfrak{g}$, such that
\begin{equation}
	(\operatorname{Ad}^*_g \alpha)(\xi)=
	 \alpha(\operatorname{Ad}_g\xi)
\end{equation}
for each $\alpha\in \mathfrak{g}^*$, each $\xi \in\mathfrak{g}$ and each $g\in G$.

Given a group $G$ acting on a manifold $M$, one can define an equivalence relation by setting $x\sim y$ (with $x, y\in M$) if there exists a $g\in G$ such that $\Phi_g(x)=y$. The equivalence classes are called the \emph{orbits} of $G$ in $M$. If $G$ is a Lie group acting smoothly, freely and properly on a differentiable manifold $M$, it can be shown (see reference \cite[Theorem 21.10]{Lee2012}) that the \emph{orbit space} $M/G$ is a differentiable manifold of dimension $\dim M - \dim G$.

\section{Symplectic geometry} \label{section_symplectic}
Symplectic geometry is the natural framework for Hamiltonian mechanics (see references \cite{abraham_foundations_2008,leon_methods_1989,arnold_mathematical_1978,godbillon_geometrie_1969,munoz-lecanda_miguel_c_geometrisistemas_2007}). As it will be seen, a phase space has the structure of a symplectic manifold. Moreover, Lagrangian systems can also be endowed with a symplectic structure when the Legendre transform is well-defined.

Consider a manifold $M$ and a $2$-form $\omega$ on $M$. For each point $x\in M$ one can define a mapping
\begin{equation}
\begin{aligned}
	\flat_{\omega(x)}:T_xM&\to T_x^*M\\
	X&\mapsto \contr{X}\omega(x). \label{mapping_S_omega}
\end{aligned}
\end{equation}
The dimension of the image of $\flat_{\omega(x)}$ is called the \emph{rank} of $\omega$ at $x$:
\begin{equation}
	\operatorname{rank} \omega(x) = \dim (\operatorname{Im}\flat_{\omega(x)}).
\end{equation}
It can be shown that the rank of every 2-form is an even number.
The dimension of the kernel of $\flat_{\omega(x)}$ is called the \emph{corank} of $\omega$ at $x$:
\begin{equation}
	\operatorname{corank} \omega(x) = \dim (\operatorname{ker}\flat_{\omega(x)}),
\end{equation}
where
\begin{equation}
	\operatorname{ker}\flat_{\omega(x)} = \left\{X\in T_xM\mid \contr{X}\omega(x)=0  \right\}.
\end{equation}
A 2-form is called \emph{non-degenerate} if $\operatorname{corank} \omega(x)=0$ for every $x\in M$. 
A \emph{symplectic form} (or \emph{symplectic structure}) on a manifold $M$ is a non-degenerate closed 2-form on $M$. A manifold $M$ endowed with a symplectic form $\omega$, denoted by $(M,\omega)$, is called a \emph{symplectic manifold}. Every symplectic manifold is even-dimensional.

By the \emph{Darboux theorem}, if $(M,\omega)$ is a $2m$-dimensional symplectic manifold, there exists a neighbourhood of each point with local coordinates $(q^1,\ldots,q^m,p_1,\ldots, p_m)$ such that the symplectic form can be written as
\begin{equation}
	\omega = \dd q^i \wedge \dd p_i. \label{symplectic_form_Darboux}
\end{equation}
Consider two symplectic manifolds $(M,\omega)$ and $(N,\alpha)$. A differentiable mapping $\varphi:M\to N$ is called a \emph{symplectic transformation} if $\varphi^*\alpha=\omega$. If a symplectic transformation is a global diffeomorphism, it is called a \emph{symplectomorphism}. In particular, a symplectic transformation $\varphi:M\to M$ which preserves the symplectic form, $\varphi^* \omega=\omega$, is called a \emph{canonical transformation}. 

 A \emph{symplectic vector field} (or an \emph{infinitesimal symplectic transformation}) is a vector field $X$ on $M$ which generates symplectic transformations. A vector field $X$ is a symplectic vector field if and only if
 \begin{equation}
 	\liedv{X}\omega=0.
 \end{equation}
 By Cartan's formula \eqref{Cartan_formula}, this is equivalent to $\contr{X}\omega$ being a closed 1-form. Let $\Phi_t$ be the flow of $X$. By the definition of Lie derivative \eqref{Lie_derivative_definition}, 
\begin{equation}
	\frac{\mathrm{d} } {\mathrm{d}t} \Phi_t^* \omega = \Phi_t^* \liedv{X}\omega=0,
\end{equation}
so $\Phi_t^*\omega$ is constant in $t$. On the other hand, $\Phi_0$ is the identity, and hence $\Phi_t^*\omega=\omega$. That is, the flow of a symplectic vector field is a canonical transformation. Conversely, a canonical transformation induces a symplectic vector field.

A symplectic manifold $(M,\omega)$ has two natural isomorphisms between $TM$ and $T^*M$:
\begin{equation}
\begin{aligned}
	\flat_{\omega}: TM & \rightarrow T^*M \\
	X & \mapsto \flat_{\omega}(X)=\contr{X} \omega,
\end{aligned}
\end{equation}
and $\sharp_\omega\coloneqq \flat_\omega^{-1}$, both of which are called \emph{musical isomorphisms}. To each function $f\in C^1(M,\RR)$, one can associate a vector field $X_f$ on $M$, given by
\begin{equation}
	X_f = \sharp_\omega(\dd f),
\end{equation}
called a \emph{Hamiltonian vector field}. Clearly, every Hamiltonian vector field is symplectic. The \emph{Poisson bracket} is the bilinear operation
\begin{equation}
\begin{aligned}
\left\{\cdot,\cdot  \right\}: C^\infty(M,\RR)\times C^\infty(M,\RR) &\to C^\infty(M,\RR)\\
\left\{f,g  \right\} &=\omega(X_f,X_g),
\end{aligned}
\end{equation}
with $X_f, X_g$ the Hamiltonian vector fields associated to $f$ and $g$, respectively. It satisfyes the following properties:
\begin{enumerate}
\item anti-symmetry: $\{f, g\}=-\{g, f\}$,
\item Leibniz product: $\{f h, g\}=\{f, g\} h+f\{h, g\}$,
\item Jacobi identity: $\{f,\{g, h\}\}+\{g,\{h, f\}\}+\{h,\{f, g\}\}=0$,
\end{enumerate}
for each $f,g \in C^\infty(M,\RR)$.
By the properties i) and iii), the vector space of $C^\infty$-functions on $M$ forms a Lie algebra, with Lie bracket being the Poisson bracket.
 In addition,
\begin{equation}
	\left\{f,g  \right\} = -\liedv{X_f}g = -X_f(g) = \liedv{X_g}f = X_g(f).
\end{equation}
In canonical coordinates $(q,p)$, 
\begin{equation}
	\left\{f,g  \right\}=\frac{\partial f}{\partial q^{i}} \frac{\partial g}{\partial p_{i}}-\frac{\partial f}{\partial p_{i}} \frac{\partial g}{\partial q^{i}},
\end{equation}
which is the definition of Poisson bracket usually given in undergraduate Classical Mechanics courses \cite{goldstein_mecanica_1987,gantmakher_lectures_1970}. The \emph{Poisson bracket} of two 1-forms $\alpha$ and $\beta$ on $M$ is a 1-form on $M$ given by
\begin{equation}
\left\{\alpha,\beta  \right\} = -\flat_\omega \left( [X_\alpha,X_\beta]  \right).
	\label{Poisson_bracket_forms}
\end{equation}
Hence the diagram
\begin{center}
\begin{tikzcd}
TM\times TM \arrow[rr, "{-[\cdot,\cdot]}"] \arrow[dd, "\flat_\omega \times \flat_\omega"'] &  & TM \arrow[dd, "\flat_\omega"] \\\\
\Omega^1(M)\times \Omega^1(M) \arrow[rr, "{\left\{\cdot,\cdot\right\}}"]                  &  & \Omega^1(M)                 
\end{tikzcd}
\end{center}
commutes (that is, all the directed paths in the diagram with the same start and endpoints lead to the same result).
 Then the vector space $\Omega^1(M)$ with the Poisson bracket defined in equation \eqref{Poisson_bracket_forms} is also a Lie algebra. The Poisson bracket of 1-forms satisfyes the following properties:
\begin{enumerate}
\item $\{\alpha, \beta\}=-L_{X_{\alpha}} \beta+L_{X_{\beta}} \alpha+d\left(i_{X_{\alpha}} i_{X_{\beta}} \omega\right)$,
\item if $\alpha$ and $\beta$ are closed 1-forms, then ${\alpha,\beta}$ is exact;
\item $\left\{\dd f,\dd g  \right\}=\dd \left\{f,g  \right\}$ for any $f,g\in C^\infty(M)$.
\end{enumerate}
A corollary of the latter property is that
\begin{equation}
	X_{\left\{F,G  \right\}}=-[X_f,X_g].
\end{equation}

Every cotangent bundle $T^*Q$ is itself a symplectic manifold. In fact, the concept of symplectic manifolds emerged from cotangent bundles and Hamiltonian mechanics. Let $Q$ be a manifold with coordinates $(q)$, and $(q,p)$ the induced coordinates on its cotangent bundle $T^*Q$.
The 1-form on $T^*Q$ defined locally as
\begin{equation}
	\theta=p_i\dd q^i \label{canonical_1_form}
\end{equation}
is called the \emph{canonical 1-form} (or \emph{tautological 1-form}). The 2-form $\omega=-\dd \theta$ is precisely the symplectic form in the Darboux coordinates as in equation \eqref{symplectic_form_Darboux}. This 2-form is known as the \emph{canonical symplectic form}. The canonical 1-form on $T^*Q$ is the unique 1-form on $T^*Q$ such that 
\begin{equation}
	\alpha^*\theta =\alpha \label{property_defining_canonical_1_form}
\end{equation}
for any 1-form $\alpha$ on $Q$, and thus
\begin{equation}
	\alpha^* \omega 
	= -\alpha^* \dd \theta
	= -\dd (\alpha^* \theta)
	= -\dd \alpha.
\end{equation}
If $Q$ is an $n$-dimensional manifold, then $T^*Q$ is $2n$-dimensional and the $2n$-form $\Omega$ given by
\begin{equation}
	\Omega = \omega \underbrace{\wedge \cdots \wedge}_{n\text{ times}} \omega\equiv \omega^n
\end{equation}
is a volume form.

\section{Hamiltonian mechanics}\label{section_Hamiltonian}
The main references for this section are \cite{leon_methods_1989,abraham_foundations_2008}; see also references \cite{arnold_mathematical_1978,jose_classical_1998}.
Geometrically, a Hamiltonian system is a triple $(M,\omega,H)$, with $(M,\omega)$ a symplectic manifold and $H:M\to \RR$ the Hamiltonian function. 
If a mechanical system has configuration space $Q$, its phase space is the cotangent bundle $T^*Q$, with canonical symplectic structure \eqref{symplectic_form_Darboux}. The trajectories of the system are the integral curves of the vector field $X_H$ on $M$ given by
\begin{equation}
	\contr{X_H} \omega = \dd H.
\end{equation}
Locally, $X_H$ can be written as
\begin{equation}
	X_H = X^i \frac{\partial}{\partial q^{i}} + \tilde{X}_i\frac{\partial}{\partial p_{i}},
\end{equation}
so
\begin{equation}
	\contr{X_H} \omega = X^i \dd p_i - \tilde{X}_i \dd q^i.
\end{equation}
On the other hand,
\begin{equation}
	\dd H = \frac{\partial H} {\partial q^i} \dd q^i + \frac{\partial H} {\partial p_i} \dd p_i,
\end{equation}
so $X^i= \partial H/\partial p_i$ and $\tilde{X}_i= \partial H/\partial q^i$, that is,
\begin{equation}
X_{H}=\frac{\partial H}{\partial p_{i}} \frac{\partial}{\partial q^{i}}-\frac{\partial H}{\partial q^{i}} \frac{\partial}{\partial p_{i}}. \label{vector_field_Hamiltonian}
\end{equation}
The integral curves of $X_H$ are then given by
\begin{equation}
\begin{aligned}
&\dot{q}^i	= \frac{\partial H} {\partial p_i},\\
& \dot{p}_i = - \frac{\partial H} {\partial q^i},
\end{aligned}	
\end{equation}
which are the well-known \emph{Hamilton's equations}. They can also be written as
\begin{equation}
\begin{aligned}
& \dot{q}^i = \left\{q^i,H  \right\},\\
& \dot{p}_i = - \left\{p_i,H \right\},
\end{aligned}
\label{Hamilton_eqs_Poisson}
\end{equation}
where $\left\{\cdot,\cdot  \right\}$ is the Poisson bracket on $(M,\omega)$.

Since $X_H$ is, by definition, a Hamiltonian vector field, it is also a symplectic vector field. Therefore its flow $\Phi_t$ is a canonical transformation, that is, $\Phi_t^*\omega=\omega$. From here, one can immediately obtain the following well-known result in Hamiltonian mechanics.

\begin{proposition}[Liouville's theorem]
	The flow of $X_H$ preserves the volume form $\omega^n$, that is,
	\begin{equation}
		\Phi_t^*\omega^n =\omega^n.
	\end{equation}
	In other words, the volume of the phase space is preserved along the evolution of the system.
\end{proposition}

One can also see that a canonical transformation $\varphi:T^*Q\to T^*Q$ preserves the Hamilton's equations. Indeed,
\begin{equation}
	\varphi^* \left\{f,g  \right\} 
	= \varphi^* (X_g(f))
	= \varphi^* (\contr{X_g} \dd f)
	=\contr{\varphi_* X_g} (\varphi^*\dd f)
	=\contr{\varphi_* X_g} \dd( \varphi^* f)
	= \contr{\varphi_* X_g} \contr{X_{\varphi^*f}} \omega
	= -\contr{X_{\varphi^*f}} \contr{\varphi_* X_g} \omega,
\end{equation}
where the properties \eqref{Cartan_properties} have been used,
but 
\begin{equation}
	\contr{\varphi_* X_g}\omega 
	=\contr{\varphi_* X_g}(\varphi^*\omega )
	=\varphi^* (\contr{X_g}\omega)
	= \varphi^*(\dd g	)
	= \dd (\varphi*g)
	= \contr{X_{\varphi^*g}}\omega,
\end{equation}
so 
\begin{equation}
	\varphi^* \left\{f,g  \right\}  
	= -\contr{X_{\varphi^*f}} \contr{X_{\varphi^*g}} \omega
	= \omega(X_{\varphi^*f},X_{\varphi^*g})
	= \left\{\varphi^*f, \varphi^*g  \right\},
\end{equation}
that is,
\begin{equation}
	\left\{f,g  \right\} \circ \varphi = 
	\left\{f\circ \varphi,g \circ \varphi \right\}.
\end{equation}
In particular, if $\varphi:(q,p)\mapsto (\tilde{q},\tilde{p})$, with $(q,p)$ and $(\tilde{q},\tilde{p})$ canonical coordinates, 
\begin{equation}
\begin{aligned}
& \left\{q^i,H  \right\} \circ \varphi
= \left\{q^i\circ \varphi, H\circ\varphi  \right\}
= \left\{\tilde{q}^i, K  \right\}
= \frac{\mathrm{d}\tilde{q}^i} {\mathrm{d} t},\\
& \left\{p_i,H  \right\} \circ \varphi
= \left\{p_i\circ \varphi, H\circ\varphi  \right\}
= \left\{\tilde{p}_i, K  \right\}
= -\frac{\mathrm{d}\tilde{p}_i} {\mathrm{d} t},
\end{aligned}
\end{equation}
with $K=H\circ \varphi$ the Hamiltonian in the new coordinates. The right-hand side equalities take the same form as the original Hamilton's equations \eqref{Hamilton_eqs_Poisson}.

Consider a vector field $X$ on a manifold $M$. A 1-form $\alpha$ on $M$ is called a \emph{first integral} of $X$ if $\alpha(X)=0$. Similarly, a function $f$ on $M$ such that $X(f)=0$ is also called a \emph{first integral} of $X$. Clearly, $\flat_\omega(X)$ is a first integral of $X$, for any $X$ on a symplectic manifold $(M,\omega)$. Any function $f$ is a first integral of its corresponding Hamiltonian vector field $X_f$. In particular, a first integral $f$ of $X_H$ is called a \emph{constant of the motion}. Indeed, if $f$ is a constant of the motion and $\sigma(t)$ is an integral curve of $X_H$, 
\begin{equation}
 	0 = X_H(f)(\sigma(t)) 
 	= \frac{\mathrm{d} } {\mathrm{d}t} f(\sigma(t )),
\end{equation} 
so $f$ is constant along the evolution of the system. In particular, the Hamiltonian function $H$ is conserved along the evolution (conservation of the energy). 

As it is well-known, there are multiple physical systems for which the energy is not conserved.  In order to describe some of these systems, one can consider non-autonomous Hamiltonians, that is, $H$ depending explicitly on time. These systems are naturally described in cosymplectic geometry\footnote{A cosymplectic manifold \cite{de_Leon_2017} is an $(2n+1)$-dimensional manifold equipped with a closed 1-form $\eta$ and a closed 2-form $\omega$ such that $\eta\wedge \omega^n\neq 0$. In $T^*Q\times \RR$ with Darboux coordinates $\left\{t,q^i, p_i  \right\}$, $\eta=\dd t$ and $\omega$ takes the form in equation \eqref{symplectic_form_Darboux}, 
which is basically the framework for autonomous Hamiltonian systems with the addition of a time dimension.
 }.
Other non-conservative systems are characterized by an external force (described by a 1-form) together with the Hamiltonian function. Part \ref{part_paper} of this thesis is devoted to mechanical systems with external forces. Other physical systems for which the energy is not conserved can be described in the framework of contact geometry, or in other geometric structures, rather than symplectic geometry. See references \cite{de_leon_contact_2019,de_leon_review_2021} and references therein for contact geometry and its applications in physics.

 In terms of Poisson brackets, one can write $X_H(f)=\left\{f,H  \right\}$. If $\left\{f,H  \right\}=0$, obviously $\left\{\dd f, \dd H  \right\}= \dd \left\{f,H  \right\}=0$. On the other hand, 
 \begin{equation}
 		\left\{\dd f, \dd H  \right\} 
 		= - \left\{\dd H, \dd f  \right\}
 		=  \flat_\omega([X_H,X_f])
 		= \contr{[X_H,\ X_f]}\omega
 		= \liedv{X_H} \contr{X_f} \omega - \contr{X_f}\liedv{X_H}\omega.
 \end{equation}
 Now, since $X_H$ is a Hamiltonian vector field, $\liedv{X_H}\omega=0$. Then
 \begin{equation}
 	\left\{\dd f, \dd H  \right\}
 	= \liedv{X_H} \contr{X_f} \omega
 	= \liedv{X_H} (\dd f)
 	= \dd (\liedv{X_H} f)
 	= \dd (X_H(f)),
 \end{equation}
 so $\left\{\dd f, \dd H  \right\} $ vanishes if and only if $X_H(f)$ is a constant function. Without loss of generality, this constant can be taken as zero (clearly, $f$ is a constant of the motion if and only if $f+K$ is a constant of the motion for any $K\in \RR$). In addition, since $\omega$ is non-degenerate, $ \contr{[X_H,\ X_f]}\omega$ vanishes if and only if $[X_H,X_f]$ does. Moreover, observe that $X_H(f)=-X_f(H)$.  This results can be summed up as follows.

 \begin{proposition}\label{proposition_Hamiltonian_constants_motion}
 The following statements are equivalent:
 \begin{enumerate}
 \item $f$ is a constant of the motion,
 \item $\left\{f,H  \right\}=0$,
 \item $[X_H,X_f]=0$,
 \item $H$ is a first integral of $X_f$.
 \end{enumerate}
 
 \end{proposition}

\begin{example}[Conservation of momentum]
Consider a Hamiltonian system $(M,\omega,H)$ with canonical coordinates $(q,p)$. Each momentum $p_i$ is itself a function on $M$. By equation \eqref{vector_field_Hamiltonian}, 
\begin{equation}
	X_H(p_i)=-\frac{\partial H} {\partial q^i},
\end{equation}
so $p_i$ is a constant of the motion if and only if ${\partial H}/{\partial q^i}=0$, that is, if and only if $H$ is invariant under translations along the $q^i$-direction.
\end{example}

\begin{example}[Conservation of angular momentum]
Consider a Hamiltonian system $(T^*(\RR^n), \omega, H)$, with Hamiltonian function of the form
\begin{equation}
	H=\frac{\vec{p}^2}{2m}+V(\norm{\vec{x}}),
\end{equation}
with $\vec{x}=(x^1,\ldots, x^n)$ and $\vec{p}=(p_1,\ldots,p_n)$ the canonical coordinates on $T^*(\RR^n)\approxeq \RR^{2n}$. Then
\begin{equation}
	X_H = \frac{p_i}{m} \frac{\partial  } {\partial x^i} - \frac{V'(\norm{\vec{x}})}{\norm{\vec{x}}} x^i\frac{\partial  } {\partial p_i}. \label{Hamiltonian_V_r}
\end{equation}
Let $f=x^1p_2-x^2p_1$. Then
\begin{equation}
	X_H(f)=\frac{p_1}{m} p_2 - \frac{p_2}{m} p_1
	- \frac{V'(\norm{\vec{x}})}{\norm{\vec{x}}}x^1(-x^2)
	- \frac{V'(\norm{\vec{x}})}{\norm{\vec{x}}}x^2 x^1 = 0,
\end{equation}
so $f$ is a constant of the motion. The Hamiltonian vector field associated to $f$ is given by
\begin{equation}
	X_f =  -x^2 \frac{\partial  } {\partial x^1} +x^1 \frac{\partial  } {\partial x^2}  - p_2 \frac{\partial  } {\partial p_1} + p_1 \frac{\partial } {\partial p_2}.
\end{equation}
The reader can easily check that each of the statements in Proposition \ref{proposition_Hamiltonian_constants_motion} is verified. Analogously, one can show that $L_i=\varepsilon_{ijk} x^j p_k$, for $i=1,\ldots,n$, are constants of the motion, where $\varepsilon_{ijk}$ is the totally antisymmetric Levi-Civita symbol. In particular, for a mechanical system on $\RR^3$ with a Hamiltonian function of the form \eqref{Hamiltonian_V_r}, the angular momentum $\vec{L}=\vec{x}\times \vec{p}$ is conserved.

\end{example}

\section{The geometry of tangent bundles} \label{section_tangent}
As it has been explained in Section \ref{section_Hamiltonian}, the Hamiltonian description of a mechanical system is settled in the cotangent bundle of its configuration space. Similarly, the Lagrangian formalism is settled in the tangent bundle of its configuration space. Unlike the cotangent bundle, the tangent bundle does not have a canonical symplectic structure. However, as it will be explained in Section \ref{section_Lagrangian}, when the Lagrangian function $L$ is ``well-behaving'', the symplectic structure on $T^*Q$ induces a symplectic structure on $TQ$ that depends on $L$.

Before presenting the Lagrangian formalism in its intrinsic form, some further geometric structures need to be introduced. See references \cite{leon_methods_1989,yano_almost_1967,yano_tangent_1973,Crampin_1983,munoz-lecanda_miguel_c_geometrisistemas_2007,godbillon_geometrie_1969}.

Let $Q$ be an $n$-dimensional manifold. If $f:Q\to \RR$ is a function on $Q$, its \emph{vertical lift} is a function
\begin{equation}
	f^v=f\circ \tau_Q:TQ\to \RR,
\end{equation}
on $TQ$, with $\tau_Q:TQ\to Q$ the canonical projection. If a point $P\in TQ$ has induced coordinates $(q,\dot{q})=(q^1,\ldots,q^n,\dot{q}^1,\ldots,\dot{q}^n)$, then
\begin{equation}
	f^v(P)=f^v(q,\dot{q})=f\circ \tau_Q (q,\dot{q})=f(q).
\end{equation}
If $\alpha$ is a 1-form on $Q$, it can be naturally regarded as a function on $TQ$, denoted by $\imath \alpha$. Locally, if $Q$ has coordinates $(q^i)$, $TQ$ has induced coordinates $(q^i,\dot{q}^i)$ and 
\begin{equation}
	\alpha=\alpha_i\dd q^i,
\end{equation}
then
\begin{equation}
	\imath \alpha = \alpha_i \dot{q}^i.
\end{equation}
A \emph{vertical vector field} $\tilde{X}$ is a vector field on $TQ$ such that
\begin{equation}
	\tilde{X}(f^v)=0
\end{equation}
for every function $f$ on $Q$. Locally,
 if
\begin{equation}
	\tilde{X}=X^i \frac{\partial  } {\partial q^i} + Y^i \frac{\partial  } {\partial \dot{q}^i},
\end{equation}
then
\begin{equation}
	\tilde{X}(f^v)
	=X^i \frac{\partial f} {\partial q^i},
\end{equation}
so
$\tilde{X}$ is a vertical vector field if and only if it is of the form
\begin{equation}
	\tilde{X}=\tilde{X}^i \frac{\partial  } {\partial \dot{q}^i}.
\end{equation}
If $X$ is a vector field on $Q$, its \emph{vertical lift} is the unique vector field $X^v$ on $TQ$ such that
\begin{equation}
	X^v(\imath\alpha) = (\alpha(X))^v
\end{equation}
for every 1-form $\alpha$ on $Q$. Clearly, every vertical lift is a vertical vector field by construction. Locally, if 
\begin{equation}
	X = X^i \frac{\partial  } {\partial q^i},
\end{equation}
then its vertical lift is given by
\begin{equation}
	X^v=X^i \frac{\partial  } {\partial \dot{q}^i}.
\end{equation}
It can be easily proven that
\begin{equation}
	[X^v,Y^v]=0
\end{equation}
holds for every pair of vector fields $X$ and $Y$ on $Q$.

If $f$ is a function on $Q$, its \emph{complete lift} is a function $f^c$ on $TQ$ defined by
\begin{equation}
	f^c=\iota(\dd f).
\end{equation}
Locally,
\begin{equation}
	f^c = \dot{q}^i \frac{\partial f } {\partial q^i}.
\end{equation}
If $X$ is a vector field on $Q$, its \emph{complete lift} is the unique vector $X^c$ on $TQ$ such that
\begin{equation}
	X^c(f^c)=(X(f))^c
\end{equation}
for every function $f$ on $Q$. Locally, if $X$ is given by
\begin{equation}
	X=X^i \frac{\partial  } {\partial q^i},
\end{equation}
then its complete lift is
\begin{equation}
	X^c=X^i \frac{\partial  } {\partial q^i} + \dot{q}^j \frac{\partial X^i} {\partial q^j} \frac{\partial  } {\partial \dot{q}^i}.
\end{equation}
Two relevant relations that hold are
\begin{equation}
	[X^v,Y^c]=[X,Y]^v,
\end{equation}
and
\begin{equation}
	[X^c,Y^c]=[X,Y]^c
\end{equation}
for every pair of vector fields $X$ and $Y$ on $Q$.

Let $M$ be a $2m$-dimensional manifold. An \emph{almost tangent structure} $S$ on $M$ is a tensor field of type $(1,1)$ such that
\begin{enumerate}
\item $S$ has constant rank $n$,
\item $S^2=0$.
\end{enumerate}
Then $M$ is called an \emph{almost tangent manifold}. For each $x\in M$,
\begin{equation}
	S_x:T_xM\to T_xM
\end{equation}
is a linear endomorphism. 

One can define a tensor field $S$ of type $(1,1)$ on $TQ$ such that
\begin{equation}
	S_y(\tilde{X})=(\tau_{Q*}\tilde{X})^v
\end{equation}
for each $y\in TQ$ and each $\tilde{X}\in T_y(TQ)$. Locally $S$ is given by
\begin{equation}
	S \left( \frac{\partial  } {\partial q^i}  \right) = \frac{\partial  } {\partial \dot q^i}, \qquad
	S \left( \frac{\partial  } {\partial \dot q^i}  \right) =0,
\end{equation}
or equivalently,
\begin{equation}
	S = \frac{\partial  } {\partial \dot{q}^i}\otimes \dd q^i. 
\end{equation}
Clearly, $S$ has constant rank $n$ and $S^2=0$, so it is an almost tangent structure on $TQ$ called the \emph{canonical almost tangent structure}. It is also known as the \emph{vertical endomorphism}. Observe that a vector field $\tilde{X}$ on $TQ$ is a vertical vector field if and only if
\begin{equation}
	S(\tilde{X})=0.
\end{equation}
In particular,
\begin{equation}
	S(X^v)=0
\end{equation}
for every $X$ on $Q$. Moreover,
\begin{equation}
	S(X^c)=X^v.
\end{equation}
The \emph{adjoint operator} $S^*$ of the vertical endomorphism $S$ on $TQ$ is defined by
\begin{equation}
\begin{aligned}
	&S^*(f)=f,\\
	&(S^*\omega)(X_1,\ldots,X_p) =\omega(S(X_1),\ldots S(X_p)),
\end{aligned}
\end{equation}
for every function $f$, every $p$-form $\omega$ and every vector fields $X_1,\ldots,X_p$ on $TQ$. Locally,
\begin{equation}
	S^*(\dd q^i)=0,
\end{equation}
and
\begin{equation}
	S^*(\dd \dot{q}^i) = \dd q^i,
\end{equation}
with $(q^i,\dot{q}^i)$ the induced coordinates on $TQ$. It satisfies
\begin{equation}
	\contr{X}S^* = S^* \circ \contr{SX}
\end{equation}
for any vector field $X$ on $TQ$.

Let $q\in Q$ and $(q,\dot{q})\in TQ$. Consider the vertical lift of $\dot{q}\in T_xQ$, namely
\begin{equation}
	\dot{q}=\dot{q}^i \frac{\partial  } {\partial q^i} \mapsto \dot{q}^v= \dot{q}^i \frac{\partial  } {\partial \dot{q}^i}.
\end{equation}
This defines a vector field on $TQ$:
\begin{equation}
\begin{aligned}
	\Delta:TQ&\to TTQ\\
	(q,\dot{q})&\mapsto ((q,\dot{q}),(q,\dot{q})^v),
\end{aligned}
\end{equation}
called the \emph{Liouville vector field}, locally given by
\begin{equation}
	\Delta = \dot{q}^i \frac{\partial  } {\partial \dot{q}^i}.
\end{equation}
A \emph{second order differential equation} (in what follows abbreviated as \emph{SODE}; also known as \emph{semispray}) on $Q$ is a vector field on $TQ$ that is a section of both $\tau_{TQ}:TTQ\to TQ$ and $T{\tau_Q}:TTQ\to TQ$. Locally,
if
\begin{equation}
	\xi = \tilde{\xi}^i \frac{\partial  } {\partial q^i}+\xi^i \frac{\partial  } {\partial \dot q^i},
\end{equation}
then
\begin{equation}
	T{\tau_Q}(\xi)
	= \tilde{\xi}^i \frac{\partial  } {\partial q^i},
\end{equation}
since $\tau_Q(q,\dot{q})=q$, that is,
\begin{equation}
	T{\tau_Q}(\xi(q^i,\dot{q}^i))
	= T{\tau_Q}\left(q^i,\dot{q}^i,\tilde{\xi}^i,\xi^i\right)
	= \left(q^i, \tilde{\xi}^i\right),
\end{equation}
but
\begin{equation}
	\tau_{TQ}(\xi(q^i,\dot{q}^i))=(q^i,\dot{q}^i),
\end{equation}
so a SODE $\xi$ is locally given by
\begin{equation}
	\xi = \dot{q}^i \frac{\partial  } {\partial q^i}+\xi^i \frac{\partial  } {\partial \dot q^i}.
	\label{SODE_local}
\end{equation}
From the local expressions, it is clear that a vector field $\xi$ on $TQ$ is a SODE if and only if
\begin{equation}
	S(\xi) = \Delta, \label{SODE_Liouville}
\end{equation}
with $S$ the vertical endomorphism on $TQ$.

Let $\gamma:I\subset \RR\to Q$ be a curve on $Q$, and let $\dot{\gamma}(t_0)$ be the tangent vector of the curve at $\gamma(t_0)$. Then the \emph{canonical lift} of $\gamma$ to $TQ$ is the curve
\begin{equation}
\begin{aligned}
	\tilde{\gamma} : I&\to TQ\\
	t &\mapsto \left( \gamma(t),\dot{\gamma}(t)  \right).
\end{aligned}
\end{equation}
Clearly, $\tau_Q\circ\tilde{\gamma}=\gamma$.

To avoid ambiguity in the following paragraph, let $(q^i,v^i)$ be the bundle coordinates on $TQ$ (which are denoted as $(q^i,\dot{q}^i)$ along the rest of the text).
Consider a SODE $\xi$ on $TQ$ locally given by 
\begin{equation}
	\xi = v^i \frac{\partial  } {\partial q^i}+\xi^i \frac{\partial  } {\partial v^i}.
\end{equation}
 A curve $\sigma:I\to TQ$ on $TQ$, with $\sigma(t)=\left( q^i(t),v^i(t)  \right)$, is an integral curve of $\xi$ if and only if
\begin{equation}
	\dot{\sigma}(t) = \xi(\sigma(t)),
\end{equation}
that is,
\begin{equation}
\begin{aligned}
	&\frac{\mathrm{d}q^i} {\mathrm{d} t} = (v^i\circ \sigma)(t)=v^i(t),\\
	&\frac{\mathrm{d}v^i} {\mathrm{d} t} = (\xi^i \circ \sigma)(t)
	= \xi^i \left( q^j(t),v^j(t)  \right),
\end{aligned}
\end{equation}
so $\sigma(t)=\left( q^i(t), \dd q^i(t)/\dd t  \right)$, and then
 \begin{equation}
 	\frac{\mathrm{d}^2 q^i } {\mathrm{d}t^2} = \xi^i \left( q^i, \frac{\mathrm{d}q^i} {\mathrm{d}t}  \right),\quad 1\leq i\leq m =\dim M. \label{equation_SODE}
 \end{equation}
 A \emph{solution} of a SODE $\xi$ is a curve $\gamma:I\to Q$ on $Q$ such that its canonical lift to $TQ$ is an integral curve of $\xi$. As a matter of fact, if $\gamma(t)=(q^i(t))$, then $\sigma(t)=\tilde{\gamma}(t)=(q^i(t),\dot{q}^i(t))$, so $\gamma$ is given by the system of ordinary differential equations \eqref{equation_SODE}. This motivates calling these vector fields second order differential equations.



A 1-form $\alpha$ on $TQ$ is called \emph{semibasic} if
\begin{equation}
	\alpha(\tilde{X}) = 0
\end{equation}
for any vertical vector field $\tilde{X}$ on $TQ$. In particular, if a 1-form $\alpha$ on $TQ$ can be written as
\begin{equation}
	\alpha=S^*(\dd f)
\end{equation}
for some function $f$ on $TQ$, then is semibasic.
Given a semibasic 1-form $\alpha$ on $TQ$, one can define the following morphism of fibre bundles 
\begin{equation}
\begin{aligned}
    &D_\alpha: TQ\to T^*Q,\\
    &\left\langle D_\alpha(v_q),w_q\right\rangle=\alpha(v_q)(u_{w_q}),
\end{aligned}
\end{equation}
for every $v_q,w_q\in T_qQ,\ u_{w_q}\in T_{w_q}(TQ)$, with $T\tau_Q(u_{w_q})=w_q$. In local coordinates, if
\begin{equation}
\alpha=\alpha_i(q,\dot q)\dd q^i,
\end{equation}
then
\begin{equation}
    D_\alpha(q^i,\dot q^i)=\left(q^i,\alpha_i(q^i,\dot{q}^i)\right).
\end{equation}
Conversely, given a morphism of fibre bundles
\begin{center}
\begin{tikzcd}
D: TQ \arrow[rr] \arrow[rd, "\tau_q"'] &   & T^*Q \arrow[ld, "\pi_Q"] \\
                                                          & Q &                         
\end{tikzcd},
\end{center}
 a semibasic 1-form $\alpha$ on $TQ$ can be defined by 
\begin{equation}
    \alpha_D(v_q)(u_{v_q})=\left\langle D(v_q),T\tau_Q(u_{v_q})\right\rangle,
\end{equation}
where $v_q\in T_qQ,\ u_{v_q}\in T_{v_q}(TQ)$.
If locally $D$ is given by
\begin{equation}
    D(q^i,\dot q^i)=(q^i,D_i(q,\dot q)),
\end{equation}
then
\begin{equation}
    \alpha_D=D_i(q,\dot q)\dd q^i.
\end{equation}
So there exists a one-to-one correspondence between semibasic 1-forms and fibred morphisms from $TQ$ to $T^*Q$. If $\theta$ is the canonical 1-form on $T^*Q$ and $D:TQ\to T^*Q$ a morphism of fibre bundles, then 
\begin{equation}
	\alpha_D= D^*\theta.
\end{equation}

\section{Lagrangian mechanics}\label{section_Lagrangian}
This section covers a geometric formulation of Lagrangian mechanics developed by J. Klein \cite{klein_espaces_1962,klein_les_1963,klein_operateurs_1963}. See also references \cite{leon_methods_1989,abraham_foundations_2008,munoz-lecanda_miguel_c_geometrisistemas_2007,Crampin_1983,godbillon_geometrie_1969}.

Consider a mechanical system with configuration space $Q$, with $Q$ an $n$-dimensional manifold. Then its \emph{space of velocities} is the tangent bundle $TQ$, and its \emph{Lagrangian function} is a differentiable function $L:TQ\to \RR$ on $TQ$. The \emph{Poincaré-Cartan 1-form} is given by
\begin{equation}
	\alpha_L=S^*(\dd L),
\end{equation}
where $S^*$ is the adjoint operator of the vertical endomorphism. The \emph{Poincaré-Cartan 2-form} is 
\begin{equation}
	\omega_L = -\dd\alpha_L.
\end{equation}
Locally, if $TQ$ has induced coordinates $(q^i,\dot{q}^i)$, then
\begin{equation}
	\dd L = \frac{\partial L} {\partial q^i} \dd q^i
	 + \frac{\partial L} {\partial \dot q^i} \dd\dot q^i,
\end{equation}
so
\begin{equation}
	\alpha_L = \frac{\partial L} {\partial \dot q^i} \dd q^i,
\end{equation}
and hence
\begin{equation}
	\omega_L = 
	\frac{\partial^2 L  } {\partial q^j \partial \dot{q}^i } \dd q^i \wedge \dd q^j
	+\frac{\partial^2 L  } {\partial \dot{q}^i \partial \dot{q}^j } \dd q^i \wedge \dd \dot q^j. \label{symplectic_form_Lagrangian_coords}
\end{equation}
One can write
\begin{equation}
	\alpha_L=p_i\dd q^i,
\end{equation}
where $p_i=\partial L/\partial \dot{q}^i$ is the canonical momenta. Notice that the Poincaré-Cartan 1-form is formally identical to the canonical 1-form $\theta$ on $T^*Q$ given by equation \eqref{canonical_1_form}. In fact, as it will be seen latter,
$\alpha_L$ is induced by $\theta$ via the Legendre transform.

The $2n$-form $\omega_L^n$ is locally given by
\begin{equation}
\begin{aligned}
	\omega_L^n & =  \frac{\partial^2 L  } {\partial \dot{q}^{i_1} \partial \dot{q}^{j_1} } \dd q^{i_1} \wedge \dd \dot q^{j_1}\wedge  \frac{\partial^2 L  } {\partial \dot{q}^{i_2} \partial \dot{q}^{j_2} } \dd q^{i_2} \wedge \dd \dot q^{j_2}
	\wedge \cdots \wedge  \frac{\partial^2 L  } {\partial \dot{q}^{i_n} \partial \dot{q}^{j_n} } \dd q^{i_n} \wedge \dd \dot q^{j_n}\\
	&= \pm \det \left( \frac{\partial^2 L  } {\partial \dot{q}^i \partial \dot{q}^j  }  \right) 
	\dd q^1\wedge \cdots \wedge \dd q^n\wedge \dd \dot{q}^1
	\wedge \cdots \wedge \dd \dot{q}^n.
\end{aligned}
\end{equation}
If $\omega_L$ is symplectic, obviously $\omega_L^n$ is a volume form. On the other hand, if the Hessian matrix
\begin{equation}
	(W_{ij}) = \left( \frac{\partial^2 L  } {\partial \dot{q}^i \partial \dot{q}^j  }  \right) 
\end{equation}
is non-singular, then $\omega_L$ is non-degenerate. Thus $\omega_L$ is a symplectic	form if and only if the Hessian matrix $(W_{ij})$ is invertible for any coordinate system $(q^i,\dot{q}^i)$. If $\omega_L$ is a symplectic form on $TQ$, the Lagrangian function $L$ on $TQ$ is said to be \emph{regular}. Otherwise $L$ is said to be \emph{singular}, \emph{irregular} or \emph{degenerate}.

The \emph{energy} associated to the Lagrangian $L$ on $TQ$ is the function $E_L:TQ\to \RR$ given by
\begin{equation}
	E_L=\Delta(L)-L,
\end{equation}
with $\Delta$ the Liouville vector field. Suppose that $L$ is regular, then
\begin{equation}
	\contr{\xi_L} \omega_L =\dd E_L \label{intrinsic_Euler_Lagrange}
\end{equation}
is the intrinsic expression of the Euler-Lagrange equations. In other words, the trajectories of the Lagrangian system are given by the integral curves of $\xi_L$. 
The vector field $\xi_L$ on $TQ$ is called the \emph{Euler-Lagrange vector field}.
Since $L$ is regular, $\omega_L$ is symplectic and equation \eqref{intrinsic_Euler_Lagrange} admits a unique solution $\xi_L$. Moreover, this vector field $\xi_L$ on $TQ$ is a SODE \cite{leon_methods_1989}.
Therefore it is locally given by
\begin{equation}
	\xi_L=\dot{q}^i \frac{\partial  } {\partial q^i} + \xi^i_L \frac{\partial  } {\partial \dot{q}^i},
\end{equation}
which combined with equation \eqref{symplectic_form_Lagrangian_coords} yields
\begin{equation}
	\contr{\xi_L} \omega_L
	= \left( \frac{\partial^2 L  } {\partial q^j \partial \dot{q}^i  } \dot{q}^i
	-\frac{\partial^2 L  } {\partial q^i \partial \dot{q}^j  } \dot{q}^i
	-\frac{\partial^2 L  } {\partial \dot{q}^i \partial \dot{q}^j  } \xi^i
	 \right)\dd q^j
	 +\frac{\partial^2 L  } {\partial \dot{q}^i \partial \dot{q}^j  } \dot{q}^i \dd\dot{q}^j.
\end{equation}
On the other hand,
\begin{equation}
	E_L = \dot{q}^i \frac{\partial L} {\partial q^i} - L,
\end{equation}
so
\begin{equation}
	\dd E_L
	= \left( \frac{\partial^2 L  } {\partial q^j \dot{q}^i }\dot{q}^i 
	- \frac{\partial L} {\partial q^j}  \right) \dd q^j
	+\frac{\partial^2 L  } {\partial \dot{q}^i \dot{q}^j} \dot{q}^i \dd \dot{q}^j. 
\end{equation}
Then equation \eqref{intrinsic_Euler_Lagrange} holds if and only if
\begin{equation}
	\frac{\partial^2 L  } {\partial q^i \partial \dot{q}^j  } \dot{q}^i
	+\frac{\partial^2 L  } {\partial \dot{q}^i \partial \dot{q}^j  } \xi^i
	=\frac{\partial L} {\partial q^j},\quad 1\leq j\leq n=\dim Q.
\end{equation}
Let $\sigma:I\subset \RR\to Q$ be a solution of $\xi_L$, with $\sigma(t)=(q^i(t))$ . Then 
\begin{equation}
	\frac{\partial^2 L  } {\partial q^i \partial \dot{q}^j  } \dot{q}^i
	+\frac{\partial^2 L  } {\partial \dot{q}^i \partial \dot{q}^j  } \ddot{q}^i
	=\frac{\partial L} {\partial q^j}, \quad 1\leq j\leq m,
\end{equation}
that is,
\begin{equation}
	\frac{\mathrm{d} } {\mathrm{d}t} \left( \frac{\partial L} {\partial \dot{q}^i}  \right) - \frac{\partial L} {\partial q^i} = 0 ,\quad 1\leq i\leq n,
\end{equation}
which are the well-known Euler-Lagrange equations.

Observe that
\begin{equation}
	\xi_L (E_L) = \contr{\xi_L} \dd E_L = \contr{\xi_L} \contr{\xi_L} \omega_L = 0,
\end{equation}
in other words, the energy is conserved along the evolution of the Lagrangian system. As it was discussed in Section \ref{section_Hamiltonian}, non-conservative systems can be described by non-autonomous Lagrangians, external forces or non-symplectic geometries.

\begin{example}\label{example_Lagrangian}
Consider a Lagrangian system on $TQ$ with Lagrangian function of the usual form
\begin{equation}
	L = \frac{1}{2}g_{ij} \dot{q}^i \dot{q}^j-V(q),
\end{equation}
where $q=(q^1,\ldots, q^n)$ and $(g_{ij})$ a positive-definite symmetric matrix. Then
\begin{equation}
	\Delta(L) = \frac{1}{2}\Delta \left( g_{ij} \dot{q}^i \dot{q}^j  \right)
	 = g_{ij} \dot{q}^i \dot{q}^j,
\end{equation}
so the energy of the system is
\begin{equation}
	E_L = \frac{1}{2}g_{ij} \dot{q}^i \dot{q}^j+V(q).
\end{equation}
The Poincaré-Cartan 1-form is
\begin{equation}
	\alpha_L = g_{ij} \dot{q}^i \dd q^j,
\end{equation}
and the Poincaré-Cartan 2-form is
\begin{equation}
	\omega_L
	=  g_{ij} \dd q^i\wedge \dd \dot{q}^j.
\end{equation}
\end{example}

Consider a mechanical system with configuration space $Q$, space of velocities $TQ$ and phase space $T^*Q$.
The \emph{Legendre transform} is a mapping $\Leg: TQ\to T^*Q$ such that the diagram

\begin{center}
\begin{tikzcd}
TQ \arrow[rr, "\Leg"] \arrow[rd, "\tau_q"'] &   & T^*Q \arrow[ld, "\pi_Q"] \\
                                                          & Q &                      
\end{tikzcd}
\end{center}
commutes. Here $\tau_q$ and $\pi_Q$ are the canonical projections on $Q$.
Locally,
\begin{equation}
	\Leg: (q^i,\dot{q}^i) \mapsto (q^i,p_i),
\end{equation}
with $p_i=\partial L/\partial \dot{q}^i$. Moreover, 
\begin{equation}
	\Leg^* \theta =\alpha_L,
\end{equation}
where $\theta$ and $\alpha_L$ is the canonical 1-form on $T^*Q$ and the Poincaré-Cartan 1-form on $TQ$, respectively. Thus
\begin{equation}
	\Leg^* \omega =\omega_L.
\end{equation}
Clearly, $\omega_L$ is symplectic if and only if $\Leg$ is a symplectic map. Then the following result holds.

\begin{theorem} The following assertions are equivalent:
\begin{enumerate}
\item $\omega_L$ is a symplectic form on $TQ$,
\item $L$ is a regular Lagrangian,
\item $\Leg:TQ\to T^*Q$ is a local diffeomorphism.
\end{enumerate}
\end{theorem}
In general, $L$ being regular does not imply $\Leg$ to be a global diffeomorphism. If $\Leg$ is globally a diffeomorphism, then $L$ is said to be \emph{hyperregular}.

In what follows, let us assume $\Leg$ to be hyperregular. Consider the vector field $\bar{\xi}$ on $T^*Q$ given by
\begin{equation}
	\bar{\xi} 
	= \Leg_* \xi_L
\end{equation}
with $\xi_L$ the Euler-Lagrange vector field on $TQ$.
Then
\begin{equation}
\begin{aligned}
	\contr{\bar{\xi}}\omega
	&= \contr{\bar{\xi}}\left(\left(\Leg^{-1}\right)^*\omega\right)
	= \contr{\left(\Leg^{-1}\right)_*\xi}\left(\left(\Leg^{-1}\right)^*\omega\right)
	= \left(\Leg^{-1}\right)^* \left(\contr{\xi_L} \omega_L\right)\\
	&= \left(\Leg^{-1}\right)^* \left(\dd E_L\right)
	= \dd\left(E_L\circ \Leg^{-1}\right).
\end{aligned}
\end{equation}
Therefore $\sigma$ is an integral curve of $\xi_L$ if and only if $\gamma=\Leg\circ \sigma$ is an integral curve of $\bar{\xi}$. In addition, 
\begin{equation}
	\pi_Q \circ \gamma =\tau_Q \circ \sigma.
\end{equation}
The vector field $\bar{\xi}$ is the Hamiltonian vector field for the Hamiltonian function $H=E_L\circ \Leg^{-1}$ on $T^*Q$.

\begin{example}
Consider the Lagrangian system from Example \ref{example_Lagrangian}. The corresponding Legendre transform is given by
\begin{equation}
	\Leg:\left( q^i,\dot{q}^i  \right)
	\mapsto
	\left( q^i, g_{ij}\dot{q}^j  \right),
\end{equation}
so its inverse is
\begin{equation}
	\Leg^{-1} \left( q^i,p_i  \right)
	\mapsto
	\left( q^i, g^{ij} p_j  \right),
\end{equation}
with $(g^{ij})$ the inverse matrix of $(g_{ij})$. Then
the Hamiltonian function of the system is given by
\begin{equation}
	H = \frac{1}{2}g_{ij} g^{ik} g^{jl} p_k p_l + V(q)
	  = \frac{1}{2} g^{kl} p_k p_l + V(q).
\end{equation}

\end{example}
\subsection{Symmetries and constants of the motion}

In this subsection the different types of symmetries in the Lagrangian formalism are discussed, as well as their corresponding constants of the motion. The main references are \cite{arnold_mathematical_1978,Crampin_1983,leon_methods_1989}. See also references \cite{prince_toward_1983,prince_complete_1985,roman-roy_summary_2020,sarlet_note_1983,sarlet_generalizations_1981,noether_invariant_1971,kosmann-schwarzbach_noether_2011,carinena_new_1989,sarlet_note_1983,noether_invariant_1971,kosmann-schwarzbach_noether_2011,lunev_analogue_1990,sarlet_generalizations_1981,neeman_impact_1999,leon_classification_1994,de_leon_symmetries_1995}. 

Let $L$ be a Lagrangian function on $TQ$, with $\alpha_L$ and $\omega_L$ the corresponding Poincaré-Cartan 1-form and 2-form, respectively. Let $\xi_L$ be the corresponding Euler-Lagrange vector field. Consider a function $F$ on $TQ$. Observe that
\begin{equation}
	\xi_L(F) 
	= \contr{\xi_L} \dd F  
	= \contr{\xi_L} \contr{X_F} \omega_L
	= - \contr{X_F} \contr{\xi_L}  \omega_L
	= - \contr{X_F} \dd E_L
	= - X_F(E_L),
\end{equation}
with $X_F$ the Hamiltonian vector field on $TQ$ associated to $F$, given by
\begin{equation}
	\contr{X_F} \omega_L = \dd F.
\end{equation}
Recalling the discussion about constants of the motion in Section \ref{section_Hamiltonian}, one can see that the following statements are equivalent:
\begin{enumerate}
	\item $F$ is a constant of the motion,
	\item $\xi_L(F) = 0$,
	\item $X_F(E_L) = 0$.
\end{enumerate}

 Consider a vector field $X$ on $Q$. Then
\begin{equation}
\begin{aligned}
	\left( \contr{\xi_L} \omega_L  \right) (X^c)
	&= \omega_L (\xi_L,X^c)
	= -\dd \alpha_L (\xi_L, X^c)
	= - \xi_L (\alpha_L(X^c)) 
		+ X^c (\alpha_L(\xi_L))
		+\alpha_L \left( [\xi_L, X^c]  \right).
\end{aligned}
\end{equation}
Now,
\begin{equation}
	\alpha_L(\xi_L) 
	= (S^*\dd L) (\xi_L)
	= \contr{\xi_L} (S^*\dd L)
	= \contr{S \xi_L} \dd L
	= \contr{\Delta} \dd L
	= \Delta(L),
\end{equation}
using equation~\eqref{SODE_Liouville} and the fact $\xi_L$ is a SODE. Similarly,
\begin{equation}
 	\alpha_L(X^c) 
 	= \contr{S X^c} \dd L 
 	= \contr{X^v} \dd L 
 	= X^v(L).
\end{equation}
In addition, $[\xi_L, X^c]$ is a vertical vector field
(which can be easily checked, for instance, using the local expressions), and hence $S[\xi_L, X^c]$ vanishes. Therefore
\begin{equation}
	\left( \contr{\xi_L} \omega_L  \right) (X^c)
	=-\xi_L (X^v(L)) + X^c(\Delta(L)).
\end{equation}
On the other hand, 
\begin{equation}
	\left( \contr{\xi_L} \omega_L  \right) (X^c)
	= \dd E_L(X^c)
	= X^c(E_L)
	= X^c (\Delta(L) - L),
\end{equation}
and thus
\begin{equation}
	\xi_L(X^v(L)) = X^c(L).
\end{equation}
In particular, the left-hand side vanishes if and only if the right-hand side does. In other words, one has the following result.
\begin{theorem}[Noether's theorem]
Let $X$ be a vector field on $Q$. Then $X^v(L)$ is a constant of the motion if and only if $X^c(L)$ vanishes.
\end{theorem}

A vector field $X$ on $Q$ is called a \emph{symmetry of the Lagrangian $L$} if $X^c(L)=0$. In fact, if $X$ generates a local one-parameter group $\Phi_t$ of transformations on $M$, then $X^c$ generates $T\Phi_t$. Therefore $X$ is a symmetry of $L$ if and only if
\begin{equation}
	(T\Phi_t)^* L = L.
\end{equation}

\begin{example}[Conservation of momentum]
	Consider a Lagrangian system on $T\RR^3$ with Lagrangian function $L$. Let $(x,y,z)$ be the Cartesian coordinates on $\RR^3$ and $(x,y,z,\dot{x},\dot{y},\dot{z})$ the induced coordinates on $T\RR^3$. Translations along the $z$-direction are generated by the vector field
	\begin{equation}
		X = \frac{\partial  } {\partial z}
	\end{equation}
	on $\RR^3$. Then the vertical and complete lifts of $X$ are
	\begin{equation}
		X^v = \frac{\partial  } {\partial \dot{z}},
	\end{equation}
	and
	\begin{equation}
		X^c = \frac{\partial  } {\partial z},
	\end{equation}
	respectively.
	 As it is well-known, the $z$-th component of momentum is conserved if and only if $L$ is invariant under translations along the $z$-direction. Indeed,
	\begin{equation}
		X^c(L)=\frac{\partial L} {\partial z}
	\end{equation}
	vanishes if and only if
	\begin{equation}
		X^v(L)=\frac{\partial L} {\partial \dot{z}}=p_z
	\end{equation}
	is a constant of the motion.
\end{example}

\begin{example}[Conservation of angular momentum]
	Consider a Lagrangian system on $T\RR^3$ with Lagrangian function $L$. Let $(x,y,z)$ be the Cartesian coordinates on $\RR^3$ and $(x,y,z,\dot{x},\dot{y},\dot{z})$ the induced coordinates on $T\RR^3$. Rotations around the $z$-direction are generated by the vector field
	\begin{equation}
		X = x \frac{\partial  } {\partial y} - y \frac{\partial  } {\partial x}
	\end{equation}
	on $\RR^3$. Then
	\begin{equation}
		X^v = x \frac{\partial  } {\partial \dot{y}} - y \frac{\partial  } {\partial\dot{x}},
	\end{equation}
	and
	\begin{equation}
		X^c = x \frac{\partial  } {\partial y} - y \frac{\partial  } {\partial x}
		  +\dot{x} \frac{\partial  } {\partial \dot{y}} - \dot{y} \frac{\partial  } {\partial\dot{x}},
	\end{equation}
	so
	\begin{equation}
		X^v(L) = x \frac{\partial L } {\partial \dot{y}} - y \frac{\partial L } {\partial\dot{x}} 
		= x p_y - y p_x= L_z
	\end{equation}
	is conserved if and only if $X^c(L)=0$.
\end{example}

\begin{example}[Lagrangian for a relativistic particle and Lorentz invariance]
Consider a free particle in Minkowski spacetime, that is, $Q=\RR^4$ endowed with the metric
\begin{equation}
	\eta=\eta_{\mu\nu} \dd x^{\mu}\otimes \dd x^\nu,\quad
	\eta_{\mu\nu}= \diag \left( -1,1,1,1  \right),
\end{equation}
in Cartesian coordinates $(x^\mu)=(t,x,y,z)$.
The Lagrangian is then
\begin{equation}
	L = \sqrt{\eta_{\mu\nu} \dot{x}^\mu \dot{x}^\nu},
\end{equation}
so
\begin{equation}
	\alpha_L = \frac{1}{\sqrt{\eta_{\alpha\beta} \dot{x}^\alpha \dot{x}^\beta}} \eta_{\mu\nu} \dot{x}^\mu \dd x^\nu,
\end{equation}
and
\begin{equation}
	\omega_L = \frac{1}{\sqrt{\eta_{\alpha\beta} \dot{x}^\alpha \dot{x}^\beta}} \eta_{\mu\nu} \dd x^\mu \wedge \dd \dot{x}^\nu.
\end{equation}
Consider the following vector fields on $Q$:
\begin{equation}
	X_1 = t \frac{\partial  } {\partial x} + x \frac{\partial  } {\partial t},\quad
	X_2 = t \frac{\partial  } {\partial y} + y \frac{\partial  } {\partial t}, \quad
	X_3 = t \frac{\partial  } {\partial z} + z \frac{\partial  } {\partial t}, 
\end{equation}
which generate the boosts along the $x$-, $y$- and $z$-directions, respectively.
These vector fields are symmetries of the Lagrangian, and their associated conserved quantities are
\begin{equation}
	X_1^v(L) 
	= \frac{1}{\sqrt{\eta_{\alpha\beta} \dot{x}^\alpha \dot{x}^\beta}} (t\dot{x}-x\dot{t}),\quad
	X_2^v(L) 
	= \frac{1}{\sqrt{\eta_{\alpha\beta} \dot{x}^\alpha \dot{x}^\beta}} (t\dot{y}-y\dot{t}),\quad
	X_3^v(L) 
	= \frac{1}{\sqrt{\eta_{\alpha\beta} \dot{x}^\alpha \dot{x}^\beta}} (t\dot{z}-z\dot{t}).
\end{equation}

\end{example}

There are other types of infinitesimal symmetries on $Q$, that is, vector fields on $Q$ which generate symmetries of the dynamics. A \emph{Lie symmetry} is a vector field $X$ on $Q$ such that
\begin{equation}
	[X^c,\xi_L] = 0.
\end{equation}
Since $\omega_L$ is non-degenerate, $X$ is a Lie symmetry if and only if
\begin{equation}
\begin{aligned}
	0&=\contr{[X^c,\xi_L]}\omega_L
	= \liedv{X^c} (\contr{\xi_L} \omega_L) - \contr{\xi_L} (\liedv{X^c} \omega_L)
	= \liedv{X^c} (\contr{\xi_L} \omega_L) + \contr{\xi_L}(\liedv{X^c} \dd \alpha_L)\\
	&= \liedv{X^c} (\dd E_L) + \contr{\xi_L}\dd (\liedv{X^c} \alpha_L)
	=  \dd (\liedv{X^c} E_L) + \contr{\xi_L}\dd (\liedv{X^c} \alpha_L)
	=  \dd ({X^c} (E_L)) + \contr{\xi_L}\dd (\liedv{X^c} \alpha_L).
\end{aligned}
\end{equation}
In particular, $X$ is a Lie symmetry if 
\begin{equation}
	\dd (\liedv{X^c} \alpha_L) = 0,
\end{equation}
and
\begin{equation}
	\dd ({X^c} (E_L)) = 0.
\end{equation}
This can be restated as follows.
\begin{proposition}
 Let $X$ be a vector field on $Q$ such that $\liedv{X^c}\alpha_L$ is a closed 1-form. Then $X$ is a Lie symmetry if and only if $X^c(E_L)$ is a constant function.
\end{proposition}
In particular, if $\liedv{X^c}\alpha_L$ is exact (that is, there exists a function $f$ on $TQ$ such that
$\liedv{X^c}\alpha_L = \dd f$)
and $X^c(E_L)=0$, then $X$ is said to be a \emph{Noether symmetry}. Obviously, every Noether symmetry is also a Lie symmetry. Observe that
\begin{equation}
\begin{aligned}
  \liedv{X^c}\alpha_L&=\contr{X^c}(\dd \alpha_L)+\dd (\contr{X^c}\alpha_L)=\contr{X^c}(\dd \alpha_L)+\dd (\contr{X^c}S^* \dd L)
  \\
  &=\contr{X^c}(\dd \alpha_L)+\dd (\contr{SX^c}\dd L)=\contr{X^c}(\dd \alpha_L)+\dd (X^v (L)),
\end{aligned}
\end{equation}
so
\begin{equation}
	\contr{X^c}(\dd \alpha_L) = \dd f - \dd (X^v (L)),
\end{equation}
and thus
\begin{equation}
  \contr{\xi_{L}}\contr{X^c}(\dd \alpha_L)=\contr{\xi_{L}}(\dd (f-X^vL))=\xi_{L}(f-X^v(L)),
\end{equation}
but
\begin{equation}
  \contr{\xi_{L}}\contr{X^c}(\dd \alpha_L)
  = - \contr{X^c} \contr{\xi_{L}}(\dd \alpha_L)
  =\contr{X^c}\contr{\xi_{L}}\omega_L=\contr{X^c}(\dd E_L)=X^c(E_L).
\end{equation}
This leads to the following result.
\begin{proposition}
Let $X$ be a vector field on $Q$ such that
\begin{equation}
  \liedv{X^c}\alpha_L=\dd f,
\end{equation}
then $X$ is a Noether symmetry if and only if $f-X^v(L)$ is a constant of the motion. 
\end{proposition}
Observe that Noether symmetries are a generalization of symmetries of the Lagrangian. Indeed, if $X$ is a Noether symmetry and $f$ is a constant function, clearly $X^v(L)$ is a constant of the motion . In addition, $\liedv{X^c}\alpha_L$ vanishes so
\begin{equation}
\begin{aligned}
  0&=(\liedv{X^c}\alpha_L)(\xi_{L})=X^c\left(\alpha_L(\xi_{L})\right)-\alpha_L ([X^c,\xi_{L}])\\
  &=X^c(\Delta (L))-S([X^c,\xi_{L}])(L)=X^c(\Delta (L)),
  \end{aligned}
\end{equation}
and thus,
\begin{equation}
  0=X^c(E_L)=X^c(\Delta (L)) -X^c(L)=-X^c(L),
\end{equation}
so a symmetry of $L$ is also a Noether symmetry (and a Lie symmetry).

The infinitesimal symmetries just discussed are vector fields on $Q$, sometimes called point-like symmetries \cite{leon_classification_1994,de_leon_symmetries_1995}. One can also consider infinitesimal symmetries on $TQ$, in other words, symmetries that depend on the velocities as well as the positions. A vector field $\tilde{X}$ on $TQ$ is called a \emph{dynamical symmetry} if
\begin{equation}
	[\tilde{X},\xi_L] = 0.
\end{equation}
A \emph{Cartan symmetry} is a vector field $\tilde{X}$ on $TQ$ such that
\begin{equation}
	\liedv{\tilde{X}}\alpha_L = \dd f,
\end{equation}
and
\begin{equation}
	\tilde{X}(E_L) = 0.
\end{equation}
Recalling the deductions above for point-like symmetries, the reader can analogously derive the following results.
\begin{proposition}
 Let $\tilde{X}$ be a vector field on $TQ$ such that $\liedv{\tilde{X}}\alpha_L$ is a closed 1-form. Then $\tilde{X}$ is a dynamical symmetry if and only if $\tilde{X}(E_L)$ is a constant function.
\end{proposition}

\begin{proposition}
Let $\tilde{X}$ be a vector field on $TQ$ such that
\begin{equation}
  \liedv{\tilde{X}}\alpha_L=\dd f,
\end{equation}
then $\tilde X$ is a Cartan symmetry if and only if $f-(S\tilde{X})(L)$ is a constant of the motion. 
\end{proposition}
Notice that there exists no analogous for Lagrangian symmetries on $TQ$, since $[\xi_L,\tilde{X}]$ is not a vertical vector field for a general vector field $\tilde{X}$ on $TQ$. Furthermore, it is clear that
\begin{enumerate}
\item $X$ is a Lie symmetry if and only if $X^c$ is a dynamical symmetry,
\item $X$ is a Noether symmetry if and only if $X^c$ is a Cartan symmetry,
\item every Cartan symmetry is also a dynamical symmetry,
\end{enumerate}
with  $X$ a vector field on $Q$.

\begin{example}[2-dimensional harmonic oscillator]
 In this case the configuration space is $Q=\RR^2$, so $TQ\simeq \RR^2 \times \RR^2$. The Lagrangian function of the harmonic oscillator is
 \begin{equation}
 	L = \frac{1}{2} \left[ (\dot{q}^1)^2 + (\dot{q}^2)^2 
 	- (\Omega_1 q^1)^2- (\Omega_2 q^2)^2 \right],
 \end{equation}
 in units where the mass is equal to 1 for the sake of simplicity. Then the Poincaré-Cartan 1-form is
 \begin{equation}
 	\alpha_L =  \dot{q}^1 \dd q^1 + \dot{q}^2 \dd q^2,
 \end{equation}
 so the corresponding 2-form is
 \begin{equation}
 	\omega_L = \dd q^1 \wedge \dd \dot{q}^1 + \dd q^2 \wedge \dd \dot{q}^2.
 \end{equation}
 Moreover, the energy of the system is
 \begin{equation}
 	E_L =\frac{1}{2} \left[ (\dot{q}^1)^2 + (\dot{q}^2)^2 
 	+ (\Omega_1 q^1)^2 + (\Omega_2 q^2)^2 \right].
 \end{equation}
 Consider the following vector fields on $TQ$:
 \begin{equation}
	\begin{aligned}
	X_{1} &=\frac{\Omega_{1}}{\left(\Omega_{1}\right)^{2}\left(q^{1}\right)^{2}+\left(\dot{q}^{1}\right)^{2}}\left(q^{1} \frac{\partial}{\partial q^{1}}+\dot{q}^{1} \frac{\partial}{\partial \dot{q}^1}\right), \\
	X_{2} &=\frac{\Omega_{2}}{\left(\Omega_{2}\right)^{2}\left(q^{2}\right)^{2}+\left(\dot{q}^{2}\right)^{2}}\left(q^{2} \frac{\partial}{\partial q^{2}}+\dot{q}^{2} \frac{\partial}{\partial \dot{q}^{2}}\right).
	\end{aligned}
 \end{equation}
One can check that
\begin{equation}
	\liedv{X_a} \omega_L = 0,
\end{equation}
and
\begin{equation}
	X_a (E_L) = \Omega_a,
\end{equation}
where $a=1,2$, so $X_1$ and $X_2$ are dynamical symmetries.
See reference \cite[Example 1]{roman-roy_summary_2020} for the Hamiltonian counterpart of this example.
\end{example}

\begin{example}[Conservation of energy]
	Consider a Lagrangian system on $TQ$ with Lagrangian function $L$ and Euler-Lagrange vector field $\xi_L$. Observe that
	\begin{equation}\begin{aligned}
			\liedv{\xi_L}\alpha_L 
			&= \contr{\xi_L}\dd \alpha_L + \dd \left( \contr{\xi_L} \alpha_L  \right)
			= -\contr{\xi_L}\omega_L + \dd \left( \contr{\xi_L} S^*(\dd L)  \right)
			= -\dd E_L + \dd \left( \contr{S\xi_L} (\dd L)  \right)\\
			&= -\dd E_L + \dd \left((S\xi_L) (L)  \right)
			= -\dd(E_L - \Delta(L))
			= \dd L,
	\end{aligned}\end{equation}
	and
	\begin{equation}
		\xi_L (E_L) = \contr{\xi_L} \dd E_L 
		= \contr{\xi_L} \contr{\xi_L} \omega_L =0,
	\end{equation}
	so $\xi_L$ is a Cartan symmetry and
	\begin{equation}
		L - (S\xi_L) (L) = L - \Delta(L) = E_L
	\end{equation}
	is the associated conserved quantity.
\end{example}

The types of symmetries for Lagrangian systems on $TQ$ with Lagrangian function $L$ are summarized in Figure \ref{figure_diagram_conservative}.

\begin{figure}[t]
\centering
\includegraphics[width=.5\linewidth]{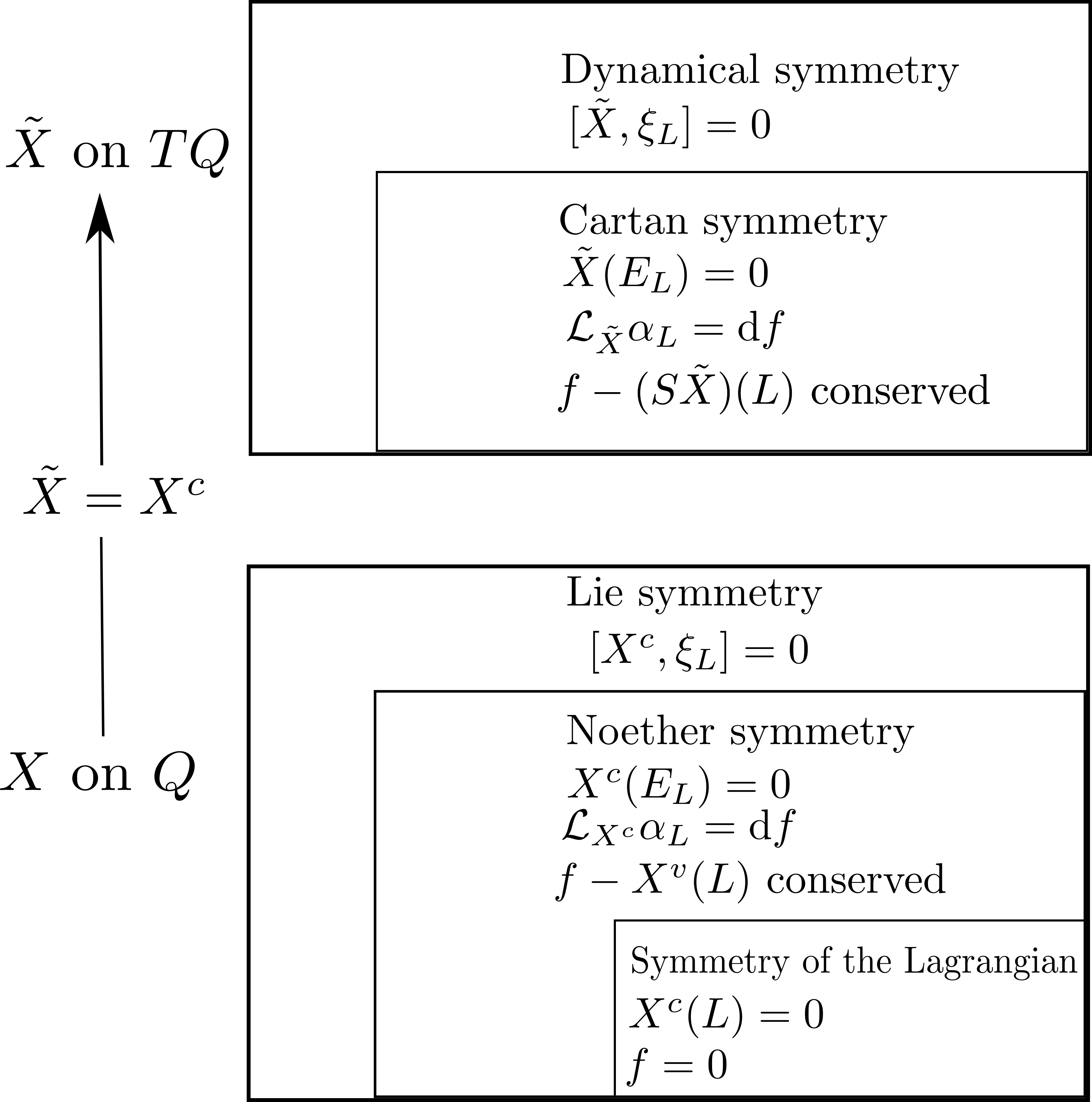}
\caption{Types of symmetries on $Q$ and on $TQ$. The complete lifts of Lie symmetries and Noether symmetries correspond to dynamical symmetries and Cartan symmetries, respectively. Noether symmetries are the subset of Lie symmetries such that $\liedv{X^c}\alpha_L=\dd f$ and $f-X^v(L)$ is a constant of the motion. Cartan symmetries are the analogous subset of dynamical symmetries.
The symmetries of the Lagrangian system are the subset of Noether symmetries for which $\liedv{X^c}\alpha_L$ vanishes, so $f$ is a constant function that can be taken as 0 without loss of generality. 
}
\label{figure_diagram_conservative}
\end{figure}

Finally, symmetries and constants of the motion in the Lagrangian and Hamiltonian formalisms can be related. Recall that $L$ is assumed to be hyperregular.
Let $H\circ \Leg = E_L$ be a Hamiltonian on $T^*Q$ and $X_H$ its corresponding Hamiltonian vector field. Let $\tilde{X}$ be a vector field on $TQ$ and $\hat{X}$ the $\Leg$-related vector field on $T^*Q$, that is,
\begin{equation}
 	\hat{X} = \Leg_* \tilde{X}.
 \end{equation} 
It can be easily seen that the following statements hold:
\begin{enumerate}[label=\roman*)]
\item $\hat{X}$ commutes with $X_{H}$ if and only if $\tilde{X}$ is a dynamical symmetry of $L$.
\item $\liedv{\hat{X}}{\theta}=\dd f$ if and only if $\liedv{\tilde{X}}\alpha_L=\dd g$, where $g=f\circ \Leg$.
\item Suppose that $\liedv{\hat{X}}{\theta}=\dd f$. Then the following assertions are equivalent:
\begin{enumerate}[label=\alph*)]
\item $\hat{X}(H)=0$,
\item $f-{\theta}(\hat{X})$ is a conserved quantity,
\item $\tilde{X}(E_L)=0$,
\item $f\circ \Leg-\alpha_L(\tilde{X})$ is a conserved quantity.
\end{enumerate}

\end{enumerate}

\section{Momentum map and symplectic reduction}\label{section_reduction}
In the previous section, the relation between infinitesimal symmetries (and hence, of one-parameter groups of symmetries) with constants of the motion has been studied. These conserved quantities are frequently used to simplify the equations of motion. As it is well-known, many physical systems have a group of symmetries (which is not in general a one-parameter group). For instance, special-relativistic systems are invariant under the Poincaré group. One could then think of simplifying the equations of the motion by acting with the whole symmetry group. 
In fact, if a symmetry group acts over a system, then its number of independent degrees of freedom is reduced. In other words, the dimensions of $Q$, and hence of $TQ$ and $T^*Q$, are reduced. 
A systematic procedure to explore this fact is known as \emph{reduction}, which is due to Marsden and Weinstein. See references \cite{marsden_reduction_1974,marsden_reduction_1990,abraham_foundations_2008,ortega_momentum_2004}. See also the seminal works by Lie, Kostant, Souriau and Smale 
\cite{souriau_structure_2008,souriau_67,smale_topology_1970,merker_theory_2010,kostant_orbits_2009}.

Consider a Lie group $G$ with Lie algebra $\mathfrak{g}$, and let $\mathfrak{g}^*$ be the dual of $\mathfrak{g}$. Let $(M,\omega)$ be a connected symplectic manifold. The action $\Phi:G\times M\to M$ of $G$ on $M$ is called \emph{symplectic} if the map
\begin{equation}
\begin{aligned}
	\Phi_g: M&\to M\\
	x&\mapsto \Phi(g,x)=gx
\end{aligned}
\end{equation}
is symplectic for each $g\in G$. 
In what follows, every group action is assumed to be free and proper. 

For each $\xi\in \mathfrak{g}$, there is a vector field $\xi_M$ on $M$ which is the infinitesimal generator of the action corresponding to $\xi$.
Consider a mapping
\begin{equation}
	J:M\to \mathfrak{g}^*,
\end{equation}
and a function $J^\xi:M\to \RR$ such that
\begin{equation}
	J^\xi (x) = \left\langle  J(x), \xi \right\rangle
\end{equation}
for each $\xi\in \mathfrak{g}$ and each $x\in M$. Then $J$ is called a \emph{momentum map} for the action if
\begin{equation}
	\dd J^\xi = \contr{\xi_M} \omega
\end{equation} 
for every $\xi\in \mathfrak{g}$. In other words, $J$ is a momentum map provided that, for each $\xi\in \mathfrak{g}$, $\xi_Q^c$ is the Hamiltonian vector field on $M$ associated to $J^\xi$.

The \emph{affine action} of $G$ on $\mathfrak{g}^*$ associated with the momentum map $J:M\to \mathfrak{g}^*$ is given by
\begin{equation}
\begin{aligned}
G\times \mathfrak{g}^* & \to \mathfrak{g}^*\\
(g,\mu) & \mapsto \operatorname{Ad}_{g^{-1}}^* \mu 
		+ J  \left( \Phi_g(x)  \right)
		- \operatorname{Ad}_{g^{-1}}^* \left( J(x)  \right).
\end{aligned}
\end{equation}
This definition is independent of the choice of $x\in M$. The isotropy group of $\mu\in \mathfrak{g}$ under the affine action is denoted by $G_\mu$.

Consider a function $F$ on $(M,\omega)$ which is invariant under the action of $G$, that is,
\begin{equation}
	F(x) = F \left( \Phi_g(x)  \right)
\end{equation}
for every $g\in G$ and every $x\in M$. Equivalently,
\begin{equation}
	\liedv{\xi_P} F = 0,
\end{equation}
for each $\xi\in \mathfrak{g}$, that is,
\begin{equation}
	\left\{F, J^\xi\right\} = 0,
\end{equation}
where $\left\{\cdot,\cdot  \right\}$ is the Poisson bracket defined by $\omega$,
so $J^\xi$ is a first integral of $X_F$. In practice, for a Hamiltonian or Lagrangian invariant under $G$, there is a constant of the motion $J^\xi$ for each $\xi\in \mathfrak{g}$.

A momentum map $J$ is called \emph{$\mathrm{Ad}^*$-equivariant} if
\begin{equation}
	J \left( \Phi_g(x)  \right) = \mathrm{Ad}_{g^{-1}}^* J(x)
\end{equation}
for each $g\in G$ and each $x\in M$. In other words, the following diagram commutes:
\begin{center}
\begin{tikzcd}
M \arrow[rr, "\Phi_g"] \arrow[d, "J"']                     &  & M \arrow[d, "J"] \\
\mathfrak{g}^* \arrow[rr, "\operatorname{Ad}^*_{g^{-1}}"'] &  & \mathfrak{g}^*  
\end{tikzcd}
\end{center}

An $\mathrm{Ad}^*$-equivariant momentum map defines an homomorphism between the Lie algebra $\mathfrak{g}$ and the Lie algebra of functions on $M$ under the Poisson bracket. More specifically, the following property holds:
\begin{equation}
	\left\{J^\xi, J^\eta  \right\} = J^{[\xi,\eta]}
\end{equation}
for each $\xi,\eta\in \mathfrak{g}$. 

If the symplectic form can be written as $\omega=-\dd \theta$, and the action leaves $\theta$ invariant (that is, $\Phi_g^*\theta=\theta$ for each $g\in G$), there is a \emph{natural momentum map} given by
\begin{equation}
	J(x) (\xi) = \left( \contr{\xi_M} \theta  \right) (x)
\end{equation}
for each $x\in M$. It can be shown that it is $\mathrm{Ad}^*$-equivariant. 
 In particular, if $\theta$ is the canonical 1-form on $T^*Q$, every group action leaves it invariant by the property \eqref{property_defining_canonical_1_form}, so the natural momentum map can always be defined in the Hamiltonian formalism. In the Lagrangian description, if the Lagrangian function $L$ is $G$-invariant, so is $\alpha_L$
and the natural momentum map can be used as well.

The following theorem is the basis of symplectic reduction.
\begin{theorem}[Symplectic point reduction] \label{theorem_symplectic_point_reduction}
Consider a Lie group $G$ with symplectic action $\Phi$ on the connected manifold $(M,\omega)$.
Then the following statements hold:
\begin{enumerate}
\item The quotient space $M_\mu\coloneqq J^{-1}(\mu)/G_\mu$ is a symplectic manifold, with symplectic form $\omega_\mu$ uniquely characterized by the relation
\begin{equation}
	\pi_\mu^* \omega_\mu = \incl_\mu^* \omega.
\end{equation}
Here the maps $\incl_\mu: J^{-1}(\mu)\hookrightarrow M$ and $\pi_\mu:J^{-1}(\mu)\to J^{-1}(\mu)/G_\mu$ denote the inclusion and the projection, respectively. The pair $(M_\mu,\omega_\mu)$ is called the \emph{symplectic point reduced space}.
\item Let $f$ be a $G$-invariant function on $M$, and $X_f$ the corresponding Hamiltonian vector field on $M$. The flow $F_t$ of $X_f$ induces a flow $F_t^\mu$ on $M_\mu$ given by
\begin{equation}
	\pi_{\mu} \circ F_{t} \circ \incl_{\mu}=F_{t}^{\mu} \circ \pi_{\mu}.
\end{equation}
\item The vector field generated by the flow $F_t^\mu$ on $(M_\mu,\omega_\mu)$ is Hamiltonian with associated reduced function $f_\mu$ on $M_\mu$ defined by
\begin{equation}
	f_\mu \circ \pi_\mu = f \circ \incl_\mu.
\end{equation}
The vector fields $X_f$ and $X_{f_\mu}$ are $\pi_\mu$-related, that is,
\begin{equation}
	\pi_{\mu*} X_{f_\mu} = \incl_{\mu*} X_f.
\end{equation}
\item Let $k$ be another $G$-invariant function. Then $\left\{f,k  \right\}$ is also $G$-invariant and its associated reduced function is given by 
\begin{equation}
	\left\{f,k  \right\}_\mu = \left\{f_\mu,k_\mu  \right\}_{M_\mu},
\end{equation}
where $\left\{\cdot,\cdot  \right\}_{M_\mu}$ denotes the Poisson bracket associated to $\omega_\mu$ on $M_\mu$.

\end{enumerate}

\end{theorem}
The proof is involved and can be found in reference \cite{ortega_momentum_2004}. 

The reduced space has dimensions
\begin{equation}
	\dim M_\mu = \dim J^{-1}(\mu) - \dim G_\mu = \dim M -\dim G - \dim G_\mu. 
\end{equation}


Consider a Lie group $G$ with action $\Phi:G\times Q\to Q$ on $Q$. If $\Phi_g:Q\to Q$ is the diffeomorphism given by $\Phi_g(q)=gq$ for each $g\in G$ and each $q\in Q$, it induces an action on $TQ$ given by
\begin{equation}
	T\Phi_g:TQ\to TQ,
\end{equation}
known as the \emph{lifted action}, which is symplectic. For each $\xi\in \mathfrak{g}$, there is a vector field $\xi_Q$ on $Q$ generating the corresponding infinitesimal action on $Q$. The generator of the lifted action is then $\xi_Q^c$. Consider a Lagrangian system on $Q$ with hyperregular Lagrangian function $L$, and let $\omega_L$ be the symplectic structure induced by $L$ on $TQ$.
As mentioned above, the natural momentum map on $TQ$ is given by
\begin{equation}
	J(\xi)=\alpha_L \left( \xi_Q^c  \right)
\end{equation}
for each $\xi\in \mathfrak{g}$. 

To sum up, if a Lagrangian $L$ on $TQ$ is invariant under the action of a symmetry group $G$ with Lie algebra $\mathfrak{g}$, then:
\begin{enumerate}
\item for each $\xi\in \mathfrak{g}$, there is a conserved quantity $J^\xi=\alpha_L(\xi_Q^c)$,
\item the symplectic point reduced space is $\left((TQ)_\mu, \omega_{\mu}\right)$, with
\begin{equation}
	\pi_\mu^* \omega_\mu = \incl_\mu^* \omega_L;
\end{equation}
\item the reduced Lagrangian function $L_\mu$ is given by
\begin{equation}
	L_\mu \circ \pi_\mu  = L \circ \incl_\mu,
\end{equation}
\item the reduced Euler-Lagrange vector field $\xi_{L_\mu}$ is given by
\begin{equation}
	\pi_{\mu*}\xi_{L_\mu}=\incl_{\mu*} \xi_{L}.
\end{equation}
\end{enumerate}

It is worth remarking that $(TQ)_\mu$ is denoted this way since it is the reduced space associated to $TQ$ but, in general, it is not a tangent bundle. One can also introduce a lifted action on $T^*Q$ and obtain the same result in the Hamiltonian formalism.

\begin{remark}[Reconstruction of the dynamics]\label{remark_reconstruction}

Given the dynamics on $J^{-1}(\mu)\subset TQ$, one can obtain the dynamics on  the reduced space $(TQ)_\mu$ by the procedure above. From a practical point of view, one poses the inverse question, that is, how to determine the integral curves $c(t)$ of $\xi_L$ knowing the integral curves $c_\mu(t)$ of $\xi_{L_\mu}$ and an initial condition $m_0\in J^{-1}(\mu)$. Take a smooth curve $d(t)$ in $J^{-1}(\mu)$ such that $d(0)=m_0$ and $\pi_\mu(d(t))=c_\mu(t)$. If $c(t)$ denotes the integral curve of $\xi_L$ with $c(0)=m_0$, there exists a smooth curve $g(t)$ in $G_\mu\subset G$ such that
\begin{equation}
  c(t)=\Phi_{g(t)}(d(t))\label{integral_curve}.
\end{equation}
The problem is now to determine $g(t)$ and thus $c(t)$.
By definition of integral curve, 
\begin{equation}
\begin{aligned}
  \xi_{L}(c(t))
  &=\dot{c}(t)\\
  &=T\Phi_{g(t)}(d(t))(\dot{d}(t)) + T\Phi_{g(t)}(d(t))(TL_{g(t)^{-1}}(\dot{g}(t)))_Q^c(d(t)),
\end{aligned}
\end{equation}
and using the $\Phi_g$-invariance of $\xi_{L}$ one has
\begin{equation}
  \xi_{L}(d(t))
  =\dot{d}(t)+(TL_{g(t)^{-1}}(\dot{g}(t)))_Q^c(d(t)).
\end{equation}
In order to solve this equation, one first has to solve the algebraic problem
\begin{equation}
  \xi_Q^c(d(t))=\xi_{L}(d(t))-\dot{d}(t),
\end{equation}
for $\xi(t)\in\mathfrak{g}$, and then solve
\begin{equation}
  \dot{g}(t)=TL_{g(t)}\xi(t)
\end{equation}
for $g(t)$. The integral curve sought is given by equation \eqref{integral_curve}.

\end{remark}

\begin{example}[Angular momentum]\label{example_angular_momentum_reduction}
    Consider $Q=\RR^n\setminus \{0\}$, and a Lagrangian $L$ on $T Q$ that is spherically symmetric, say $L(q,\dot{q})= L(\lVert q \rVert,\lVert{\dot{q}}\rVert)$. 
    Consider the Lie group $G=\mathrm{SO}(n) = \{O \in \RR^{n\times n} \mid {O}^t O = \Id, \det(O)=1 \}$ acting by rotations on $Q$. The action can be lifted to $TQ$ via the tangent lift. Explicitly, for $O\in \mathrm{SO}(3)$ let
    \begin{equation}
        \begin{aligned}
            g_O :  Q  &\to  Q, \\
            (q) &\mapsto (O \cdot q)\\
            Tg_O : T Q  &\to T Q, \\
            (q,\dot{q}) &\mapsto (O \cdot q, O \cdot \dot{q}).
        \end{aligned}
    \end{equation}
The group $\mathrm{SO}(n)$ acts freely and properly. The Lie algebra of the group is given by $\mathfrak{g}=\mathfrak{so}(n) = \{o \in \RR^{n\times n} \mid {o}^t + o = 0\}$.
In particular, for $n=3$ the algebra $\mathfrak{g}$ can be identified with the algebra of 3-dimensional vectors with the cross product by taking
\begin{equation}
    \begin{pmatrix}
        0 & -\xi_3 & \xi_2\\
        \xi_3 & 0 & -\xi_1\\
        -\xi_2 & \xi_1 & 0\\
    \end{pmatrix}
    \mapsto 
    \begin{pmatrix}
        \xi_1 \\ \xi_2 \\ \xi_3
    \end{pmatrix},
\end{equation}
so, for each  $\xi\in \mathfrak{g}$, one has
\begin{equation}
\begin{split}
    \xi_Q (q) &= (\xi \times q),\\
    \xi_Q^c(q,\dot{q}) &= (\xi \times q,\xi \times \dot{q}),\\
    \xi_Q^v(q,\dot{q}) &= (0,\xi \times \dot{q}).
\end{split}
\end{equation}
One can identify $\mathfrak{g}$ with $\mathfrak{g}^*$ by using the inner product on $\RR^3$. The moment map is then given by \cite[Example~4.2.15]{abraham_foundations_2008}
\begin{equation}
    J(q,\dot{q}) = q \times \dot{q}. 
\end{equation}
Identifying $\mathfrak{g}^*\simeq \RR^3$ one sees that the affine action of $G$ is the usual one (by rotations). Let $\mu \in \mathfrak{g}^*$, $\mu\neq 0$ then the isotropy group $G_\mu \simeq S^1$ of $\mu$ under the affine action, which are the rotations around the axis $\mu$.

Without loss of generality, one can take $\mu=(0,0,\mu_0)$. Hence, if $(q,\dot{q}) \in J^{-1}(\mu)$, both $q$ and $\dot{q}$ lie on the $xy$-plane. Moreover, they must satisfy the equation $\dot{q}^1 p_2-\dot{q}^1 q^2 = \mu_0$. Finally, applying the reduction method above to the system, one can find out $((TQ)_\mu, L_\mu)$, which is a Lagrangian system over a $2$-dimensional manifold~\cite[Example~4.3.4]{abraham_foundations_2008}.
\end{example}

See reference \cite{marsden_reduction_1990} for explicit expressions for $c(t)$ in some particular cases and their derivation.
\chapter{Mechanical systems subjected to external forces} \label{part_paper}
 A well-known result from elementary calculus and classical mechanics is that a conservative force (that is, a force such that the total work it does moving a particle between two given points is independent of the path taken) can be written as the gradient of a potential. This potential can then be incorporated to the Lagrangian or Hamiltonian function of the system. However, this is not always the case. As a matter of fact, when non-conservative forces act on a mechanical system, its equations of motion depend on both the Lagrangian or the Hamiltonian and on the external forces \cite{goldstein_mecanica_1987}. This is the case of many physical systems of interests \cite{cantrijn_symplectic_1984,cantrijn_vector_1982,cantrijn_geometry_2002,esen_geometrical_2021}.
  Furthermore, external forces can arise in a more sophisticated manner, for instance, after a process of reduction of a nonholonomic system with symmetries \cite{cantrijn_geometry_2002,cantrijn_reduction_1999,cortes_reduction_1999}.

 In this part of the thesis, the concepts studied in Part \ref{part_1} are generalized for systems with external forces. The main reference is our original work \cite{de_leon_symmetries_2021}.


\section{Mechanical systems subjected to external forces}\label{section_external_forces}
Geometrically, an external force is expressed as a semibasic 1-form on a symplectic manifold \cite{leon_methods_1989,godbillon_geometrie_1969,abraham_foundations_2008}. Consider a forced Hamiltonian system on $(T^*Q, \omega)$, with Hamiltonian function $H:M\to \RR$ and external force $\gamma$. If $(q^i)$ are local coordinates on $Q$, and $(q^i,p_i)$ the induced coordinates on $T^*Q$, the external force can be locally written as
\begin{equation}
	\gamma = \gamma_i (q,p) \dd q^i.
\end{equation}
The trajectories of the system are now the integral curves of the vector field $X_{H,\gamma}$ on $T^*Q$ given by
\begin{equation}
	\contr{X_{H,\gamma}} \omega = \dd H + \gamma.
\end{equation}
Let $Z_\gamma$ be the vector field defined by
\begin{equation}
	\contr{Z_\gamma} \omega = \gamma.
\end{equation}
Then $X_{H,\gamma}$ can be written as
\begin{equation}
	X_{H,\gamma} = X_H + Z_\gamma,
\end{equation}
with $X_H$ the Hamiltonian vector field of $H$. 
Locally, 
\begin{equation}
	Z_\gamma=-\gamma_i\parder{}{p_i},
\end{equation}
and recall that
\begin{equation}
	X_H=\parder{H}{p_i}\parder{}{q^i}-\parder{H}{q^i}\parder{}{p_i} \label{Hamiltonian_VF_local},
\end{equation}
so
\begin{equation}
	X_{H,\gamma}=\parder{H}{p_i}\parder{}{q^i}-\left(\parder{H}{q^i}+\gamma_i\right)\parder{}{p_i}.
\end{equation}
The equations of motion are then
\begin{equation}
\begin{aligned}
&\der{q^i}{t}=\parder{H}{p_i},\\
&\der{p_i}{t}=-\left(\parder{H}{q^i}+\gamma_i\right).
\end{aligned}
\end{equation}
Suppose that the external force is exact, that is, $\gamma = \dd f$ for some function $f$ on $Q$. Then $Z_\gamma$ is the Hamiltonian vector field for $f$, so one can define a new Hamiltonian function $H'=H+f$, such that $X_{H,\gamma}=X_{H'}$ is its Hamiltonian vector field. In that case, the problem is reduced to the conservative one (that is, without external forces). Therefore, \emphoriginal{stricto sensu} external forces are those which are non-exact. 

The Lagrangian formalism is analogous. An external force is now given by a semibasic 1-form $\beta$ on $TQ$. In bundle coordinates,
\begin{equation}
 	\beta = \beta_i(q,\dot q) \dd q^i. 
 \end{equation} 
 Hereinafter, a Lagrangian system on $TQ$ with Lagrangian function $L$ and subject to an external force $\beta$ will be called a \emph{forced Lagrangian system} $(L,\beta)$ on $TQ$. The dynamics of $(L,\beta)$ is given by the \emph{forced Euler-Lagrange vector field} $\xi_{L,\beta}$ via
 \begin{equation}
  	\contr{\xi_{L,\beta}} \omega_L = \dd E_L + \beta. \label{dynamics_vector_Lagrangian}
  \end{equation} 
  Recall that $\omega_L$ and $E_L$ are the Poincaré-Cartan 2-form and the energy function associated to $L$, respectively (see Section \ref{section_Lagrangian}).
  Here L is assumed to be hyperregular. Let $\xi_\beta$ be the vector field given by
  \begin{equation}
  	\contr{\xi_\beta} \omega_L = \beta,
  \end{equation}
  so that
  \begin{equation}
  	\xi_{L,\beta} = \xi_L + \xi_\beta,
  \end{equation}
  with $\xi_{L}$ the Hamiltonian vector field on $(TQ,\omega_L)$ associated to $E_L$. Recall that locally
  \begin{equation}
  	\omega_L 
  	= \frac{\partial^2 L  } {\partial \dot{q}^i \partial {q}^j  } \dd q^i \wedge \dd {q}^j
  	+W_{ij} \dd q^i \wedge \dd \dot{q}^j,
  \end{equation}
  where the Hessian matrix
  \begin{equation}
  	\left( W_{ij}  \right) = \left( \frac{\partial^2 L  } {\partial \dot{q}^i \partial \dot{q}^j  }  \right)
  \end{equation}
  is invertible since L is regular. Let
  \begin{equation}
  	\xi_\beta = A^i \frac{\partial  } {\partial q^i} + B^i \frac{\partial  } {\partial \dot q^i},
  \end{equation}
  then
  \begin{equation}
  	\contr{\xi_\beta}\omega_L = A^i \left(\frac{\partial^2 L  } {\partial \dot{q}^i \partial {q}^j  }
  	- \frac{\partial^2 L  } {\partial \dot{q}^j \partial {q}^i  }  \right)	\dd q^j
  	+ A^i W_{ij} \dd \dot q^j
  	- B^j W_{ij} \dd q^i,
  \end{equation}
  so $A^i=0$ and $B^j W_{ij}=-\beta_i$, and thus
  \begin{equation}
  	\xi_\beta = - \beta_i W^{ij} \frac{\partial  } {\partial \dot{q}^i}.
  \end{equation}
Then
\begin{equation}
	\xi_{L,\beta}=\dot{q}^i\parder{}{q^i}+\left(\xi^i-\beta_j W^{ij}\right)\parder{}{\dot{q}^i}.
  \label{local_expression_xi_L_beta}
\end{equation}
Therefore, if $(q^i(t))$ is a solution of the SODE $\xi_{L,\beta}$, it satisfies
\begin{equation}
\ddot{q}^i  =\xi^i-\beta_jW^{ji}.
\end{equation}
Recall that $\xi^i$ is given by
\begin{equation}
	\frac{\partial^2 L  } {\partial q^i \partial \dot{q}^j  } \dot{q}^i
	+W_{ij} \xi^i
	=\frac{\partial L} {\partial q^j},
  \label{equation_xi_i}
\end{equation}
so
\begin{equation}
\ddot{q}^i \frac{\partial ^2 L } {\partial \dot{q}^i \dot q^j}
+\dot{q}^i\frac{\partial ^2 L } {\partial {q}^i \dot q^j} -\parder{L}{q^j}+\beta_k W^{ki}W_{ij}=0,
\end{equation}
 which finally yields
\begin{equation}
\der{}{t}\left(\parder{L}{\dot{q}^i}\right)-\parder{L}{q^i}=-\beta_i. \label{forced_Euler_Lagrange_eq}
\end{equation} 
These are the \emph{forced Euler-Lagrange equations}. They can also be derived from the Lagrange-d'Alembert's principle \cite{goldstein_mecanica_1987,marsden_west_01,lew_variational_2004}.

 Notice that 
  \begin{equation}
    S(\xi_{L,\beta})=S(\xi_\beta)=\Delta,
  \end{equation}
  where $S$ and $\Delta$ are the vertical endomorphism and the Liouville vector field, respectively. Hence, $\xi_{L,\beta}$ is a SODE.

Since $L$ is assumed to be hyperregular, the Legendre transform $\Leg: TQ\to T^*Q$ is a diffeomorphism. Therefore the external force $\beta$ on $TQ$ is related with the external force $\gamma$ on $T^*Q$ by
\begin{equation}
	\beta = \Leg_* \gamma.
\end{equation}
\section[Symmetries and constants of the motion]{Symmetries and constants of the motion in the Lagrangian description} \label{section_Lagrangian_forced}

In this section, the results regarding symmetries and constants of the motion studied in section \ref{section_Lagrangian} are extended for forced Lagrangian systems. There are other approaches that can be found in the previous literature and have some relation with the one followed here. For instance, Cantrijn \cite{cantrijn_vector_1982} considers Lagrangian systems that depend explicitly on time, and defines a 2-form on $\RR\times TQ$ that depends on the Poincaré-Cartan 2-form of the Lagrangian and the semibasic 1-form representing the external force.  Alternatively, van der Schaft \cite{schaft_hamiltonian_1981,van_der_schaft_symmetries_1983} considers a framework steming from system theory, in which an ``observation'' manifold appears together with the usual state space, and obtains a Noether's theorem for Hamiltonian systems in this frame. Other approaches using variational tools can be found in reference~\cite{bahar_extension_1987}. However, in the approach followed here \cite{de_leon_symmetries_2021} no additional structure or objects are introduced besides the proper external force. 



Consider a forced Lagrangian system $(L,\beta)$ with forced Euler-Lagrange vector field $\xi_{L,\beta}$.
If $f$ is a function on $TQ$, then
\begin{equation}
	\frac{\mathrm{d} } {\mathrm{d}t} f(\sigma(t))=\xi_{L,\beta} (f)(\sigma(t)),
\end{equation}
for any integral curve $\sigma(t)$ of $\xi_{L,\beta}$.
Thus $f$ is a \emph{constant of the motion} (or a \emph{conserved quantity}) if it is a first integral of $\xi_{L,\beta}$ (that is, if $\xi_{L,\beta}(f)=0$). Observe that now a constant of the motion will not be, in general, a first integral of $\xi_L$.

Suppose that, for a given coordinate $q^i$, $\partial L/\partial q^i=\beta_i$. Then
\begin{equation}
	\frac{\mathrm{d} } {\mathrm{d}t} \left( \frac{\partial L} {\partial \dot{q}^i}  \right)=0,
\end{equation}
so $p_i=\partial L/\partial q^i$ is a constant of the motion. This motivates generalising Noether's theorem for systems with external forces.


 Consider a vector field $X$ on $Q$. Then, by equation~\eqref{dynamics_vector_Lagrangian},
\begin{equation}
  \begin{aligned}
  (\dd E_L+\beta)(X^c)=& (\contr{\xi_{L,\beta}} \omega_L  ) (X^c)
  = \omega_L (\xi_{L,\beta},X^c)
  =-\dd \alpha_L(\xi_{L,\beta},X^c)\\
  =&-\xi_{L,\beta}(\alpha_L(X^c))+X^c(\alpha_L(\xi_{L,\beta}))+\alpha_L([\xi_{L,\beta},X^c]).
  \end{aligned}
\end{equation}
Now, since $\xi_{L,\beta}$ is a SODE, 
\begin{equation}
\alpha_L(\xi_{L,\beta})=\contr{\xi_{L,\beta}}(S^*\dd L)=(S\xi_{L,\beta})L=\Delta L.  
\end{equation}
Moreover, recall that $S X^c= X^v$. In addition, $[\xi_{L,\beta},X^c]$ is a vertical vector field, and thus $S[\xi_{L,\beta},X^c]=0$. Then
\begin{equation}
  (\dd E_L+\beta)(X^c)=-\xi_{L,\beta}(X^v L)+X^c(\Delta L).
\end{equation}
On the other hand, one can write
\begin{equation}
  \dd E_L (X^c)=X^c(E_L)=X^c(\Delta L)-X^c(L).
\end{equation}
Combining these last two equations yields
\begin{equation}
  \xi_{L,\beta} (X^v L)=X^c (L) -\beta (X^c).
\end{equation}
In particular, the right-hand side vanishes if and only if the left-hand side does.
 In other words, the following result holds.

\begin{theorem}[Noether's theorem for forced Lagrangian systems]\label{Noether_th}
Let $X$ be a vector field on $Q$. Then $X^c(L)=\beta(X^c)$ if and only if $X^v(L)$ is a constant of the motion.
\end{theorem}

A vector field $X$ on $Q$ such that $X^c(L)=\beta(X^c)$ will be called a \emph{symmetry of the forced Lagrangian} $(L,\beta)$.

Analogously to the conservative case, a vector field $X$ on $Q$ will be called a \emph{Lie symmetry} if $[X^c,\xi_{L,\beta}]=0$.
Observe that
 \begin{equation}
    \begin{aligned}
     \contr{[X^c,\xi_{L,\beta}]}\omega_{L}=& \liedv{X^c} (\contr{\xi_{L,\beta}}\omega_{L})- \contr{\xi_{L,\beta}} ( \liedv{X^c}\omega_{L}) \\
     =& \liedv{X^c} (\dd E_L+\beta)+ \contr{\xi_{L,\beta}} \dd ( \liedv{X^c}\alpha_{L}).
    \end{aligned}
  \end{equation}
Suppose that $\dd(\liedv{X^c}\alpha_L)=0$. Then $X$ is a Lie symmetry if and only if
\begin{equation}
  \liedv{X^c}\beta= -\dd (X^c (E_L)). \label{proposition_Lie_symmetry}
\end{equation}
Moreover, suppose that $\liedv{X^c}\alpha_L$ is exact. Then there exists a function $f$ on $TQ$ such that 
\begin{equation}
\begin{aligned}
  \dd f&= \liedv{X^c}\alpha_L=\contr{X^c}(\dd \alpha_L)+\dd (\contr{X^c}\alpha_L)=\contr{X^c}(\dd \alpha_L)+\dd (\contr{X^c}S^* \dd L)
  \\
  &=\contr{X^c}(\dd \alpha_L)+\dd (\contr{SX^c}\dd L)=\contr{X^c}(\dd \alpha_L)+\dd (X^v L),
\end{aligned}
\end{equation}
so
\begin{equation}
  \contr{\xi_{L,\beta}}\contr{X^c}(\dd \alpha_L)=\contr{\xi_{L,\beta}}(\dd (f-X^vL))=\xi_{L,\beta}(f-X^v(L)),
\end{equation}
but
\begin{equation}
  \contr{\xi_{L,\beta}}\contr{X^c}\dd \alpha_L=
  \contr{X^c}\contr{\xi_{L,\beta}}\omega_L=\contr{X^c}(\dd E_L+\beta)=X^c(E_L)+\beta(X^c).
\end{equation}
This implies that $f-X^v(L)$ is a constant of the motion if and only if $X^c(E_L)+\beta(X^c)$ vanishes. A vector field $X$ on $Q$ such that $\liedv{X^c}\alpha_L$ is exact and $X^c(E_L)+\beta(X^c)$ vanishes will be called a \emph{Noether symmetry}.

It is worth remarking that, unlike in the conservative case, Noether symmetries are not necessarily Lie symmetries. As a matter of fact, if $X$ is a Noether symmetry,
 \begin{equation}
    \begin{aligned}
      \liedv{X^c} \beta&=\contr{X^c}(\dd \beta)+\dd (\contr{X^c}\beta)
      =\contr{X^c}(\dd \beta)+\dd(\beta(X^c))\\
      &=\contr{X^c}(\dd \beta)-\dd(X^c(E_L)),
      \end{aligned}
  \end{equation}
   and obviously $\liedv{X^c}\alpha_L$ is closed.
  Then, by equation \eqref{proposition_Lie_symmetry}, $X$ is a Lie symmetry if and only if $\contr{X^c}(\dd \beta)$ vanishes. 

 Observe that every symmetry of the forced Lagrangian is a Noether symmetry. In fact, if $f$ is a constant function, clearly $X^v (L)$ is a conserved quantity. Moreover,
\begin{equation}
  \liedv{X^c}\alpha_L=0,
\end{equation}
so
\begin{equation}
\begin{aligned}
  0&=(\liedv{X^c}\alpha_L)(\xi_{L,\beta})=X^c\left(\alpha_L(\xi_{L,\beta})\right)-\alpha_L ([X^c,\xi_{L,\beta}])\\
  &=X^c(\Delta L)-S[X^c,\xi_{L,\beta}]L=X^c(\Delta L),
  \end{aligned}
\end{equation}
and thus,
\begin{equation}
  0=X^c(E_L)+\beta(X^c)=X^c(\Delta L) -X^c(L)+\beta(X^c)=-X^c(L)+\beta(X^c).
\end{equation}


These results can be summarized in the following way.
\begin{proposition}  \label{proposition_point_like}
Let $(L,\beta)$ be a forced Lagrangian system with forced Euler-Lagrange vector field $\xi_{L,\beta}$.
Let $X$ be a vector field on $Q$. Then the following results hold.
\begin{enumerate}
\item $X$ is a symmetry of the forced Lagrangian if and only if
$X^v(L)$ is a constant of the motion.
\item If $\liedv{X^c}\alpha_L$ is closed, then $X$ is a Lie symmetry if and only if 
\begin{equation}
  \liedv{X^c}\beta= -\dd (X^c (E_L)).
\end{equation}
\item If $X$ is a Noether symmetry, it is also a Lie symmetry if and only if
\begin{equation}
	\contr{X^c} \dd \beta =0.
\end{equation}
\item If $\liedv{X^c}\alpha_L=\dd f$,
then $X$ is a Noether symmetry if and only if $f-X^v(L)$ is a conserved quantity. 
\item If $X$ is a Noether symmetry, it is also a symmetry of the forced Lagrangian if and only if $\liedv{X^c}\alpha_L=0$. 
\end{enumerate}
\end{proposition}

The results just exposed apply for point-like symmetries, that is, infinitesimal symmetries on $Q$. Analogous results can be obtained for non-point-like symmetries, that is, infinitesimal symmetries on $TQ$.

A vector field $\tilde X$ on $TQ$ will be called a \emph{dynamical symmetry} if $[\tilde X, \xi_{L,\beta}]=0$. It will be called a \emph{Cartan symmetry} if $\tilde{X}(E_L)+\beta(\tilde{X})=0$ and $\liedv{X^c}\alpha_L$ is exact.
 

\begin{proposition} \label{proposition_non_point_like}
Consider a forced Lagrangian system $(L,\beta)$with forced Euler-Lagrange vector field $\xi_{L,\beta}$.
Let $X$ be a vector field on $Q$, and let $\tilde{X}$ be a vector field on $TQ$. Then the following statements hold.
\begin{enumerate}
 \item $X$ is a Lie symmetry if and only if $X^c$ is a dynamical symmetry.
 \item $X$ is a Noether symmetry if and only if $X^c$ is a Cartan symmetry.
 \item If $\liedv{\tilde{X}}\alpha_L$ is closed, then 
 $\tilde{X}$ is a dynamical symmetry if and only if
\begin{equation}
  \dd (\tilde{X} (E_L))=-\liedv{\tilde{X}}\beta.
\end{equation}
\item A Cartan symmetry is a dynamical symmetry if and only if
\begin{equation}
  \contr{\tilde{X}}\dd \beta=0.
\end{equation}
\item If $\liedv{\tilde{X}}\alpha_L=\dd f$, then $\tilde{X}$ is a Cartan symmetry if and only if $f-(S\tilde{X})(L)$ is a constant of the motion.

\end{enumerate}

\end{proposition}
The first two assertions are derived straightforwardly from the definitions. The rest can be proven in a completely analogous manner to the corresponding results for point-like symmetries. As it was discussed in the conservative case (see Section \ref{section_Lagrangian}), Theorem \ref{Noether_th} cannot be generalised for symmetries on $TQ$, since $[\xi_{L,\beta},\tilde{X}]$ is not a vertical vector field for a general $\tilde{X}$ on $TQ$.

\begin{example}[Harmonic oscillator with a gyroscopic force]
Consider an $n$-dimensional harmonic oscillator, with mass and frequency equal to 1 for simplicity's sake, so that its Lagrangian is
\begin{equation}
	L = \frac{1}{2} \sum_{i=1}^n \left[ \left( \dot q^i  \right)^2 - \left(  q^i  \right)^2  \right].
\end{equation}
Then the Poincaré-Cartan forms are simply
\begin{equation}
	\alpha_L = \dot{q}^i \dd q^i,\quad \omega_L = \dd q^i \wedge \dd \dot q^i,
\end{equation}
and the energy of the system is
\begin{equation}
	E_L = \frac{1}{2} \sum_{i=1}^n \left[ \left( \dot q^i  \right)^2 + \left(  q^i  \right)^2  \right].
\end{equation}
Suppose that it is subject to a gyroscopic force \cite{esen_geometrical_2021}, that is, an external force $\beta$ of the form
\begin{equation}
	\beta = s_{ij}(q) \dot{q}^j \dd q^i,
\end{equation}
where $s_{ij}$ is skew-symmetric. 
 The forced Euler-Lagrange vector field of the system $(L,\beta)$ is thus
\begin{equation}
	\xi_{L,\beta} = \dot{q}^i \frac{\partial  } {\partial q^i} 	
	- \left( q^i + s_{ij}   \dot{q}^j \right) \frac{\partial  } {\partial \dot q^i}
\end{equation} 
which can be straightforwardly derived from equations \eqref{local_expression_xi_L_beta} and \eqref{equation_xi_i}. Consider now the two-dimensional case. Obviously, in that case there is only one independent non-vanishing component of the gyroscopic force: $s_{12}(q)\eqqcolon g(q)$. Then
\begin{equation}
	\beta= g(q)\left(   \dot{q}^2 \dd q^1 - \dot{q}^1 \dd q^2\right),
\end{equation}
and
\begin{equation}
	\xi_{L,\beta} 
	 = \dot{q}^1 \frac{\partial  } {\partial q^1} 	
	+\dot{q}^2 \frac{\partial  } {\partial q^2} 	
	- \left( q^1 + g   \dot{q}^2 \right) \frac{\partial  } {\partial \dot q^1}
	- \left( q^2 - g   \dot{q}^1 \right) \frac{\partial  } {\partial \dot q^2}.
\end{equation} 
 Consider the vector field $\tilde X$ on $TQ$ given by
\begin{equation}
	\tilde X = \left( g \dot{q}^1 +q^2 \right) \frac{\partial  } {\partial q^1}
	+\left(  g \dot{q}^2 -q^1 \right) \frac{\partial  } {\partial q^2}
	+\dot{q}^2 \frac{\partial  } {\partial \dot{q}^1}
	-\dot{q}^1 \frac{\partial  } {\partial \dot{q}^2}
\end{equation}
Clearly, this vector field verifies
\begin{equation}
	\tilde{X}(E_L)+\beta(\tilde{X})=0.
\end{equation}
In addition,
\begin{equation}
	\contr{\tilde{X}} \alpha_L = \dot q^1 \left( g\dot q^1 +q^2  \right) + \dot q^2 \left( g \dot q^2 -q^1  \right),
\end{equation}
and
\begin{equation}
	\contr{X}\omega_L 
	= \left( g\dot q^1 +q^2  \right) \dd \dot{q}^1 + \left( g \dot q^2 -q^1  \right) \dd \dot q^2
	-\dot{q}^2 \dd q^1 + \dot{q}^1 \dd q^2
\end{equation}
If $\beta$ does not depend on $q$, that is, $g$ is a constant, one can write
\begin{equation}
	\contr{X}\omega_L = \dd \left[
	 \frac{1}{2} g \left( \dot{q}^1  \right)^2 +\frac{1}{2} g \left( \dot{q}^2  \right)^2 
	 -q^1 \dot{q}^2 + q^2 \dot{q}^1
	  \right],
\end{equation}
and hence
\begin{equation}
	\liedv{\tilde{X}}\alpha_L 
	= \dd (\contr{\tilde X}\alpha_L) - \contr{\tilde{X}} \omega_L
	= \dd f,
\end{equation}
where
\begin{equation}
	f =  \frac{1}{2} g \left( \dot{q}^1  \right)^2 +\frac{1}{2} g \left( \dot{q}^2  \right)^2.
\end{equation}
Therefore, when $\beta$ is independent of $q$, $\tilde{X}$ is a Cartan symmetry of $(L,\beta)$, and
\begin{equation}
	\ell= f - (S\tilde X) (L) = q^1 \dot q^2 -q^2 \dot q^1 
	-\frac{g}{2} \left[ \left( \dot{q}^1  \right)^2 +  \left( \dot{q}^2  \right)^2  \right] 
\end{equation}
is its associated constant of the motion. Observe that when $g=0$ the conservation of angular momentum is recovered. In fact, when $g=0$ one can write $\tilde{X}=X^c$ for
\begin{equation}
	X = q^2 \frac{\partial  } {\partial q^1} -q^1 \frac{\partial  } {\partial q^2}, 
\end{equation}
which is the generator of rotations on the $(q_1,q_2)$ plane.
\end{example}

The types of symmetries for a forced Lagrangian system $(L,\beta)$ on $TQ$ are summarized in Figure \ref{figure_diagram_forced}.
Observe that, unlike in the conservative case (\oldemph{cf.} Figure \ref{figure_diagram_conservative}), Noether symmetries are not a subset of Lie symmetries, and Cartan symmetries are not a subset of dynamical symmetries neither. Obviously, when $\beta$ vanishes, $\xi_{L,\beta}=\xi_L$ and the conservative results are recovered.

\begin{figure}[b]
\centering
\includegraphics[width=.7\linewidth]{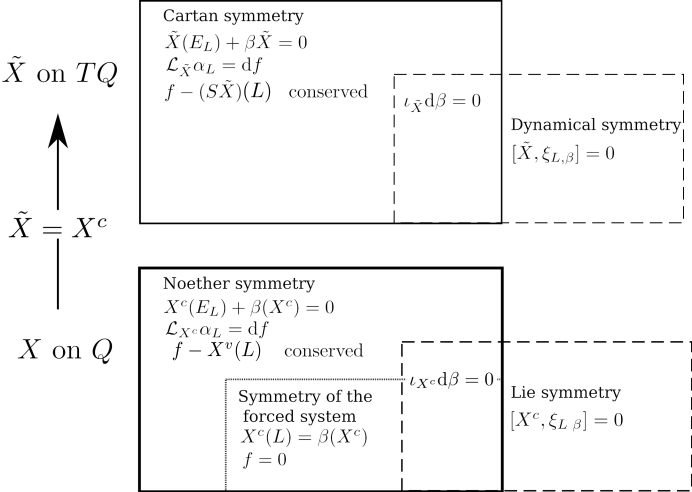}
\caption{Types of symmetries on $Q$ and on $TQ$. The complete lifts of Lie symmetries and Noether symmetries correspond to dynamical symmetries and Cartan symmetries, respectively. The symmetries of the forced Lagrangian system are the subset of Noether symmetries for which $\liedv{X^c}\alpha_L=0$, so $f$ is a constant that can be taken as 0 without loss of generality. The Noether symmetries that satisfy $\contr{X^c}\dd \beta=0$ are also Lie symmetries, and analogously with the intersection between Cartan and dynamical symmetries.}
\label{figure_diagram_forced}
\end{figure}

Similar results can be obtained in the Hamiltonian framework. Consider a forced Hamiltonian system on $(T^*Q,\omega)$ with Hamiltonian function $H$ and external force $\gamma$. Denote by $X_{H,\gamma}$ the dynamical vector field associated to $(H,\gamma)$. A \emph{constant of the motion} (or a \emph{conserved quantity}) is then a function $f$ on $T^*Q$ which is a first integral of $X_{H,\gamma}$ (that is, $X_{H,\gamma}(f)=0$). Recall that the Poisson bracket of two functions $f$ and $g$ on $(T^*Q,\omega)$ is given by
\begin{equation}
	\left\{f,g  \right\} = \omega(X_f,X_g),
\end{equation}
where $X_f$ and $X_g$ are the Hamiltonian vector fields associated with $f$ and $g$, respectively. 
Clearly,
\begin{equation}
\begin{aligned}
  X_{H,\gamma}(f)&=\contr{X_{H,\gamma}}\dd f= \contr{X_{H,\gamma}}( \contr{X_f}\omega_Q)=-\contr{X_f}(\contr{X_{H,\gamma}}\omega_Q)\\
  &=-\omega_Q(X_H,X_f)-\gamma(X_f)= \left\{f,H  \right\}-\gamma(X_f),
\end{aligned}
\end{equation}
and hence $f$ is a constant of the motion if and only if
\begin{equation}
  \left\{f,H  \right\}=\gamma(X_f).
\end{equation}
The Hamiltonian counterpart of Proposition \ref{proposition_non_point_like} is as follows.
\begin{proposition}
Let $\hat{X}$ be a vector field on $T^*Q$. Then the following assertions hold:
\begin{enumerate}
\item If $\liedv{\hat{X}}\theta $ is closed, then $\hat{X}$ commutes with $X_{H,\gamma}$ if and only if
\begin{equation}
  \dd (\hat{X} (H))=-\liedv{\hat{X}}\gamma.
\end{equation}
\item If $\liedv{\hat{X}}\theta =\dd f$, then $\hat{X}(H)+\gamma(\hat{X})=0$ if and only if $f-\theta (\hat{X})$ is a constant of the motion.
\item If the previous statement holds, $\hat{X}$ commutes with $X_{H,\gamma}$ if and only if 
\begin{equation}
  \contr{\hat{X}}\dd \gamma=0.
\end{equation}
\end{enumerate}
\end{proposition}

If the Lagrangian is hyperregular, the symmetries and constants of the motion on $TQ$ can be easily associated with their counterparts on $T^*Q$ via the Legendre transform. Suppose that $(L,\beta)$ is a Lagrangian system such that $H\circ \Leg = E_L$ and $\Leg^* \gamma=\beta$. Let $\tilde{X}$ be a vector field on $TQ$ and $\hat{X}$ the $\Leg$-related vector field on $T^*Q$. Then:
\begin{enumerate}[label=\roman*)]
\item $\hat{X}$ commutes with $X_{H,\gamma}$ if and only if $\tilde{X}$ is a dynamical symmetry of $(L,\beta)$.
\item $\liedv{\hat{X}}\theta =\dd f$ if and only if $\liedv{\tilde{X}}\alpha_L=\dd g$, where $g=f\circ \Leg$.
\item Suppose that $\liedv{\hat{X}}\theta =\dd f$. Then the following assertions are equivalent.
\begin{enumerate}[label=\alph*)]
\item $\hat{X}(H)+\gamma(\hat{X})=0$.
\item $f-\theta (\hat{X})$ is a conserved quantity.
\item $\tilde{X}(E_L)+\beta(\tilde{X})=0$.
\item $f\circ \Leg-\alpha_L(\tilde{X})$ is a conserved quantity.
\end{enumerate}

\end{enumerate}

\section{Momentum map and reduction}\label{section_reduction_forced}
In this section, the symplectic reduction method explained in Section \ref{section_reduction} is extended for forced Lagrangian systems. Consider a Lie group $G$ with action $\Phi_g$ on $Q$, and consider the lifted action to $TQ$ by tangent prolongation. In what follows, every group action will be assumed to be free and proper. Denote by $\mathfrak{g}$ the Lie algebra of $G$, and by $\mathfrak{g}^*$ the dual of $\mathfrak{g}$.

 As it was previously explained, the symplectic point reduction (see Theorem \ref{theorem_symplectic_point_reduction}) requires the action of $G$ to be symplectic. In the conservative case, if the Lagrangian $L$ on $TQ$ is $G$-invariant, the dynamics given by $\xi_L$ on $TQ$ can be reconstructed from the dynamics given by $\xi_{L_\mu}$ on the reduced space $(TQ)_\mu$. Since the symplectic structure $\omega_L$ on $TQ$ is induced by the Lagrangian $L$, in the case of a forced Lagrangian system $(L,\beta)$, reduction will require $L$ and $\beta$ to be invariant under $G$ (or a subgroup of $G$) independently.

 Hereinafter, suppose that the $G$-action leaves $L$ invariant, and hence, $\alpha_L$ and $\omega_L$ are invariant.  Recall that then the natural momentum map 
\begin{equation}
\begin{aligned}
  J:TQ&\to \mathfrak{g}^*,\\
  J(\xi)&=\alpha_L\left(\xi_Q^c  \right),
\end{aligned}
\end{equation}
is $\operatorname{Ad}^*$-equivariant. For each $\xi\in \mathfrak{g}$ and $v_q\in TQ$, $J^\xi:TQ\to \RR$ is the function given by
\begin{equation}
  J^\xi (v_q) = \left\langle J(v_q),\xi  \right\rangle   \label{J_xi_function}.
\end{equation} 
In the conservative case, $J^\xi$ is a constant of the motion for each $\xi\in \mathfrak{g}$. Consider a forced Lagrangian system $(L,\beta)$ with forced Euler-Lagrange vector field $\xi_{L,\beta}$. Clearly, 
\begin{equation}
  J^\xi= \alpha_L (\xi_Q^c) = \contr{\xi_Q^c}\alpha_L,
\end{equation}
so
\begin{equation}
 \dd J^\xi = \dd (\contr{\xi_Q^c}\alpha_L) 
 = \liedv{\xi_Q^c}\alpha_L - \contr{\xi_Q^c} \dd \alpha_L
 = \contr{\xi_Q^c}\omega_L 
 \label{Hamiltonian_vector_momentum_map},
\end{equation}
since $\alpha_L$ is $\mathfrak{g}$-invariant.
Contracting this equation with $\xi_{L,\beta}$, one gets
\begin{equation}
  \contr{\xi_{L,\beta}}\dd J^\xi=\xi_{L,\beta}(J^\xi),
\end{equation}
on the left-hand side, and
\begin{equation}
  \contr{\xi_{L,\beta}}\contr{\xi_Q^c}\omega_L
  =-\contr{\xi_Q^c}\contr{\xi_{L,\beta}}\omega_L
  =-\contr{\xi_Q^c}(\dd E_L+\beta)=-\xi_Q^c(E_L)-\beta(\xi_Q^c),
\end{equation}
on the right-hand side.
Thus $J^\xi$ is a conserved quantity for $\xi_{L,\beta}$ if and only if
\begin{equation}
  \xi_Q^c(E_L)+\beta(\xi_Q^c)=0. \label{conservation_dynamics_algebra}
\end{equation}
Now observe that
\begin{equation}
\begin{aligned}
  \xi_Q^c(E_L)&=\xi_Q^c(\Delta(L))-\xi_Q^c(L)=\xi_Q^c(\Delta(L))
  =\liedv{\xi_Q^c}(\Delta(L))=\liedv{\xi_Q^c}(\contr{\Delta}\dd L)\\
  &=\contr{[\xi_Q^c,\Delta]}\dd L+\contr{\Delta}(\liedv{\xi_Q^c} \dd L)
  =\contr{[\xi_Q^c,\Delta]}\dd L
  =[\xi_Q^c,\Delta](L),
\end{aligned}
\end{equation}
since $\xi_Q^c(L)=0$ by the $G$-invariance of $L$, but $[\xi_Q^c,\Delta]=0$, and thus
\begin{equation}
  \xi_Q^c(E_L)=0 
\end{equation}
for each $\xi\in \mathfrak{g}$, that is, $E_L$ is $G$-invariant. By equation \eqref{conservation_dynamics_algebra}, $J^\xi$ is a conserved quantity for $\xi_{L,\beta}$ if and only if
\begin{equation}
  \beta(\xi_Q^c) = 0. \label{conservation_dynamics_algebra_simplified}
\end{equation}
This result can also be obtained from a variational approach (see reference \cite[Theorem 3.1.1]{marsden_west_01}).

In addition,
\begin{equation}
  \liedv{\xi_Q^c} \beta = \dd (\contr{\xi_Q^c} \beta) + \contr{\xi_Q^c} \dd \beta = \dd (\beta(\xi_Q^c)) + \contr{\xi_Q^c} \dd \beta.
\end{equation}
If equation \eqref{conservation_dynamics_algebra_simplified} holds, then $\beta$ is $\xi$-invariant (that is, $\liedv{\xi_Q^c} \beta = 0$) if and only if
\begin{equation}
  \contr{\xi_Q^c} \dd \beta = 0. \label{condition_subalgebra_Lie_symmetry}
\end{equation}
This motivates defining the vector subspace $\mathfrak{g}_\beta$ of $\mathfrak{g}$ given by
\begin{equation}
  \mathfrak{g}_\beta = \left\{\xi \in \mathfrak{g} \mid
    \beta(\xi_Q^c)=0,\ \contr{\xi_Q^c}\dd \beta = 0
    \right\}.
\end{equation}
Furthermore, for each $\xi,\eta\in \mathfrak{g}_\beta$,
\begin{equation}
\begin{aligned}
  \beta \left([\xi_Q,\eta_Q]^c  \right) &=  \beta \left([\xi_Q^c,\eta_Q^c]  \right) 
  = \contr{[\xi_Q^c,\eta_Q^c]}\beta
  = \liedv{\xi_Q^c} \contr{\eta_Q^c} \beta 
    - \contr{\eta_Q^c} \liedv{\xi_Q^c}  \beta\\
  &= \xi_Q^c (\beta(\eta_Q^c)) - \eta_Q^c (\beta(\xi_Q^c)) 
  - \contr{\eta_Q^c} (\contr{\xi_Q^c} \dd \beta)
  = 0 
    , 
\end{aligned}
\end{equation}
and similarly,
\begin{equation}
\begin{aligned}
   \contr{[\xi_Q,\eta_Q]^c} \dd\beta
   &= \contr{[\xi_Q^c,\eta_Q^c]} \dd\beta
   = \liedv{\xi_Q^c} \contr{\eta_Q^c} \dd \beta 
    - \contr{\eta_Q^c}\liedv{\xi_Q^c}\dd \beta\\
    &= \liedv{\xi_Q^c} \contr{\eta_Q^c} \dd \beta 
    - \contr{\eta_Q^c}\dd \liedv{\xi_Q^c} \beta=0.
\end{aligned}
\end{equation}
Then $[\xi,\eta]\in \mathfrak{g}_\beta$ for each $\xi,\eta\in \mathfrak{g}_\beta$, and hence $\mathfrak{g}_\beta$ is a Lie subalgebra of $\mathfrak{g}$.

Since $\alpha_L$ is invariant, 
\begin{equation}
  \liedv{\xi_Q^c}\alpha_L=0
\end{equation}
for each $\xi \in \mathfrak{g}$. In combination with equation \eqref{conservation_dynamics_algebra}, this implies that $\xi_Q$ is a Noether symmetry of $(L,\beta)$. In fact, it is also a symmetry of the forced Lagrangian
(see Proposition \ref{proposition_point_like} v) ). By equation \eqref{condition_subalgebra_Lie_symmetry}, $\xi_Q$ is also a Lie symmetry. As a matter of fact, $\mathfrak{g}_\beta$ is the subspace of $\xi\in \mathfrak{g}$ such that $\xi_Q$ is both a Lie symmetry and a Noether symmetry of $(L,\beta)$.

Symplectic point reduction (Theorem \ref{theorem_symplectic_point_reduction}) can be generalized for forced Lagrangian systems as follows.

\begin{theorem}\label{theorem_reduction}
 Let  $G_\beta\subset G$ be the Lie subgroup generated by $\mathfrak{g}_\beta$ and $J_\beta:TQ\to \mathfrak{g}_\beta^*$ the reduced momentum map. Let $\mu\in \mathfrak{g}_\beta^*$ be a regular value of $J_{\beta}$ and $(G_\beta)_\mu$ the isotropy group in $\mu$.
 Then the following results hold:
\begin{enumerate}[label=\roman*)]
\item $J_\beta^{-1}(\mu)$ is a submanifold of $TQ$ and $\xi_{L,\beta}$ is tangent to it. 
\item The quotient space $(TQ)_\mu\coloneqq J_\beta^{-1}(\mu)/(G_\beta)_\mu$ is endowed with an induced symplectic structure $\omega_\mu$, namely 
\begin{equation}
  \pi_\mu^*\omega_\mu = \incl_\mu^*\omega_L,
\end{equation}
where $\pi_\mu:J_\beta^{-1}(\mu)\to (TQ)_\mu$ and $\incl_\mu:J_\beta^{-1}(\mu)\hookrightarrow TQ$ denote the projection and the inclusion, respectively. 
\item $L$ induces a function $L_\mu:(TQ)_\mu\to \RR$ defined by
\begin{equation}
  L_\mu \circ \pi_\mu = L\circ \incl_\mu. \label{equation_reduced_Lagrangian}
\end{equation}
\item One can also introduce a function $E_{L_\mu}:(TQ)_\mu\to \RR$, given by $E_{L_\mu}=\Delta_\mu(L_\mu)-L_\mu$, which satisfies
\begin{equation}
  E_{L_\mu} \circ \pi_\mu = E_L\circ \incl_\mu. \label{reduced_energy}
\end{equation}
\item $\beta$ induces a reduced semibasic 1-form $\beta_\mu$ on $(TQ)_\mu$, given by
\begin{equation}
  \pi_\mu^* \beta_\mu=\incl_\mu^*\beta.
\end{equation}
\end{enumerate}
\end{theorem}
\begin{proof}
The proof of the first three statements can be found in references \cite{ortega_momentum_2004,abraham_foundations_2008,marsden_reduction_1990}. 
If $\Delta$ is the Liouville vector field on $TQ$, let $\Delta_\mu$
is a $\pi_\mu$-related vector field on $(TQ)_\mu$, namely
\begin{equation}
 \pi_{\mu*}\Delta_\mu = \incl_{\mu*}\Delta.
\end{equation}
Observe that
\begin{equation}
\begin{aligned}
\Delta(L)\circ \incl_\mu
&= \incl_\mu^*\Delta(L)
=\incl_\mu^* (\contr{\Delta}\dd L)
=\contr{\incl_{\mu*}\Delta} (\incl_\mu^*\dd L)
=(\incl_{\mu*}\Delta)(\incl_{\mu}^* L)\\
&=(\incl_{\mu*}\Delta)(L\circ \incl_{\mu})
=(\incl_{\mu*}\Delta)(L_\mu\circ \pi_{\mu})
= (\pi_{\mu*}\Delta_\mu) (L_\mu\circ \pi_{\mu})\\
&= \pi_\mu^* (\Delta_\mu(L_\mu)) 
= \Delta_\mu(L_\mu) \circ \pi_\mu.
\end{aligned} \label{equation_reduced_Delta_L}
\end{equation}
Combining equations \eqref{equation_reduced_Lagrangian} and \eqref{equation_reduced_Delta_L}, it is clear that equation~\eqref{reduced_energy} holds.

On the other hand,
\begin{equation}
\begin{aligned} 
	\beta(\xi_Q^c) = \contr{\xi_Q^c}  \beta = \contr{\xi_Q^c} \contr{\xi_\beta} \omega_L
	= - \contr{\xi_\beta} \contr{\xi_Q^c}  \omega_L ,
\end{aligned}
\end{equation}
but
\begin{equation}
	\contr{\xi_Q^c}  \omega_L = -\liedv{\xi_Q^c} \alpha_L + \dd(\contr{\xi_Q^c} \alpha_L)
	= \dd(\contr{\xi_Q^c} \alpha_L) = \dd J^\xi,
\end{equation}
since $\alpha_L$ is $\mathfrak{g}$-invariant, and thus
\begin{equation}
	\beta(\xi_Q^c) = - \contr{\xi_\beta} \dd J^\xi = -\xi_\beta(J^\xi).
\end{equation}
The left-hand side vanishes for any $\xi\in \mathfrak{g}_\beta$ by definition, and hence the flow of $\xi_\beta$ preserves $J^\xi$. In other words, the flow $F_t$ of $\xi_\beta$ leaves invariant the connected components of $J^{-1}(\mu)$. One can then define a flow $F_t^\mu$ on $(TQ)_\mu$ given by
\begin{equation}
	\pi_\mu \circ F_t \circ \incl_\mu = F_t^\mu \circ \pi_\mu. \label{equation_flow_beta}
\end{equation}
Let $Y$ be the vector field on $(TQ)_\mu$ whose flow is $F_t^\mu$.
By construction, $Y$ is $\pi_\mu$-related to $\xi_\beta$, so one can introduce a 1-form $\beta_\mu \coloneqq \contr{Y}\omega_\mu$ on $(TQ)_\mu$ that satisfies
\begin{equation}
	\pi_\mu^* \beta_\mu= \pi_\mu^* (\contr{Y}\omega_\mu)
	 =\contr{\pi_{\mu*}Y}( \pi_\mu^* \omega_\mu)
	 = \contr{\incl_{\mu*}\xi_{\beta}}( \incl_\mu^* \omega_L)
	 = \incl_\mu^*(\contr{\xi_\beta} \omega_L)
	 = \incl_\mu^* \beta.
\end{equation}

\end{proof}

The dynamics on $(TQ)_\mu$ are determined by the vector field $\xi_{L_\mu,\beta_\mu}$, defined by
 \begin{equation}
   \contr{\xi_{L_\mu,\beta_\mu}}\omega_\mu =\dd E_{L_\mu}+\beta_\mu.
 \end{equation}
This vector field is $\pi_\mu$-related to $\xi_{L,\beta}$. The dynamics on $J^{-1}(\mu)\subset TQ$ can be reconstructed by the same procedure as in the conservative case (see Remark \ref{remark_reconstruction}).



\section{Rayleigh systems}\label{section_Rayleigh}
This section presents a special type of external forces: the so-called Rayleigh dissipation (see references \cite{minguzzi_rayleighs_2015,Cline2021Rayleigh’s,goldstein_mecanica_1987,gantmakher_lectures_1970,lurie_analytical_2002,strutt_general_1871}). As far as the author is aware, these forces have not been previously studied in a geometric formalism. The results regarding symmetries and reduction from sections \ref{section_Lagrangian_forced} and \ref{section_reduction_forced}, respectively, are now particularized for Rayleigh dissipative forces. 

Rayleigh's hypothesis is a non-conservative force which is linear in the velocities.
Let $R$ be a bilinear form on $TQ$ given by
\begin{equation}
\begin{aligned}
R:\ TQ\times TQ &\to \RR,\\
 R(q,\dot q_1,\dot q_2)&=R_{ij}(q) \dot q^i_1\dot q^j_2. \label{Rayleigh_bilinear_form}
\end{aligned}
\end{equation}
In other words, $R$ is a symmetric $(0,2)$-tensor field on $Q$.
 Consider the quadratic form $\mathcal{R}$ on $TQ$ given by\footnote{Here the $1/2$ factor is introduced for convenience.}
\begin{equation}
	\mathcal{R}(q,\dot{q}) 
	= \frac{1}{2} R(q,\dot{q},\dot{q})
	= \frac{1}{2} R_{ij} (q) \dot{q}^i \dot{q}^j,
\end{equation}
which is called \emph{Rayleigh dissipation function}. This dissipation function generates an external force $\bar{R}$, whose component along the $i$-th direction is given by
\begin{equation}
	\bar{R}_i = \frac{\partial \mathcal{R}} {\partial \dot{q}^i} = R_{ij}(q) \dot{q}^j,
\end{equation}
that is,
\begin{equation}
	\bar{R} = R_{ij}(q) \dot{q}^j \dd q^i.
\end{equation}
 Since $R$ is a $(0,2)$-tensor on $Q$, it defines a linear mapping 
\begin{equation}
\tilde{R}: TQ\to T^*Q
\end{equation} 
by
\begin{equation}
\tilde{R}(v_q)=\contr{v_q}R,
\end{equation}
that is,
\begin{equation}
\tilde{R}(v_q)(w_q)=R(v_q,w_q).
\end{equation}
Therefore
\begin{equation}
\tilde{R}(q^i,\dot q^i)=(q^i,R_{ij}(q)\dot q^i \dot q^j).
\end{equation}
Recall that, given a morphism of fibre bundles $TQ\to T^*Q$, there is always an associated semibasic 1-form on $TQ$ and conversely. Hence $\bar{R}$ can also be defined as the semibasic 1-form on $TQ$ associated with the fibre morphism $\tilde{R}$.

Consider a forced Lagrangian system $(L,\bar{R})$ on $TQ$. If $\xi_{\bar{R}}$ is the vector field on $TQ$ given by
\begin{equation}
	\contr{\xi_{\bar{R}}}\omega_L = \bar{R},
\end{equation}
then locally
\begin{equation}
	\xi_{\bar{R}}=-R_{ik} \dot{q}^k W^{ij} \frac{\partial  } {\partial \dot q^j}.
\end{equation}
The dynamics of the system $(L,\bar{R})$ is given by the vector field $\xi_{L,\bar{R}}$, given by
\begin{equation}
	\xi_{L,\bar{R}} = \xi_L + \xi_{\bar{R}},
\end{equation}
where $\xi_L$ is the Hamiltonian vector field of $L$ on $(TQ,\omega_L)$.
The equations of motion of the system are given by
\begin{equation}
	\frac{\mathrm{d} } {\mathrm{d}t} \left( \frac{\partial L} {\partial \dot{q}^i}  \right)
	- \frac{\partial L} {\partial q^i}
	= -R_{ij}(q) \dot{q}^j.
\end{equation}

Notice that
\begin{equation}
	\xi_{L,\bar{R}} (E_L) 
	= \contr{\xi_{L,\bar{R}}} \dd E_L
	 =   \contr{\xi_{L,\bar{R}}} \contr{\xi_{L}} \omega_L
	 = - \contr{\xi_{L}} \contr{\xi_{L,\bar{R}}} \omega_L
	 = - \contr{\xi_{L}} (\dd E_L + \bar{R}),
\end{equation}
but
\begin{equation}
	\contr{\xi_{L}} \dd E_L = \contr{\xi_{L}}\contr{\xi_{L}} \omega_L = 0,
\end{equation}
so
\begin{equation}
	\xi_{L,\bar{R}} (E_L)
	= - \contr{\xi_{L}} \bar{R}
	= - R_{ij} (q) \dot{q}^i \dot{q}^j
	= -2 \mathcal{R}.
\end{equation}
Suppose that $\sigma$ is a trajectory of the system (that is, an integral curve of $\xi_{L,\bar{R}}$). Then
\begin{equation}
	\frac{\mathrm{d} } {\mathrm{d}t} E_L(\sigma(t))
	= \xi_{L,\bar{R}}(E_L)(\sigma(t))
	= -2 \mathcal{R}(\sigma(t)).
\end{equation}
Therefore $\mathcal{R}$ expresses the rate at which energy is dissipated away along the evolution of the system.

Of course, Rayleigh dissipation can also be described in the Hamiltonian formalism.
Indeed, since $L$ is assumed to be hyperregular, the external force $\hat{R}$ on $T^*Q$ can be defined by
\begin{equation}
\bar{R}=\Leg^*\hat{R},
\end{equation}
and the Hamiltonian function by $H\circ \Leg =E_L$. Locally $\hat{R}$ can be written as
\begin{equation}
\hat{R}={R}_{ij}(q)p_i \dd q^j.
\end{equation}
If $Z_{\hat{R}}$ is the vector field defined by
\begin{equation}
\contr{Z_{\hat{R}}}\omega_Q=\hat{R},
\end{equation}
then the equations of motion of the system are the integral curves of the vector field $X_{H,\hat{R}}$, given by
\begin{equation}
X_{H,\hat{R}}=X_H+Z_{\hat{R}},
\end{equation}
where $X_H$ is the Hamiltonian vector field of $H$ on $(T^*Q,\omega)$. In canonical coordinates, 
\begin{equation}
	Z_{\hat{R}}=-{R}_{ij}(q)p_i\parder{}{p_j}.
\end{equation}
The equations of motion are
\begin{equation}
\begin{aligned}
&\der{q^i}{t}=\parder{H}{p_i},\\
&\der{p_i}{t}=-\left(\parder{H}{q^i}+R_{ij}(q) p_j\right).
\end{aligned}
\end{equation}


If a forced Lagrangian system $(L,\bar{R})$ is subject to an external force which is linear in the velocities, a Rayleigh dissipation function $\mathcal{R}$ from which $\bar{R}$ is derived can be defined. In that case, the system can be characterized by the pair of functions $L$ and $\mathcal{R}$ on $TQ$.
 As a matter of fact, an external force does not have to be linear in the velocities for this to be the case. Consider a function $\mathcal{R}$ on $TQ$ (not necessarily a quadratic form)
 and the external force $\bar{R}$ given by
  \begin{equation}
 	\bar{R} = \frac{\partial \mathcal{R}} {\partial \dot{q}^i} \dd q^i.
 	\label{local_Rayleigh_force}
 \end{equation}
 Geometrically, this can be expressed as
 \begin{equation}
 	\bar{R} 
 	= S^*(\dd \mathcal{R}).
 \end{equation}
 Here $\mathcal{R}$ takes the role of a ``potential'' from which the external force is derived. In what follows, this type of forces will be referred to as \emph{Rayleigh forces}, and $\mathcal{R}$ will be called its \emph{Rayleigh dissipation function} (or \emph{Rayleigh potential}). A \emph{Rayleigh system}  $(L,\mathcal{R})$ will denote a forced Lagrangian system $(L,\bar{R})$ with Rayleigh force $\bar{R}$ derived from $\mathcal{R}$.

 Consider a Rayleigh system $(L,\mathcal{R})$.
If $\xi_L$ is the Hamiltonian vector field of $L$, then
 \begin{equation}
 	\contr{\xi_L} \bar{R} = \contr{\xi_L} (S^*\dd \mathcal{R})
 	= \contr{S\xi_L} \dd \mathcal{R} 
 	= \contr{\Delta} \dd \mathcal{R}
 	= \Delta (\mathcal{R}).
 \end{equation}
 The result above generalizes straightforwardly for non-linear Rayleigh forces as follows.
 \begin{proposition} [Dissipation of energy]
  Consider a Rayleigh system $(L,\mathcal{R})$ on $TQ$. If $\sigma:TQ\to \RR$ is a trajectory of the system, then 
  \begin{equation}
  	\frac{\mathrm{d}  } {\mathrm{d}t} E_L \circ \sigma(t) = -\Delta(\mathcal{R}) \circ \sigma(t),
  \end{equation}
  where $E_L$ is the energy function associated with $L$ and $\Delta$ is the Liouville vector field on $TQ$.
 \end{proposition}

 There is a certain ``gauge freedom'' in the choice of $\mathcal{R}$.
 \begin{remark}\label{remark_gauge_potential}
 For any function $f$ on $Q$, $\mathcal{R}$ and $\mathcal{R}+f$ define the same Rayleigh force $\bar{R}$. In other words, $(L,\mathcal{R})\cong(L,\mathcal{R}+f)$.
 
 \end{remark}

\subsection{Constants of the motion for Rayleigh systems}
Consider a Rayleigh system $(L,\mathcal{R})$ on $TQ$. Let $\tilde{X}$ be a vector field on $TQ$ and $X$ a vector field on $Q$. Observe that
\begin{equation}
  \bar{R}(\tilde{X})
  =\contr{\tilde{X}} \bar{R}
  =\contr{\tilde{X}}(S^{*}\dd \mathcal{R})
  =\contr{S \tilde{X}}\dd \mathcal{R}
  =(S\tilde{X}) (\mathcal{R}).
\end{equation}
In particular, $SX^c=X^v$, so
\begin{equation}
	\bar{R}(X^c) = X^v (\mathcal{R}).
\end{equation}
By equation \eqref{local_Rayleigh_force}, locally
\begin{equation}
	\dd \bar{R} = \frac{\partial^2 \mathcal{R}  } {\partial q^i \dot q^j  } \dd q^i \wedge \dd q^j
	+  \frac{\partial^2 \mathcal{R}  } {\partial \dot q^i \dot q^j  }  \dd \dot q^j\wedge \dd q^i,
\end{equation}
so
\begin{equation}
	\contr{X^c} \dd \bar{R}
	=\left(  X^j  \frac{\partial^2 \mathcal{R}  } {\partial q^j \dot q^i  } -  X^j  \frac{\partial^2 \mathcal{R}  } {\partial \dot q^j  q^i  } 
	+ \dot{q}^k \frac{\partial X^j} {\partial q^k} \frac{\partial^2 \mathcal{R}  } {\partial \dot q^i \dot q^j  }
	 \right) \dd q^i - X^i \frac{\partial^2 \mathcal{R}  } {\partial \dot q^i \dot q^j  }  \dd \dot q^j.
\end{equation}
In addition,
\begin{equation}
	X^v(\mathcal{R}) = X^i \frac{\partial \mathcal{R}} {\partial \dot{q}^i},
\end{equation}
so
\begin{equation}
	\dd (X^v (\mathcal{R}))
	= \left( X^j \frac{\partial^2 \mathcal{R}} {\partial q^i \dot q^j}  
	+\frac{\partial X^j} {\partial q^i} \frac{\partial \mathcal{R}} {\partial \dot{q}^j}
	\right) \dd q^i
	+ X^j \frac{\partial ^2 \mathcal{R} } {\partial \dot{q}^i \dot{q}^j} \dd \dot q^i.
\end{equation}
Therefore
\begin{equation}
\begin{aligned}
	\liedv{X^c} \bar{R} 
	&= \contr{X^c} \dd \bar{R} + \dd(\contr{X}^c \bar{R})
	= \contr{X^c} \dd \bar{R} + \dd(\bar{R}(X^c))
	=  \contr{X^c} \dd \bar{R} + \dd(X^v(\mathcal{R}))\\
	& =\left(  X^j  \frac{\partial^2 \mathcal{R}  } {\partial q^j \dot q^i  } 
	+ \dot{q}^k \frac{\partial X^j} {\partial q^k} \frac{\partial^2 \mathcal{R}  } {\partial \dot q^i \dot q^j  } +\frac{\partial X^j} {\partial q^i} \frac{\partial \mathcal{R}} {\partial \dot{q}^j}
	 \right) \dd q^i.
\end{aligned}
\end{equation}
On the other hand,
\begin{equation}
	X^c(\mathcal{R}) = X^j \frac{\partial \mathcal{R}} {\partial q^j} +\dot{q}^k \frac{\partial X^j} {\partial q^k  } \frac{\partial \mathcal{R}} {\partial \dot{q}^j},
\end{equation}
and then
\begin{equation}
	S^*(\dd (X^c(\mathcal{R})) )
	= \frac{\partial (X^c(\mathcal{R}))} {\partial \dot{q}^i } \dd q^i
	=\left(  X^j  \frac{\partial^2 \mathcal{R}  } {\partial q^j \dot q^i  } 
	+ \dot{q}^k \frac{\partial X^j} {\partial q^k} \frac{\partial^2 \mathcal{R}  } {\partial \dot q^i \dot q^j  } +\frac{\partial X^j} {\partial q^i} \frac{\partial \mathcal{R}} {\partial \dot{q}^j}
	 \right) \dd q^i.
\end{equation}
This shows that
\begin{equation}
	\liedv{X^c} \bar{R} = S^*(\dd (X^c(\mathcal{R})) ).
\end{equation}

The results from Section \ref{section_Lagrangian_forced} can now be expressed in the following way. 

\begin{proposition}
	The following assertions are equivalent:
	\begin{enumerate}
		\item $X$ is a symmetry of the forced Lagrangian $(L,\bar{R})$,
		\item $X^c(L)=X^v(\mathcal{R})$,
		\item $X^v(L)$ is a conserved quantity.
	\end{enumerate}
\end{proposition}

\begin{example}[Fluid resistance]
Consider a body of mass $m$ moving through a fluid that fully encloses it. For the sake of simplicity, suppose that the motion takes place along one dimension. 
Then the drag force \cite{batchelor_introduction_2000,falkovich_fluid_2011} is given by
\begin{equation}
  \bar{R}
  =\frac{1}{2} \rho CA \dot{q}^2 \dd q,
\end{equation}
where $C$ is a dimensionless constant depending on the body shape, $\rho$ is the mass density of the fluid and $A$ is the area of the projection of the object on a plane perpendicular to the direction of motion. For further simplification, suppose that the density is uniform, and then $k=CA\rho/2$ is constant.
The dissipation function is thus
\begin{equation}
  \mathcal{R}=\frac{k}{3}\dot{q}^3.
\end{equation}
If the body is not subject to forces besides the drag, its Lagrangian is $L=m\dot{q}^2/2$. Consider the vector field $X=e^{kq/m }\partial/\partial q$. It verifies that $X^c(L)=X^v(\mathcal{R})$, so $X^v(L)=me^{kq/m}\dot{q}$ is a constant of the motion. In particular, when $k\to 0$, $X$ is the generator of translations and the conservation of momentum is recovered.
\end{example}

\begin{proposition}[Point-like symmetries in Rayleigh systems]
Let $X$ be a vector field on $Q$.
Then the following results hold:
\begin{enumerate}
\item If $\liedv{X^c}\alpha_L$ is closed, then $X$ is a Lie symmetry of $(L,\bar{R})$ if and only if 
\begin{equation}
  \dd(X^c(E_L))=-S^* (\dd (X^c \mathcal R)).
\end{equation}

\item If $\liedv{X^c}{\alpha_L}=\dd f$ for some function $f$ on $TQ$, then the following statements are equivalent:
\begin{enumerate}
\item $X$ is a Noether symmetry,
\item $X^c(E_L)+X^v(\mathcal{R})=0$,
\item $f-X^v(L)$ is a constant of the motion.
\end{enumerate}

\item If $X$ is Noether symmetry, it is also a Lie symmetry if and only if $\contr{X^c}\dd \bar{R}=0$.
\item If $X$ is Noether symmetry, it is also a symmetry of the forced Lagrangian if and only if $\liedv{X^c}\alpha_L=0$.
\end{enumerate}
\end{proposition}


\begin{proposition}[Not-point-like symmetries in Rayleigh systems]
Let $\tilde X$ be a vector field on $TQ$. Then the following results hold:
\begin{enumerate}
\item If $\liedv{\tilde{X}}\alpha_L$ is closed, then $\tilde{X}$ is a dynamical symmetry if and only if
\begin{equation}
  \dd(\tilde{X}(E_L)+(S\tilde{X})(\mathcal{R}))=-\contr{\tilde{X}} \dd \bar{R}.
\end{equation}
\item If $\liedv{\tilde{X}}{\alpha_L}=\dd f$, then the following statements are equivalent:
 \begin{enumerate}
 \item  $\tilde{X}$ is a Cartan symmetry,
 \item $\tilde{X}(E_L)+(S\tilde{X})(\mathcal{R})=0$,
 \item $f-(S\tilde{X})(L)$ is a conserved quantity.
 \end{enumerate}
 \item If $\tilde X$ is Cartan symmetry, it is also a dynamical symmetry if and only if $\contr{\tilde X}\dd \bar{R}=0$.
 \end{enumerate}
\end{proposition}

The following examples of Rayleigh systems are presented in reference \cite{minguzzi_rayleighs_2015}. Here their symmetries and constants of the motion are obtained.

\begin{example}[A rotating disk]
Consider a disk of mass $m$ and radius $r$ placed on a horizontal surface. Let $\varphi$ be the angle of rotation of the disk with respect to a reference axis. The Lagrangian of the disk is $L=T=mr^{2}\dot{\varphi}^{2}/4$ and its Rayleigh dissipation function is $\mathcal{R}=\mu mgr\dot\varphi/2$. The Poincaré-Cartan 1-form is $\alpha_L=mr^2 \dot \varphi/2\  \dd \varphi$. The Rayleigh force is $\bar{R}=\mu mgr/2\ \dd \varphi$.

 Consider the vector field $\tilde{X}=-r\dot{\varphi}\partial/\partial \varphi+\mu g \partial/\partial \dot{\varphi}$. Clearly, $\tilde{X}(E_L)=\tilde{X}(L)=-(S\tilde{X})(\mathcal{R})$. In addition,
 \begin{equation}
   \liedv{\tilde{X}}\alpha_L=
    \frac{\mu mgr^2}{2} \dd \varphi-
   \frac{mr^3}{2}\dot \varphi \dd \dot \varphi=\dd f,
 \end{equation}
 where
 \begin{equation}
   f=\frac{\mu mgr^2}{2}\varphi -\frac{mr^3}{4}\dot \varphi^2
 \end{equation}
 modulo a constant, and $(S\tilde{X})(L)=-mr^{3} \dot{\varphi}^{2}/2$,
  so
  \begin{equation}
    f-(S\tilde X) (L)=\frac{\mu mgr^2}{2}\varphi +\frac{mr^3}{4}\dot \varphi^2
  \end{equation}
  is a constant of the motion. Since $\bar{R}$ is closed, $\contr{\tilde{X}}\dd \bar{R}=0$ is trivially satisfied, so $\tilde{X}$ is a dynamical symmetry as well as a Cartan symmetry. 

  However, since $\bar{R}$ is closed, it is not strictly an external force. In fact, the Lagrangian
  \begin{equation}
    \tilde{L}=L+\frac{\mu mgr}{2} \varphi
  \end{equation}
  leads to the same equations of motion as $L$ with the external force $\bar{R}$.
\end{example}

\begin{example}[The rotating stone polisher]
Consider a system formed by two concentric rings of the same mass $m$ and radius $r$, which are placed over a rough surface, and rotate in opposite directions. Let $(x,y)$ be the position of the centre and $\theta$ the orientation of the machine. Let $\omega$ be the angular velocity of the rings. The Rayleigh dissipation function is given by
\begin{equation}
  \mathcal{R}=2\mu mg r\omega +\frac{\mu mg}{2r\omega}(\dot{x}^{2}+\dot{y}^{2}), 
\end{equation}
and the Lagrangian is $L=T=m(\dot{x}^{2}+\dot{y}^{2} +r^{2}\dot{\theta}^{2}+r^{2}\omega^{2})$. The Poincaré-Cartan 1-form is $\alpha_L=2m(\dot x\dd  x+ \dot y \dd  y +r^2 \dot \theta \dd  \theta)$
, and the Rayleigh force is
\begin{equation}
  \bar{R}=\frac{\mu mg}{r\omega}(\dot x\dd x + \dot y \dd y).
\end{equation}
 Let $\tilde{X}^{(1)}=-2r\omega \partial/\partial x +\mu g \partial/ \partial \dot{x}$ and $\tilde{X}^{(2)}=-2r\omega \partial/\partial y +\mu g  \partial/ \partial \dot{y}$. One can check that $\tilde{X}^{(i)}(E_L)=\tilde{X}^{(i)}(L)=-(S\tilde{X}^{(i)})(\mathcal{R})$ (for $i=1,2$). Moreover, $\liedv{\tilde{X}^{(i)}}\alpha_L=\dd f_i$ for $f_1=2\mu mgx$ and $f_2=2\mu mgy$, along with $(S\tilde{X}^{(1)})(L)=-4mr \dot x$ and $(S\tilde{X}^{(2)})(L)=-4mr \dot y$, so $2mr\omega \dot{x}+\mu mgx$ and $2mr\omega \dot{y}+\mu mgy$ are constants of the motion.
\end{example}
\subsection{Reduction of Rayleigh systems with symmetries}
Consider a Rayleigh system $(L,\mathcal{R})$ on $TQ$, and a Lie group $G$ that acts freely and properly on $TQ$ and leaves $L$ invariant. Recalling the results from Section \ref{section_reduction_forced}, one can define a Lie subalgebra $\mathfrak{g}_{\mathcal{R}}\equiv \mathfrak{g}_{\bar{R}}$ of $\mathfrak{g}$ by
\begin{equation}
     \mathfrak{g}_\mathcal{R} = \left\{\xi \in \mathfrak{g} \mid
    \xi_Q^v(\mathcal{R}) = 0, S^* \dd \xi_Q^c(\mathcal{R}) = 0
    \right\},
\end{equation}
where the condition $S^* \dd \xi_Q^c(\mathcal{R}) = 0$ means that $\xi_Q^c(\mathcal{R})$ is basic: it does not depend on $\dot{q}^i$. For each $\xi\in \mathfrak{g}_{\mathcal{R}}$:
\begin{enumerate}
\item $J^\xi$ is a constant of the motion,
\item $\xi$ leaves  $\bar{R}$ invariant.
\end{enumerate}
If additionally $\mathcal{R}$ is $\mathfrak{g}_\mathcal{R}$-invariant, that is, $\xi_Q^c(\mathcal{R})=0$ for each $\xi\in \mathfrak{g}_\mathcal{R}$, then it induces a dissipation function $\mathcal{R}_\mu:(TQ)_\mu\to \RR$ given by
\begin{equation}
  \mathcal{R}_\mu \circ \pi_\mu = \mathcal{R} \circ \incl_\mu.
\end{equation}

\begin{example}
Consider the Rayleigh system $(L,\mathcal{R})$ on $TQ$. 
Suppose that $Q=\RR^n\setminus \{0\}$, and that the Lagrangian $L$ on $T Q$ is spherically symmetric, say $L(q,\dot{q})= L(\lVert q \rVert,\lVert{\dot{q}}\rVert)$. 
Consider the Lie group $G=\mathrm{SO}(n)$ acting by rotations on $Q$ (see Example \ref{example_angular_momentum_reduction}). Recall that
\begin{equation}
	J(q,\dot{q}) = q \times \dot{q}.
\end{equation}
Now look for Rayleigh potentials $\mathcal{R}$ such that $\mathfrak{g}_\mathcal{R} = \mathfrak{g}$. The condition that $\xi_Q^v (\mathcal{R})=0$ implies that $\mathcal{R}$ must be spherically symmetric on the velocities. Then, the condition that $\xi_Q^c(\mathcal{R})$ is basic means that the terms which are not spherically symmetric on the positions cannot involve the velocities, that is,
\begin{equation}
    \mathcal{R} = A(q) + B(\lVert{q}\rVert, \lVert \dot{q} \rVert).
\end{equation}
Without loss of generality (see Remark \ref{remark_gauge_potential}), the Rayleigh potential can be taken as
\begin{equation}
    \mathcal{R} =  B(\lVert{q}\rVert, \lVert \dot{q} \rVert). \label{Rayleigh_rotation_invariant}
\end{equation}
In particular, if $\mathcal{R}$ is a quadratic form, say
\begin{equation}
	\mathcal{R} = \frac{1}{2} R_{ij} (q)\dot{q}^i \dot{q}^j,
\end{equation}
then requirement \eqref{Rayleigh_rotation_invariant} leads to
\begin{equation}
	R_{ij} = r_i\left( \lVert q \rVert  \right) \delta_{ij}.
\end{equation}
\end{example}

\chapter{Conclusions and future work}\label{conclusions}
In this master's thesis, the field of geometric mechanics has been explored, in particular considering autonomous Lagrangian and Hamiltonian systems. The first part has presented the geometric framework of mechanical systems. Several results regarding symmetries, constants of the motion and reduction have been reviewed. The second part has generalized these results for systems subjected to external forces. The main original contributions are the following:
\begin{enumerate}
\item Noether's theorem has been extended for forced Lagrangian systems. Different types of symmetries, and their associated constants of the motion, have also been generalized for forced Lagrangian systems. Their Hamiltonian counterparts have been obtained as well.
\item A theory for the reduction of forced Lagrangian systems invariant under the action of a group of symmetries has been presented.
\item Rayleigh forces have been expressed in a geometric language, which has allowed to obtain particular results regarding the symmetries and reduction of Rayleigh systems.
\item Several examples have been presented.
\end{enumerate}
 These results have been published in an article \cite{de_leon_symmetries_2021}.

 There are several lines of research to be pursued during the PhD thesis consequent to this master's thesis. Some of them are the following:
 \begin{enumerate}
 \item \emph{Hamilton-Jacobi theory}. As it was explained in Section \ref{section_Hamiltonian}, given a Hamiltonian $H$ on $T^*Q$, by performing a canonical transformation one can obtained a new Hamiltonian $K$ which satisfies Hamilton's equations as well. In particular, choosing the transformation so that $K$ vanishes identically makes Hamilton's equations tautological. The relation between $H$ and $K$ leads to the so-called Hamilton-Jacobi equation (see references \cite{goldstein_mecanica_1987,abraham_foundations_2008})
 \begin{equation}
	H\left(q^{i}, \frac{\partial W}{\partial q^{i}}\right)=E. \label{H-J_equation}
\end{equation}
The Hamilton-Jacobi problem consists in finding a solution $W$, known as the characteristic function. This is the formulation of classical mechanics which resembles the most to quantum mechanics. As a matter of fact, Hamilton-Jacobi-type equations are obtained when considering the classical limit of Schrödinger equation \cite{carinena_06,abraham_foundations_2008,esposito_marmo_sudarshan_2004}.
Geometrically, the Hamilton-Jacobi equation can be written as \cite{carinena_06,iglesias-ponte_towards_2008}
\begin{equation}
	(\dd W)^* H = E,
\end{equation}
where $\dd W$ is a section of $T^*Q$. The main interest of the Hamilton-Jacobi theory lies in finding complete solutions, that is, solutions which depend on $n=\dim Q$ parameters. Each value of the parameters provides a trajectory of the system \cite{carinena_06}. Moreover, a complete solution provides $n$ constants of the motion and conversely. The Hamilton-Jacobi problem can be extended for forced Hamiltonian systems \cite{iglesias-ponte_towards_2008} in a natural manner. However, it has not been studied in detail in the previous literature, for instance complete solutions have not been characterized, and it will be covered in a future paper \cite{H-J_forced_21}.
\item \emph{Discrete mechanics}. The numerical analysis of a mechanical system requires to discretize it. The naive way to do this is by first obtaining the continuous equations of motion of the system, and then discretizing them, usually by means of a Runge-Kutta method (see references \cite{newman_computational_2013,scherer_equations_2017}). However, this approach does not guarantee the constants of the motion of the continuous system to be preserved on its discrete counterpart. Instead, one can approximate the action integral (for instance, by the trapezoidal method \cite{newman_computational_2013}) and, by taking variations on the discretized action, obtain the discrete Euler-Lagrange equations. A similar approach can be carried out for forced Lagrangian systems (see references \cite{lew_variational_2004,marsden_west_01}). 
In the previous literature \cite{ohsawa_discrete_2011,de_leon_geometry_2018}, a Hamilton-Jacobi theory for discrete systems has been developed. The extension of this theory for forced mechanical systems will be covered in a future paper \cite{H-J_forced_21}.
\item \emph{Higher-order systems}. Throughout this thesis, the Lagrangians considered only depended on the positions and the velocities. Of course, one can also consider Lagrangians depending on the accelerations and subsequent derivatives, which are called higher-order Lagrangians. Geometrically, a Lagrangian $L(q,\dot{q}, \ddot q)$ is a function on $T^2Q=TTQ$, and analogously for higher orders. Higher-order Hamiltonians can be described in a similar manner \cite{prieto-martinez_lagrangian-hamiltonian_2011,Rashid:1993xk}. Symmetries and constants of the motion for higher-order Lagrangian systems have been studied in references \cite{leon_classification_1994,de_leon_symmetries_1995,scomparin_nonlocal_2021}. Systems subject to higher-order external forces (such as the Abraham-Lorentz force) will be considered in future papers.
\item \emph{Non-autonomous systems}. Lagrangians and Hamiltonians can also depend explicitly on time, in which case they are called non-autonomous \cite{prieto-martinez_lagrangian-hamiltonian_2011,sardanashvily_geometric_2013,de_Leon_2017}. Some results regarding non-autonomous forced Lagrangian systems were already obtained in reference \cite{cantrijn_vector_1982}. A complete classification of the symmetries or the reduction of these systems has not been studied so far.
\end{enumerate}

\let\emph\emphoriginal

\printbibliography

\end{document}